%% file: main.tex
\documentclass[acmsmall,screen]{acmart}
\setcopyright{rightsretained}
\acmPrice{}
\acmDOI{10.1145/3591251}
\acmYear{2023}
\copyrightyear{2023}
\acmSubmissionID{pldi23main-p179-p}
\acmJournal{PACMPL}
\acmVolume{7}
\acmNumber{PLDI}
\acmArticle{137}
\acmMonth{6}
\received{2022-11-10}
\received[accepted]{2023-03-31}

\startPage{1}


  %% \keywords are mandatory in final camera-ready submission

\input{packages}

\input{defs}

\input{pseudocode_macros}

\renewcommand{\smallskip}{}
\excludeversion{pldi} 
\includeversion{arxiv}

\begin{document}

%% Title information
\title{Optimal Reads-From Consistency Checking for C11-Style Memory Models}         %% [Short Title] is optional;
                                        %% when present, will be used in
                                        %% header instead of Full Title

% %% Author with single affiliation.
% \author{Author}
% \affiliation{
%   \institution{institution}            %% \institution is required
%   \country{country}                    %% \country is recommended
% }
% \author{Author}
% \affiliation{
%   \institution{institution}            %% \institution is required
%   \country{country}                    %% \country is recommended
% }
\author[H. C. Tun\c{c}]{H\"{u}nkar Can Tun\c{c}}
\orcid{0000-0001-9125-8506} 
\affiliation{
	\institution{Aarhus University}            %% \institution is required
	\country{Denmark}                    %% \country is recommended
}
\email{tunc@cs.au.dk}          %% \email is recommended

\author[P. A. Abdulla]{Parosh Aziz Abdulla}
\orcid{0000-0001-6832-6611}
\affiliation{
  \institution{Uppsala University}            %% \institution is required
  \country{Sweden}                    %% \country is recommended
}
\email{parosh@it.uu.se}          %% \email is recommended

\author[S. Chakraborty]{Soham Chakraborty}
\orcid{0000-0002-4454-2050}
\affiliation{
  \institution{TU Delft}            %% \institution is required
  \country{Netherlands}                    %% \country is recommended
}
\email{s.s.chakraborty@tudelft.nl} 

\author[K. Shankaranarayanan]{Shankaranarayanan Krishna}
\orcid{0000-0003-0925-398X}
\affiliation{
  \institution{IIT Bombay}            %% \institution is required
  \country{India}                    %% \country is recommended
}
\email{krishnas@cse.iitb.ac.in}

\author[U. Mathur]{Umang Mathur}
\orcid{0000-0002-7610-0660} 
\affiliation{
  \institution{National University of Singapore}            %% \institution is required
  \country{Singapore}                    %% \country is recommended
}
\email{umathur@comp.nus.edu.sg}          %% \email is recommended

\author[A. Pavlogiannis]{Andreas Pavlogiannis}
\orcid{0000-0002-8943-0722}
\affiliation{
  \institution{Aarhus University}            %% \institution is required
  \country{Denmark}                    %% \country is recommended
}
\email{pavlogiannis@cs.au.dk}          %% \email is recommended

%%% Author with two affiliations and emails.
%\author{First2 Last2}
%\authornote{with author2 note}          %% \authornote is optional;
%                                        %% can be repeated if necessary
%\orcid{nnnn-nnnn-nnnn-nnnn}             %% \orcid is optional
%\affiliation{
%  \position{Position2a}
%  \department{Department2a}             %% \department is recommended
%  \institution{Institution2a}           %% \institution is required
%  \streetaddress{Street2a Address2a}
%  \city{City2a}
%  \state{State2a}
%  \postcode{Post-Code2a}
%  \country{Country2a}                   %% \country is recommended
%}
%\email{first2.last2@inst2a.com}         %% \email is recommended
%\affiliation{
%  \position{Position2b}
%  \department{Department2b}             %% \department is recommended
%  \institution{Institution2b}           %% \institution is required
%  \streetaddress{Street3b Address2b}
%  \city{City2b}
%  \state{State2b}
%  \postcode{Post-Code2b}
%  \country{Country2b}                   %% \country is recommended
%}
%\email{first2.last2@inst2b.org}         %% \email is recommended

\input{abstract}

\maketitle

\input{introduction}

\input{models-new}

\input{helpers}
\input{rf}

\input{experiments-new}

\input{conclusion}

\begin{acks}
Andreas Pavlogiannis was partially supported by a research grant (VIL42117) from VILLUM FONDEN.
Umang Mathur was partially supported by the Simons Institute for the Theory of Computing,
and by a Singapore Ministry of Education (MoE) Academic Research Fund (AcRF) Tier 1 grant.
Shankaranarayanan Krishna was partially supported by the SERB MATRICS grant MTR/2019/000095. 
Parosh Aziz Abdulla was partially supported by the Swedish Research Council. 
\end{acks}

\input{artifact}

\bibliography{references}

\begin{arxiv}
\newpage

\appendix

\input{app_rf}

\newpage

\input{app_experiments}

\end{arxiv}

\end{document}

%% file: packages.tex
\usepackage{tikz}
\usepackage{xspace}

\usepackage{booktabs}   %% For formal tables:
\usepackage{subcaption} %% For complex figures with subfigures/subcaptions
\usepackage{enumitem} 
\usepackage{makecell} 

\usepackage{wrapfig}

\usepackage{amsmath}
\usepackage{pifont}
\usepackage{amsthm}
\usepackage{thm-restate}
\usetikzlibrary{arrows,automata,shapes,decorations,decorations.markings,calc, matrix,decorations.pathmorphing, patterns,backgrounds,shapes.misc,arrows.meta,positioning}

\usepackage{xspace}
\usepackage[ruled,vlined,resetcount,linesnumbered,noend]{algorithm2e}
\usepackage{parskip}
\usepackage{graphicx}
\usepackage{comment}
\usepackage{versions}

\usepackage{arydshln}
 \usepackage{multicol}
 \usepackage{multirow}

\usepackage{hyperref}
\usepackage[capitalise]{cleveref}
\usepackage{varwidth}
\usepackage{mathbbol}

\usepackage{array}
\newcolumntype{H}{>{\setbox0=\hbox\bgroup}c<{\egroup}@{}}

%% file: defs.tex
\newcommand*\Let[2]{\State #1 $\gets$ #2}

\newcounter{sarrow}

\newtheorem{claim}{Claim}
\theoremstyle{remark}

\newcommand{\bparagraph}[1]{\smallskip \noindent{\bf {#1}.}}
\newcommand{\eparagraph}[1]{\smallskip \noindent{\emph{#1}.}}

\colorlet{colorPO}{darkgray!80!black}
\colorlet{colorRF}{blue}
\colorlet{colorCO}{red!80!black}
\colorlet{colorMO}{red!80!black}
\colorlet{colorFR}{purple}
\colorlet{colorECO}{orange}
\colorlet{colorCOM}{magenta}
\colorlet{colorSW}{teal}
\colorlet{colorHB}{green!40!black}
\colorlet{colorPPO}{magenta}
\colorlet{colorRSEQ}{green!40!black}
\colorlet{colorSC}{violet}
\colorlet{colorHBMO}{brown}
\colorlet{colorPSC}{violet}
\colorlet{colorREL}{olive}
\colorlet{colorWB}{olive}
\colorlet{colorRMW}{brown}

\tikzset{
   every path/.style={>=stealth},
   po/.style={->,draw=colorPO,shorten >=-0.5mm,shorten <=-0.5mm},
   ppo/.style={->,draw=colorPPO,shorten >=-0.5mm,shorten <=-0.5mm},
   sw/.style={->,draw=colorSW,shorten >=-0.5mm,shorten <=-0.5mm},
   rf/.style={->,draw=colorRF,dashed,shorten >=-0.5mm,shorten <=-0.5mm},
   mrf/.style={->,draw=colorRF,dashed,shorten >=-0.5mm,shorten <=-0.5mm},   
   urf/.style={->,draw=colorRF,shorten >=-0.5mm,shorten <=-0.5mm},
   fr/.style={->,draw=colorFR,dashed,shorten >=-0.5mm,shorten <=-0.5mm},
   hb/.style={->,draw=colorHB,thick,shorten >=-0.5mm,shorten <=-0.5mm},
   co/.style={->,draw=colorCO,dotted,very thick,shorten >=-0.5mm,shorten <=-0.5mm},
   mo/.style={->, draw=colorMO,dotted,very thick,shorten >=-0.5mm,shorten <=-0.5mm},
   rmw/.style={->,draw=colorRMW,thick,shorten >=-0.5mm,shorten <=-0.5mm},
   rseq/.style={->,draw=colorRSEQ,thick,dotted,shorten >=-0.5mm,shorten <=-0.5mm},
   com/.style={->,draw=colorCOM,thick,shorten >=-0.5mm,shorten <=-0.5mm},
}

\newcommand\inarr[1]{\begin{array}{@{}l@{}}#1\end{array}}
\newcommand\inarrC[1]{\begin{array}{@{}c@{}}#1\end{array}}

\newcommand\inparII[2]{\begin{array}{@{}l@{~~}||@{~~}l@{}}\inarr{#1}&\inarr{#2}\end{array}}

\newcommand{\set}[1]{\{#1\}}
\newcommand{\setpred}[2]{\{#1 \,|\, #2\}}
\newcommand{\tup}[1]{\langle#1\rangle}

\newcommand{\imm}{\mathsf{imm}}
\renewcommand{\emptyset}{\varnothing}
\newcommand{\nats}{\mathbb{N}}

\newcommand{\R}{\mathsf{R}}
\newcommand{\W}{\mathsf{W}}
\newcommand{\E}{\mathsf{E}}
\newcommand{\F}{\mathsf{F}}
\newcommand{\Upd}{\mathsf{RMW}}

\newcommand{\RRMW} {\R\cup\Upd}
\newcommand{\WRMW} {\W\cup\Upd}
\newcommand{\scmm}{\mathsf{SC}}
\newcommand{\wramm}{\mathsf{WRA}}
\newcommand{\ramm}{\mathsf{RA}}
\newcommand{\rlxmm}{\mathsf{Relaxed}}
\newcommand{\sramm}{\mathsf{SRA}}
\newcommand{\rcmm}{\mathsf{RC20}}

\newcommand{\mmorder}{\preccurlyeq}

\newcommand{\MOrlx}{\textsc{rlx}}
\newcommand{\MOacq}{\textsc{acq}}

\newcommand{\MOacqrel}{\textsc{acq-rel}}
\newcommand{\MOrel}{\textsc{rel}}
\newcommand{\MOsc}{\textsc{sc}}

\newcommand{\po}{\mathsf{\color{colorPO}po}}

\newcommand{\rf}{\mathsf{\color{colorRF}rf}}
\newcommand{\rfinv}{\rf^{-1}}

\newcommand{\mo}{\mathsf{\color{colorMO}mo}}
\newcommand{\mopartial}{\overline{\mo}}

\newcommand{\fr}{\mathsf{\color{colorFR}fr}}

\newcommand{\hb}{\mathsf{\color{colorHB}hb}}
\newcommand{\hbloc}{\mathsf{\color{colorHB}hbloc}}

\newcommand{\hbmo}{\mathsf{\color{colorHBMO}hbmo}}
\newcommand{\wb}{\mathsf{\color{colorWB}wb}}
\newcommand{\sw}{\mathsf{\color{colorSW}sw}}

\newcommand\locx[1]{{{#1}_{x}}}

\newcommand{\irr}{\mathsf{irr}}
\newcommand{\acy}{\mathsf{acy}}

\newcommand{\op}{\mathsf{op}}
\newcommand{\tid}{\mathsf{tid}}
\newcommand{\event}{e}
\newcommand{\id}{\mathsf{id}}
\newcommand{\llab}{\mathsf{lab}}
\newcommand{\lloc}{\mathsf{loc}}
\newcommand{\ord}{\mathsf{ord}}

\newcommand{\bottom}{\perp}
\newcommand{\rd}{\mathtt{r}}
\newcommand{\wt}{\mathtt{w}}
\newcommand{\ud}{\mathtt{rmw}}
\newcommand{\RMW}{\ud}

\newcommand{\ex}{\mathsf{\color{black}X}}
\newcommand{\partialtext}{\mathsf{partial}}
\newcommand{\expartial}{\overline{\ex}}

\newcommand{\WComplexity}[1]{\mathcal{W}[#1]}
\newcommand{\NP}{\mathcal{NP}}

\newcommand{\NumEvents}{n}

\newcommand{\NumThreads}{k}

\newcommand{\Path}{\rightsquigarrow}

\newcommand{\MemModel}{\mathcal{M}}

\newcommand{\hunkar}[1]{\textcolor{green!45!black}{Hünkar: #1}}

\newcommand{\celeventester}{\textsf{C11Tester}\xspace}
\newcommand{\Trust}{\textsf{TruSt}\xspace}
\newcommand{\bname}[1]{\textsf{#1}}

\newlist{compactenum}{enumerate}{3} % 3 is max-depth
\setlist[compactenum]{label=(\arabic*), nosep,leftmargin=*}
\crefname{compactenumi}{Item}{Items}

\newlist{compactitem}{itemize}{3} % 3 is max-depth
\setlist[compactitem]{label=\textbullet, nosep,leftmargin=*}
\crefname{compactitemi}{Item}{Items}

\newlist{compactdesc}{description}{3} % 3 is max-depth
\setlist[compactdesc]{nosep,leftmargin=1em}
\crefname{compactdesci}{Item}{Items}

\newcommand{\maxspeedup}{$99.3$}
\newcommand{\avgspeedup}{$13.9$}

%% file: pseudocode_macros.tex
\SetKwProg{myfun}{function}{}{}
\SetKwProg{myhandler}{handler}{}{}
\SetKwFunction{init}{Initialization}
\SetKwFunction{getLW}{getLastWriteBeforeOffline}
\SetKwFunction{poprop}{poPropagate}
\SetKwFunction{joinchangedindex}{pairwiseMaxIndexes}
\SetKwFunction{getLWB}{getLastWriteBefore}
\SetKwFunction{getLRB}{getLastReadBefore}
\SetKwFunction{getLRWB}{getLastReadWriteBefore}
\SetKwFunction{initWtLst}{initWtLst}
\SetKwFunction{initRdLst}{initRdLst}
\SetKwFunction{initWtRdLst}{initWtRdLst}
\SetKwFunction{get}{get}
\SetKwFunction{getHBTS}{getHBTimestamps}
\SetKwFunction{getTCPC}{getTopAndPositionInRFChain}
\SetKwFunction{checkRace}{checkRace}
\SetKwFunction{rdhandler}{readHandler}
\SetKwFunction{wthandler}{writeHandler}
\SetKwFunction{acqhandler}{acquire}
\SetKwFunction{relhandler}{release}
\SetKwInput{Input}{Input}
\SetKwInOut{Output}{Output}
\SetKw{Let}{let}
\SetKw{Break}{break}
\SetKw{NOT}{not}
\SetKw{declare}{declare}
\SetKw{exit}{exit}
\SetKw{Continue}{continue}
\SetKwFor{Foreach}{for each}{}{}%
\DontPrintSemicolon

\SetCommentSty{mycommfont}
\SetNoFillComment

\newcommand{\LW}{\mathsf{LW}}
\newcommand{\TC}{\mathbb{TC}}
\newcommand{\PC}{\mathbb{PC}}
\newcommand{\HB}{\mathbb{HB}}
\newcommand{\Hh}{\mathbb{H}}

\newcommand{\LWB}{\mathsf{LWB}}
\newcommand{\LRB}{\mathsf{LRB}}
\newcommand{\glw}{\textsf{lastWriteBefore}}
\newcommand{\glr}{\textsf{lastReadBefore}}
\newcommand{\HBMO}{\mathsf{HBMO}}
\newcommand{\PHBMO}{\mathsf{PredHBMO}}

\newcommand{\RdLst}{\mathbb{RList}}
\newcommand{\WtLst}{\mathbb{WList}}

\newcommand{\mx}{\sqcup}

\newcommand{\view}[3]{{#1}_{#2}^{#3}}

\newcommand{\nil}{\mathtt{NIL}}
\makeatletter
\newcommand*\bigcdot{\mathpalette\bigcdot@{.5}}
\newcommand*\bigcdot@[2]{\mathbin{\vcenter{\hbox{\scalebox{#2}{$\m@th#1\bullet$}}}}}
\makeatother
\newcommand{\apicall}[3]{{#1}\bigcdot{\textup{\texttt{#2(}#3\texttt{)}}}}

%% file: abstract.tex
%!TEX root = ./main.tex
\begin{abstract}
% \Andreas{We need a couple of sentences about C/C++ concurrency, the C11 mem model and its variants.}
% \scomment{Possible opening: 
% In concurrency world \emph{consistency testing} is one of the most fundamental problem. 
% Given a memory model $\MemModel$ and a (partial) execution $\ex$, consistency testing asks: \emph{is $\ex$ consistent with $\MemModel$?}
% }
%in the hope that the right formalization can be used to precisely analyze 
%real world 
Over the years, several memory models have been proposed to capture the subtle concurrency semantics of C/C++.
%, with the goal to rigorously analyze concurrent software written in these languages.
One of the most fundamental problems associated with a memory model $\MemModel$ is \emph{consistency checking}:~given an execution $\ex$, is $\ex$ consistent with $\MemModel$?
This problem lies at the heart of numerous applications, including specification testing and litmus tests, stateless model checking, and dynamic analyses.
As such, it has been explored extensively and its complexity is well-understood for 
traditional models like SC and TSO.
%as well as popular hardware-level memory models 
%such as TSO.
However, less is known for the numerous model variants of C/C++,
for which the problem becomes challenging due to the intricacies of their concurrency primitives.
%~\cite{cstandard,cppstandard}, 
%possibly due to the abundance of their variants and 
%Consistency checking is challenging for these models due to the intricate nature of concurrency primitives offered by them.
% release-acquire-style semantics.
% \scomment{Rephrasing (I think beging specific to release-acquire is not required here): 
% possibly due to the abundance of their variants and the  intricate semantics of the concurrency primitives. 
% }
% 
In this work we study the problem of consistency checking for popular variants of the C11 memory model, in particular, the $\rcmm$ model, its release-acquire ($\ramm$) fragment, 
the strong and weak variants of $\ramm$ ($\sramm$ and $\wramm$),
as well as the $\rlxmm$ fragment of $\rcmm$.

Motivated by applications in testing and model checking,
we focus on \emph{reads-from} consistency checking.
The input is an execution $\ex$ specifying a set of events, 
their program order and their reads-from relation,
and the task is to decide the existence of a modification order on the writes of $\ex$
that makes $\ex$ consistent in a memory model.
%$\MemModel$ of interest.
% We draw a rich complexity landscape that varies depending on the model 
% at hand and the exact form of the consistency checking problem 
% \ucomment{since we are adding this, shall we also mention somewhere that if co is known, the consistency checking is obviously linear time. Alternatively, we may want to precisely state rf consistency in the abstract}.
We draw a rich complexity landscape for this problem;
our results include 
(i)~nearly-linear-time algorithms for certain variants, which improve over prior results,
(ii)~fine-grained optimality results, as well as 
(iii)~matching upper and lower bounds ($\NP$-hardness) for other variants.
% , including optimality results for the tractable variants.
To our knowledge, this is the first work to characterize the complexity of consistency checking 
for C11 memory models.
We have implemented our algorithms inside the \Trust model checker and the \celeventester testing tool.
Experiments on standard benchmarks show that our new algorithms improve consistency checking, often by a significant margin.
\end{abstract}

%\begin{abstract}
%In recent years, there have been increasingly many research efforts
%to formalize the semantics of the C11 memory model.
%One of the central goals behind such efforts is to establish
%the correct semantics that explains all behaviors, rules out the scrupulous ones,
%and, more importantly, can be used for reason about the correctness of
%programs running under this memory model.
%We take an algorithmic and complexity-theoretic
%view towards this goal.
%We consider the \emph{consistency checking} problem --- 
%\emph{given a memory model $\MemModel$
%and an execution $\ex$ with partial information,
%can $\ex$ be extended to a complete execution that is consistent with $\MemModel$?}
%Consistency checking is the core algorithmic problem involved in concurrency-centric program analyses
%such as stateless model checking and in dynamic analysis techniques for
%detecting concurrency bugs like data races.
%In this paper, we study the consistency problem for the C11 memory model and its
%popular fragments such as RA, SRA, WRA.
%We study \ucomment{\ldots}
%\end{abstract}

%% file: introduction.tex
%!TEX root=./main.tex

\section{Introduction}\label{sec:intro}

% - Shared memory concurrency is dominant paradigm.
% - Weak memory is useful but error-prone.
% - Testing and verification is required.
% - Consistency checking is a core problem in weak memory testing and verification

%Modern architectures follow weak memory concurrency models which exhibit additional program behavior than sequential consistency. 
Modern programming languages such as C/C++ \cite{cstandard,cppstandard} have introduced first-class platform-independent concurrency primitives to gain performance from weak memory architectures. 
The programming model is popularly known as C11 \cite{Boehm:2008,Batty:2011}. 
The formal semantics of C11 has been an active area of research \cite{Batty:2011,Vafeiadis:2015,Batty:2016,Lahav:2017,Chakraborty:2019,Lee:2020,Margalit:2021} and 
other programming languages such as Java \cite{Bender:2019} and Rust \cite{Dang:2019} have also adopted similar concurrency primitives.   

One of the most fundamental computational problems associated with a memory model, particularly in testing and verification, is that of \emph{consistency checking}~\cite{Gibbons:1997,Furbach:2015,Kokologiannakis:2023}.
Here, one focuses on a fixed memory model $\MemModel$, typically described using constraints or axioms.
The input to the consistency problem pertaining to the memory model $\MemModel$ is then a partial execution $\ex$, typically described using a set of events $\E$ together with a set of relations on $\E$. 
The consistency problem then asks to determine if $\ex$ is \emph{consistent} with $\MemModel$.
Here, by \emph{partial} execution, we mean that the set of relations is not fully described in the input, in which case the problem asks whether $\ex$ can be extended to a complete execution that is consistent with $\MemModel$. The focus of this paper is \emph{reads-from} consistency checking; in the rest of the paper we refer to this simply as consistency checking. 

The problem of consistency checking has numerous applications in both software and hardware verification.
First, viewing memory models as contracts between the system designer and the software developers,
consistency checking is a common approach to testing memory subsystems, cache-coherence protocols and compiler transformations against the desired contract~\cite{Qadeer2003,Manovit2006,Chen2009,Wickerson2017,Windsor:STVR22}.
Second, since public documentations of memory architectures are typically not entirely formal, litmus tests can reveal or dismiss behaviors that are not covered in the documentation~\cite{Alglave2010,Alglave2011,Alglave:2014}.
Consistency checking for litmus tests makes testing more efficient (and thus also more scalable), by avoiding the enumeration of behaviors that are impossible under the given model.
%Second, since public documentations of memory architectures are typically not entirely formal, consistency checking on litmus tests can reveal or dismiss behaviors that are not covered in the documentation~\cite{Alglave2010,Alglave2011,Alglave:2014}.
%\hunkar{
%The relevance of this second item is still not very clear to me.
%Reviewer-B has also asked about this and our rebuttal seems to just point out that enumerating rf instead %of mo is more efficient.
%However, it seems to me that we are also talking about this point on the third item as well.
%}
Third, in the area of model checking, (partial) executions typically serve the role of abstraction mechanisms.
Consistency checking, thus, ensures that model checkers indeed explore valid system behavior.
As such, it has been instrumental in guiding recent research in partial-order reduction techniques and stateless model checking of concurrent software~\cite{Kokologiannakis:2022,Chalupa2017,Abdulla:2019b,Abdulla:2018,Bui:2021,Chatterjee:2019,Agarwal:2021,NorrisDemsky:2013}.
Focusing on partial executions allows such algorithms to consider coarser equivalences such as 
the \emph{reads-from} equivalence,
allowing for more proactive state-space reductions and better performance as a result.
These advances have also propelled the use of formal methods in the industry~\cite{VSync2021,loom,AmazonFM2021}.
Consistency checking of partial executions also forms the foundation of dynamic \emph{predictive analyses} by characterizing the space of perturbations that can be applied to an observed execution in an attempt to expose a
bug~\cite{Kalhauge2018,Huang14,Pavlogiannis2019,Mathur:lics2020,Luo:2021,Kini2017, MathurSHB2018,SyncP2021}. 

%\hunkar{We could also cite~\cite{Windsor:STVR22}}

The ubiquitous relevance of consistency checking has led to a systematic study of its computational complexity under various memory models.
Under sequential consistency (SC), most variants of the problem were shown to be $\NP$-hard in the seminal work
of Gibbons and Korach~\cite{Gibbons:1997}.
Subsequently, more fine-grained investigations have characterized 
how input parameters such as the number of threads, memory locations, write accesses and communication topology
affect the complexity of consistency checking~\cite{Abdulla:2019b,Chini2020,Mathur:lics2020,Agarwal:2021}.
As the consistency problems remain intractable under most common weak memory models (such as SPARC/X86-TSO, PSO, RMO, PRAM)~\cite{Furbach:2015}, 
parametric results have also been established for these models~\cite{Bui:2021,Chini2020}.
Given its applications in analysis of concurrent programs,
clever heuristics have been proposed to enhance the efficiency
of checking consistency in practice~\cite{Zennou2019}.

The C11 memory model provides the flexibility to derive 
different weak memory model paradigms based on different 
subsets and combinations of the concurrency primitives, their memory orders, and their respective semantics. 
For instance, the release and acquire memory orders allow programmers to 
derive release-acquire ($\ramm$) as well as its weak ($\wramm$) and strong ($\sramm$) variants~\cite{Lahav:2022}.
% , RC20 \cite{Lahav:2019} and so on. 
The $\ramm$ model is weaker than SC and provides a rigorous foundation in defining 
synchronization and locking primitives~\cite{Lahav:2016}. 
The $\wramm$ and $\sramm$ are equivalent to variants of causal consistency~\cite{Lahav:2022}, 
a well studied consistency model in the distributed systems literature. 
% While these models provide the desired flexibility for writing weak memory programs, due to the subtle semantics, programming in these weak memory models is a challenging task and often error prone. In this scenario, testing and verification play a crucial role in developing correct and efficient weak memory programs.
% \scomment{We should write about Relaxed model as well.}
C11 also provides `relaxed' memory access modes which constitutes the weaker memory model fragment $\rlxmm$. 
Relaxed memory accesses can reorder freely and are the most performant compared to 
accesses with stronger memory orders. 
In our work, we focus on the recently proposed declarative $\rcmm$ memory model~\cite{Margalit:2021}
capturing a rich fragment of C11, consisting of release, acquire and relaxed memory accesses as well as
memory fence operations.
This memory model is a natural fragment of the C11 model, given that
``\emph{only a few (practical) algorithms that actually employ SC accesses and become wrong when release/acquire accesses are used instead}''~\cite{Margalit:2021}.
Further, focusing on the non-$\scmm$ fragment allows us to reap the benefits
of polynomial time consistency checking, which otherwise quickly becomes intractable~\cite{Gibbons:1997}.
%and a significant bottleneck in verification tasks (such as stateless model checking) relying on a consistency checking oracle.

% These semantic models in the C11 framework define a new class of consistency problems.
The intricacies of C11 and the abundance of its variants give rise to a plethora of consistency-checking instances.
Some first results show that consistency checking for $\ramm$ admits a polynomial bound~\cite{Lahav:2015,Abdulla:2018}, a stark difference to $\scmm$ for which this problem is $\NP$-hard and is not even well-parameterizable~\cite{Gibbons:1997,Mathur:lics2020}.
These positive results have facilitated efficient model checking and testing techniques~\cite{Abdulla:2018,Kokologiannakis:2019,Luo:2021}.
However, beyond these recent developments, little is known about the complexity of consistency checking for C11-style memories, and, to our knowledge, the setting remains poorly understood.
Our work fills this gap.

% \ucomment{Cite Dartagnan~\cite{Dartagnan2021} (model checker that extracts execution graphs and check each of them using SMT)}

%\input{figures/new-RA-example.tex}

\input{figures/po-rf-mo.tex}

\input{contributions}

%% file: figures/po-rf-mo.tex
%!TEX root = ../main.tex

\begin{figure}
\centering
\def\ystep{0.4}
\begin{subfigure}{0.16\textwidth}
\vspace{0.2cm}
\[
\inparII{
x:=1;
\\ y:=1;
\\ x:=2;
}{
\ \ a:=y;
\\ \ \ x:=3;
\\ \ \ b:=x;
}
\]
\vspace{0.5cm}
\caption{Program}
\label{subfig:prg:mp}
\end{subfigure}
\hfill
\begin{subfigure}{0.26\textwidth}
\centering
\scalebox{0.9}{
\begin{tikzpicture}[yscale=1]
  \node (t11) at (-1.2,0*\ystep)  {$\wt(x)$};
  \node (t12) at (-1.2,-2*\ystep) {$\wt(y)$};
  \node (t13) at (-1.2,-4*\ystep) {$\wt(x)$};
  \node (t21) at (1,0*\ystep) {$\rd(y)$};
  \node (t22) at (1,-2*\ystep) {$\wt(x)$};
  \node (t23) at (1,-4*\ystep) {$\rd(x)$};

  \draw[po] (t11) to (t12);
  \draw[po] (t12) to (t13);
  \draw[po] (t21) to (t22);
  \draw[po] (t22) to (t23);
  \draw[rf,bend left=0] (t12) to node[above,pos=0.9, sloped]{$\rf$} (t21);
  \draw[rf,bend right=0] (t13) to node[above,pos=0.9, sloped]{$\rf$} (t23);
%    
%  \draw[mo,bend right=15] (t22) to node[below,pos=0.5, sloped]{$\mo$} (t13);
\end{tikzpicture} 
}  
\caption{
$\ramm$-consistent execution.
}
\label{subfig:ra:partial}
\end{subfigure}
\hfill
\begin{subfigure}{0.24\textwidth}
\centering
\scalebox{0.9}{
\begin{tikzpicture}[yscale=1]
  \node (t11) at (-1.2,0*\ystep)  {$\wt(x)$};
  \node (t12) at (-1.2,-2*\ystep) {$\wt(y)$};
  \node (t13) at (-1.2,-4*\ystep) {$\wt(x)$};
  \node (t21) at (1,0*\ystep) {$\rd(y)$};
  \node (t22) at (1,-2*\ystep) {$\wt(x)$};
    \node (t23) at (1,-4*\ystep) {$\rd(x)$};

  \draw[po] (t11) to (t12);
  \draw[po] (t12) to (t13);
  \draw[po] (t21) to (t22);
  \draw[po] (t22) to (t23);
  \draw[rf,bend left=0] (t12) to node[above,pos=0.9, sloped]{$\rf$} (t21);
  \draw[rf,bend right=0] (t13) to node[above,pos=0.9, sloped]{$\rf$} (t23);
  \draw[mo,bend right=15] (t22) to node[below,pos=0.5, sloped]{$\mo$} (t13);
\end{tikzpicture} 
} 
\caption{
$\ramm$-consistent $\mo$.
}
\label{subfig:ra:consistent}
\end{subfigure}
\hfill
\begin{subfigure}{0.28\textwidth}
\centering
\scalebox{0.9}{
\begin{tikzpicture}[yscale=1]
  \node (t11) at (-1.2,0*\ystep)  {$\wt(x)$};
  \node (t12) at (-1.2,-2*\ystep) {$\wt(y)$};
  \node (t13) at (-1.2,-4*\ystep) {$\wt(x)$};
  \node (t21) at (1,0*\ystep) {$\rd(y)$};
  \node (t22) at (1,-2*\ystep) {$\wt(x)$};
    \node (t23) at (1,-4*\ystep) {$\rd(x)$};

  \draw[po] (t11) to (t12);
  \draw[po] (t12) to (t13);
  \draw[po] (t21) to (t22);
  \draw[po] (t22) to (t23);
  \draw[rf,bend left=0] (t12) to node[above,pos=0.9, sloped]{$\rf$} (t21);
  \draw[rf,bend left=0] (t11) to node[below,pos=0.7, sloped]{$\rf$} (t23.west);
%    
  %\draw[mo,bend right=0] (t22) to node[below,pos=0.3, sloped]{$\mo$} (t11);
\end{tikzpicture} 
}
\caption{
$\ramm$-inconsistent execution.
}
\label{subfig:ra:inconsistent}
\end{subfigure}
\caption{
A program (\subref{subfig:prg:mp}) and a partial execution $\expartial$ specifying the writer $\rf^{-1}(\rd)$ of each read $\rd$ (\subref{subfig:ra:partial}).
$\expartial$ is $\ramm$-consistent, as witnessed by the modification order $\mo$ that abides to $\ramm$ semantics (\subref{subfig:ra:consistent}).
The partial execution in (\subref{subfig:ra:inconsistent}) is $\ramm$-inconsistent, as there is no modification order $\mo$ that abides to $\ramm$ semantics. 
}
\label{fig:nwra}
\end{figure}

%% file: contributions.tex
%!TEX root=./main.tex

\bparagraph{{Our contributions}}\label{subsec:contributions}
In this paper we study the reads-from consistency checking for the $\ramm$, $\sramm$, $\wramm$, $\rlxmm$ and $\rcmm$ memory models, with results that are optimal or nearly-optimal.
In all cases, the input is a partial execution $\expartial=\tup{\E, \po, \rf}$ 
with $\NumEvents=|\E|$ events and $\NumThreads$ threads, where $\po$ and $\rf$ are the program order and reads-from relation, respectively (see \cref{subsec:executions}), and the task is to determine if there is a modification order $\mo$
such that the extension $\ex = \tup{\E, \po, \rf, \mo}$ is consistent with the memory model in consideration. 
\cref{fig:nwra} illustrates an example for $\ramm$-consistency.

Our first result concerns $\rcmm$.
%\Andreas{Perhaps we should do more justice and mention that the cubic bound is for full $\rcmm$}
Consistency checking is known to be in polynomial time~\cite{Lahav:2015,Abdulla:2018,Luo:2021},
though of degree $3$ (i.e., $O(\NumEvents^3)$).
This cubic complexity has been identified as a challenge for efficient model checking (e.g.,~\cite{Kokologiannakis:2019,Kokologiannakis:2022}).
Here we show that the full $\rcmm$ model admits an algorithm that is nearly linear-time; i.e., a bound that becomes linear when the number of threads is bounded.

\begin{restatable}{theorem}{thmrcrfupper}
\label{thm:rc_rf_upper}
Consistency checking for $\rcmm$ can be solved in $O(\NumEvents \cdot \NumThreads)$ time.
\end{restatable}

A key step towards \cref{thm:rc_rf_upper} is our notion of \emph{minimal coherence}, 
which is a novel characterization that serves as a witness of consistency.
Our consistency-checking algorithm proves consistency by constructing a minimally coherent (partial) modification order.
Although similar witnesses have been used in the past (e.g., the writes-before order~\cite{Lahav:2015}, saturated traces~\cite{Abdulla:2018}, or C11Tester's framework~\cite{Luo:2021}),
the simplicity of minimal coherence allows, for the first time to our knowledge, for a nearly linear-time algorithm.

%\Andreas{Remove this}
%To obtain \cref{thm:rc_rf_upper}, we also develop an efficient algorithm for computing the happens-before relation in $\rcmm$.
%Although polynomial-time algorithms are folklore, our nearly linear-time bound is non-trivial.

Next we turn our attention to $\sramm$.
Perhaps surprisingly, although the model is conceptually close to $\ramm$, it turns out that checking consistency for $\sramm$ is intractable.

\begin{restatable}{theorem}{thmsrarflower}
\label{thm:sra_rf_lower}
Consistency checking for $\sramm$ is $\NP$-complete, and $\WComplexity{1}$-hard in the parameter $\NumThreads$.
\end{restatable}

Here $\WComplexity{1}$ is a parameterized complexity class~\cite{Chen2004FPT}.
This result states that, not only is the problem $\NP$-complete, but it is also unlikely to be \emph{fixed parameter tractable} in $\NumThreads$, i.e., solvable in time $O(\NumEvents^c \cdot f(\NumThreads))$, where $c >0$ and $f$ are independent of $\NumEvents$.
Nevertheless, our next result shows that this problem admits an upper bound that 
is polynomial when $\NumThreads=O(1)$.
Given the $\WComplexity{1}$-hardness, the next result is thus optimal, in the sense that $k$ has to appear in the exponent of $n$.

% The usual approach to alleviate intractability is to look into the problem from a parameterized point of view, 
% and see whether hardness vanishes when certain parameters of the input are small.
% The most common parameterization in concurrency is with respect to the number of threads.
% Our next result shows that indeed the problem admits a polynomial bound when $\NumThreads=O(1)$.

\begin{restatable}{theorem}{thmsrarfupper}
\label{thm:sra_rf_upper}
Consistency checking for $\sramm$ can be solved in $O(\NumThreads\cdot \NumEvents^{\NumThreads+1})$ time.
\end{restatable}

Taking a closer look into the model, we identify RMWs as the source of intractability.
Indeed, the RMW-free fragment of $\sramm$ admits a nearly linear bound, much like $\rcmm$.
This fragment is relevant, as it coincides with the causal convergence model~\cite{bouajjani:2017}.

\begin{restatable}{theorem}{thmsrarfnormwupper}
\label{thm:sra_rf_normw_upper}
Consistency checking for the RMW-free fragment of $\sramm$ can be solved in $O(\NumEvents \cdot \NumThreads)$ time.
\end{restatable}

Next, we show that the problem can be solved just as efficiently for $\wramm$.

\begin{restatable}{theorem}{thmwrarfupper}
\label{thm:wra_rf_upper}
Consistency checking for $\wramm$ can be solved in $O(\NumEvents \cdot \NumThreads)$ time.
\end{restatable}

Turning our attention to the $\rlxmm$ fragment of $\rcmm$, we show that the problem admits a truly linear bound (i.e., regardless of the number of threads).

\begin{restatable}{theorem}{thmrlxrfupper}
\label{thm:rlx_rf_upper}
Consistency checking for $\rlxmm$ can be solved in $O(\NumEvents)$ time.
\end{restatable}

Finally, observe that, in contrast to \cref{thm:rlx_rf_upper}, the bounds in \cref{thm:rc_rf_upper}, \cref{thm:sra_rf_normw_upper} and \cref{thm:wra_rf_upper} can become super-linear in the presence of many threads.
It is thus tempting to search for a truly linear-time algorithm for any of $\ramm$, $\wramm$ and (RMW-free) $\sramm$.
Unfortunately, our final result shows that this is unlikely, in all models.

\begin{restatable}{theorem}{thmrflower}
\label{thm:rf_lower}
There is no consistency-checking algorithm for the RMW-free fragments of any of
$\ramm$, $\wramm$, and $\sramm$ that runs in time
$O(\NumEvents^{\omega/2 - \epsilon})$, for any fixed $\epsilon > 0$.
Moreover, there is no combinatorial algorithm for the problem that runs in time $O(\NumEvents^{3/2-\epsilon})$, under the combinatorial BMM hypothesis.
\end{restatable}

Here $\omega$ is the matrix multiplication exponent, with currently $\omega\simeq2.37$.
\cref{thm:rf_lower} states that a truly linear-time algorithm for any of these models would bring matrix multiplication in $n^{2+o(1)}$ time, a major breakthrough.
Focusing on combinatorial algorithms~(i.e., excluding algebraic fast-matrix multiplication, which appears natural in our setting), consistency checking for any of these models requires at least $n^{3/2}$ time unless (boolean) matrix multiplication (BMM) is improved below the classic cubic bound (which is considered unlikely, aka the BMM hypothesis~\cite{Williams19}).

Due to space restrictions, we relegate all proofs to\begin{arxiv}~\cref{sec:app_rf}\end{arxiv}\begin{pldi}~to our technical report~\cite{arxiv}\end{pldi}.

\bparagraph{Experiments}
We have implemented our algorithms inside the \Trust model checker and the \celeventester testing tool,
and evaluated their performance on consistency checking for benchmarks utilizing instructions in the $\ramm$ model.
Our results report consistent and often significant speedups that reach $162\times$ for \Trust and $104.2\times$ for \celeventester.

%\hunkar{
Overall, our efficiency results enable practitioners to perform model checking and testing for $\rcmm$, RMW-free $\sramm$, $\wramm$, and $\rlxmm$ more efficiently, and apply these techniques to larger systems.
On the other hand, our hardness result for $\sramm$ indicates that, akin to $\scmm$,  
performing consistency checking for $\sramm$ efficiently requires developing 
practically oriented heuristics that work well in the common cases.
Finally, our super-linear lower bound for all models except $\rlxmm$ indicates that further improvements over our $O(\NumEvents\cdot \NumThreads)$ bounds will likely be highly non-trivial.
%}

%% file: models-new.tex
%!TEX root=./main.tex
\section{Axiomatic Concurrency Semantics}\label{sec:models}

In this section we introduce the C/C++ concurrency semantics we consider in this work, along with  the RC20 model and its variants~\cite{Margalit:2021,Lahav:2022}.
%As such, the exposition here is close to those works.

\bparagraph{Syntax} 
C/C++ defines a shared memory concurrency model using different kinds of concurrency primitives. 
In addition to plain (or non-atomic) load and store accesses, C/C++ provides atomic accesses for load, store, atomic read-modify-write (RMW -- such as atomic increment), and fence operations.
We only consider atomic accesses here.
An atomic access is parameterized by a memory mode, among relaxed ($\MOrlx$), acquire ($\MOacq$), release ($\MOrel$), acquire-release ($\MOacqrel$), and sequential-consistency ($\MOsc$). 
The memory order for a read, write, RMW, and fence access is one of $\set{\MOrlx,\MOacq, \MOsc}$, 
$\set{\MOrlx, \MOrel, \MOsc}$, $\set{\MOrlx, \MOacq, \MOrel, \MOacqrel, \MOsc}$, 
and $\set{\MOacq, \MOrel, \MOacqrel, \MOsc}$, respectively.
These accesses result in different types of events during execution. 
In this paper we consider the models which are based on non-$\scmm$ primitives. 
Nevertheless, $\rcmm$ defines $\scmm$ fences using the release-acquire primitives \cite{Lahav:2022}.   
The memory modes are partially ordered on increasing strength of synchronization according to the lattice $\MOrlx \sqsubset \set{\MOacq,\MOrel} \sqsubset \MOacqrel$. 
An access is \emph{acquire} (\emph{release}) if its order is $\MOacq$ ($\MOrel$), or stronger. %$\MOacqrel$.
%Based on these syntactic primitives, a number of concurrency models have been proposed. 
% These models are usually defined based on either axiomatic or operational semantics.
% Here we discuss various axiomatic models
%We follow the axiomatic models, which are better suited for the consistency-checking problem we tackle in this work.

%\ucomment{Mention SC and then say RC20 does not talk about it.} 
%\scomment{addressed}

\input{executions-new}

\input{consistency_axioms-new}
\input{consistency_models-new}

%% file: executions-new.tex
%!TEX root = ./main.tex

\subsection{Executions}\label{subsec:executions}

The axiomatic concurrency models are defined with respect to the executions they allow.
Hence, a program can be represented as a set of executions.
In turn, an execution is defined by a set of events that are generated from shared memory accesses or fences, and relations between these events. 

\bparagraph{Events} 
An event is a tuple $\tup{\id,\tid,\llab}$ where $\id$, $\tid$, $\llab$ denote a unique identifier, thread identifier, and label of the event. 
The label $\llab=\tup{\op, \ord, \lloc}$ is a tuple where $\op$ denotes the corresponding memory access or fence operation and $\ord$ denotes the 
memory mode. 
For memory accesses, $\lloc$ denotes its memory location, while in the case of fences, we have $\lloc = \bottom$.
% Memory access events may further contain  the values they read/write, but 
For the purpose of the reads-from consistency problem we consider in this paper, we 
omit the \emph{values} read or written in memory access events.
% In other settings, memory-access events are also labeled by the values they read/write.
% As these values are not important for the reads-from consistency problem we study in this paper, we will ignore such values to simplify notation.
%Given an event $\event$, we let $\event.\id$, $\event.\tid$, $\event.\llab$, $\event.\op$, $\event.\ord$, $\event.\lloc$ be the identifier, thread identifier, label, operation, memory order and memory location, respectively.
% For an event $\event$, we use notations such as $\event = \wt^o(t, x)$ 
% to mean $\event.\op = \wt$, $\event.\ord = o$, $\event.\tid = t$, and $\event.\lloc = x$. 
% %\ucomment{When do we use this notation $\wt_o$?}
% % \scomment{This was changed. Removing inconsistent notation.}
% When the event's memory order is clear from its context then 
We write $\wt(t, x)$/$\rd(t,x)$/$\ud(t,x)$ to denote a write/read/read-modify-write event in thread $t$ on location $x$,
and simply write $\event(x)$ to denote an event for which the thread is implied or not relevant.
As a matter of convention, we omit mentioning the memory order throughout the paper, as it will either be clear from the context or not relevant.
%Here and in the rest of the paper, we have omitted explicit mention of the memory order
%\hunkar{Here, why do we only mention how we denote write events and not mentioned how we denote other types of events (e.g., read)? }
%\ucomment{Mention here the different memory orders and that release events are W,RMW,F, acquire events are R,RMW,F and release-acquire are RMW,F}.

%\noindent\textit{Accesses to Events.} 
The set of read, write, atomic update, and fence events are $\R$, $\W$, $\Upd$ and $\F$ respectively, and are generated 
from the executions of load, store, atomic load store, fence accesses respectively. 
%generate an event in $\R$, $\W$, $\F$ respectively. 
As we only deal with executions (as opposed to program source), we use $\Upd$ to denote a successful read-modify-write operation.
Failed read-modify-write operations simply result in read accesses.
%A successful RMW access generates an $\U$ event and a failed RMW generates an $\R$ event.
We refine the set of events in various ways. 
For instance, $\E^{\sqsupseteq \MOrel}$ denotes the set of 
events with memory order that is at least as strong as $\MOrel$. 
For a set of events $\E$, we write $\E.locs$, $\locx{\E}$, and $\E.tids$ to denote the set of distinct locations accessed by events in $\E$, the subset of events in $\E$ that access memory location $x$, and the different threads participating in $\E$. 
% For a set of events $\E$, we use $\locx{\E}$ to denote the subset of events in $\E$
% that access memory location $x$. 
%The memory locations are initialized at the start of the execution, represented by a set of write events. 
%\ucomment{What is the significance of init events being non-atomic?}
% non-atomic init doesn't race. Removing to avoid confusion.

\bparagraph{Notation on relations}
Consider a binary relation $S$ over a set of events $\E$.
%we will use $S(x,y)$ to denote the predicate $(x, y) \in S$.
The reflexive, transitive, reflexive-transitive closures, and inverse relations of $S$
are denoted as $S^?$, $S^+$, $S^*$, $S^{-1}$ respectively. 
The relation $S$ is acyclic if $S^+$ is irreflexive. 
We write $\irr(S)$ and $\acy(S)$ to denote that  relation $S$ irreflexive and acyclic respectively. 
%Relation $S_\imm$ denotes the immediate relation:  
%$S_\imm(x,y) \triangleq S(x,y) \land \nexists z \ S(x,z) \land S(z,y)$. %\ucomment{$S(x,z) \land S(z, y)$}. 
Given two relations $S_1$ and $S_2$, we denote their composition by $S_1;S_2$. 
$[A]$ denotes the identity relation on a set $A$, i.e. 
$[A](x,y) \triangleq x = y \land x \in A$.
For a given memory location $x$, we let $S_x=[\locx{\E}];S;[\locx{\E}]$ be the restriction of $S$ to all events of $\E$ on $x$.

\bparagraph{Candidate executions and relations}
An execution (also referred to as candidate execution~\cite{Batty:2011} or execution graph~\cite{Lahav:2017}) 
is a tuple $\ex = \tup{\E, \po, \rf, \mo}$  where $\ex.\E$ is a set of events and $\ex.\po$, $\ex.\rf$, $\ex.\mo$ are binary relations over $\ex.\E$. 
In particular, the \emph{program order} $\po$ is a partial order that enforces a total order on events of the same thread. 
The \emph{reads-from relation} $\rf\subseteq (\WRMW) \times (\RRMW)$ associates write/RMW events $\event_1$ to read/RMW events $\event_2$ reading from $\event_1$.
Naturally, we require that $\event_1.\lloc=\event_2.\lloc$, and that $\rfinv$ is a function, i.e., every read/RMW event has a unique writer.
The \emph{modification order} $\mo\subseteq (\WRMW)\times (\WRMW)$ is the union of modification orders $\mo_x$, where each $\mo_x$ is a total order over $(\locx{\W}\cup \locx{\Upd})$.

We frequently also use some derived relations (see \cref{fig:ax:models}).
The \emph{from-read relation} $\fr\subseteq (\RRMW)\times (\WRMW)$ relates every read or RMW event to all the write or RMW events that are $\mo$-after its own writer.
% \begin{align*}
% \fr \triangleq (\rf^{-1};\mo)\setminus [\id]\hfill \tag{from-read}
% \end{align*}
The \emph{synchronizes-with relation} $\sw\subseteq \E^{\sqsupseteq \MOrel}\times \E^{\sqsupseteq \MOacq}$ relates release and acquire events, for instance, when an acquire read event reads from a release write. 
Fence instructions combined with relaxed accesses also participate in $\sw$. 
%The precise definition, taken from~\cite{Margalit:2021} (also used in other prior works, such as the \celeventester~\cite{Luo:2021}) is:
% \begin{align*}
% \sw \triangleq & \ [\E_{\sqsupseteq \MOrel}];([\F];\po)^?;\rf^+;(\po;[\F])^?;[\E_{\sqsupseteq \MOacq]}] \hfill \tag{synchronizes-with}
% \end{align*}
The \emph{happens-before relation} $\hb$ is the transitive closure of $\po$ and $\sw$.
We also project $\hb$ to individual locations, denoted as $\hbloc$.

\input{figures/ax_models}

%% file: figures/ax_models.tex
\begin{figure}
\scalebox{0.9}{
\begin{tikzpicture}
   \node[draw=black, dashed,thick,rounded corners=.3cm, text width=7cm,align=left, inner sep=1pt] at (0, 0) {
    \small
    \begin{align*}
    \begin{array}{lll}
        \fr&{\triangleq}& (\rf^{-1};\mo)\setminus [\id] \\
        \sw&{\triangleq}& [\E^{\sqsupseteq \MOrel}];\!([\F];\po)^?;\rf^+;\!(\po;[\F])^?;\![\E^{\sqsupseteq \MOacq}] \\
        \hb&{\triangleq}&(\po\cup \sw)^+ \quad\quad\quad \hbloc\quad{\triangleq}\quad\bigcup_x \hb_x\\
        % \hbloc\!\!\!\!&{\triangleq}&\bigcup_x \hb_x
    \end{array}
    \end{align*}
    };
   \node[draw=black, dashed,thick,rounded corners=.3cm, text width=7cm,align=left, inner sep=1pt] at (0,-2.4) {
    \small
    \begin{align*}
    \begin{array}{lr}
        \irr(\mo;\rf^?;\hb^?) & \text{(Wcoh)}\\
        \acy(\hb \cup \mo) & \text{(strong-Wcoh)} \\
        \irr(\fr;\rf^?;\hb) & \text{(Rcoh)} \\
        \irr(\hbloc;[\WRMW];\hb;\rf^{-1})\!\! & \text{(weak-Rcoh)} \\
        \irr(\fr;\mo) & \text{(atomicity)} \\
        \forall u_1,\! u_2\!\in\! \Upd, \rf(u_1)\!\!\neq \! \rf(u_2) & \text{(weak-atomicity)} \\
        \acy(\po \cup \rf) & \text{(PO-RF)}
    \end{array}
    \end{align*}
    };
    \node[draw=black, dashed,thick,rounded corners=.3cm, text width=8.2cm,align=left, inner sep=1pt,minimum height=4.5cm] at (7.8,-1.6) {
    \small
    \renewcommand{\arraystretch}{1.5}
    \begin{align*}
    \begin{array}{lr}
        \acy(\po \cup \rf \cup \mo \cup \fr)& (\scmm) \\
        \text{(PO-RF)} \land \text{(Wcoh)} \land \text{(Rcoh)} \land \text{(atomicity)} & (\ramm)^\dagger \\
        \text{(PO-RF)} \land \text{(strong-Wcoh)} \land \text{(Rcoh)} \land \text{(atomicity)} & (\sramm) \\
        \text{(PO-RF)} \land \text{(weak-Rcoh)} \land \text{(weak-atomicity)} & (\wramm) \\
        \text{(PO-RF)} \land \text{(Wcoh)} \land \text{(Rcoh)} \land \text{(atomicity)}&(\rcmm) \\
        \text{(PO-RF)} \land \text{(Wcoh)} \land \text{(Rcoh)} \land \text{(atomicity)} & (\rlxmm)^\dagger
    \end{array}
    \end{align*}
    %\;\;$^\dagger$ ($\rlxmm$) only consists of $\MOrlx$ accesses.
    \;\;$^\dagger$ These are fragments of $\rcmm$.
    };
\end{tikzpicture}
}
    \caption{Relations, Axioms, and Consistency Models on C11 Concurrency.}
    \label{fig:ax:models}
\end{figure}

%% file: consistency_axioms-new.tex
%!TEX root=./main.tex

\subsection{Consistency Axioms}\label{subsec:consistency_axioms}

Consistency axioms characterize different aspects or constraints of a memory model.
We broadly classify these constraints as 
coherence, atomicity, global ordering, and causality cycles.  
Different interpretations of these constraints give rise to different consistency models as shown in \cref{fig:ax:models}.

\bparagraph{Coherence}
In an execution, \emph{coherence} enforces `SC-per-location': the memory accesses per memory locations are totally ordered. 
% We may also further categorize coherence as write-coherence and read-coherence. 
\emph{Write-coherence} enforces that $\mo$ agrees with $\hb$. 
A stronger variant is \emph{strong-write-coherence}, which requires that this condition holds transitively.
\emph{Read coherence} enforces that a read $\rd(x)$ cannot read from a write $\wt(x)$ if 
there is an \emph{intermediate} write $\wt'(x)$ that happens-before $\rd(x)$, 
i.e. $\hb(\wt'(x),\rd(x))$ holds.
In the vanilla read-coherence, the notion of `intermediate' relates to $\mo$,
i.e., we have $(\wt(x),\wt'(x))\in \mo$,
while in \emph{weak-read-coherence}~\cite{Lahav:2022}, 
`intermediate' relates to $\hb$, i.e., we have $(\wt(x),\wt'(x))\in\hb$.
% Depending on the precise notion of ``intermediate'', we obtain two variants~\cite{Lahav:2022}:~\emph{read-coherence}, where intermediate is wrt $\mo$ (i.e., we have $(\wt'(x),\wt(x))\in \mo$), and
% \emph{weak-read-coherence}, where intermediate is wrt $\hb$ (i.e., we have $(\wt'(x),\wt(x))\in\hb$).
%We list the formal definitions of the coherence constraints in \cref{fig:ax:models}. 
% \begin{compactitem}
% %\item $(\poloc \cup \rf \cup \mo \cup \fr)$ is acyclic. \hfill (SC-per-location)
% \item $(\mo;\rf^?;\hb^?)$ is irreflexive.  \hfill (write-coherence) 
% \item $(\hb \cup \mo)$ is acyclic. \hfill (strong-write-coherence) 
% \item $(\fr;\rf^?;\hb)$ is irreflexive. \hfill (read-coherence)
% \item $(\hbloc;[\W \cup \U];\hb;\rf^{-1})$ is irreflexive. \hfill (weak-read-coherence)
% \end{compactitem}
%Going forward we define memory models that vary based on their respective coherence axioms. 

\bparagraph{Atomicity}
The property ensures that (successful) RMW accesses indeed update the memory locations atomically. 
Following~\cite{Lahav:2022}, we consider two variants. 
\emph{Atomicity} ensures that no intermediate write event 
on the same location takes place between an RMW and its writer.
%For example, suppose $u$ is an update event and $w$ is a write on the same location that violates atomicity. 
%In that case $\fr(u,w)$ as well as $\mo(w,u)$ hold. Axiom (atomicity) restricts it. 
\emph{Weak-atomicity} simply prohibits two RMW events to have the same writer.
% \begin{compactitem}
% \item $\fr;\mo$ is irreflexive. \hfill (atomicity)
% \item $\forall \ud_1, \ud_2\in \U, \rf(\ud_1)\neq \rf(\ud_2)$. \hfill (weak-atomicity)
% %\item $\forall (w_1, \event_1), (w_2, \event_2) \in \rf;[\U].~ w_1 = w_2 \implies \event_1 = \event_2$ \hfill (weak-atomicity)
% \end{compactitem}

%\bparagraph{Global Ordering}
%Coherence and atomicity enforce ordering constraints on same-memory location accesses. 
%In the concurrency models we consider here, additional orderings are imposed
%between fence operations or memory accesses on different memory locations, primarily due to $\hb$.
%The following axiom, integral to many memory models, enforces that an event cannot happen-before itself.
% Using $\hb$ relation we ensure that in an execution an event cannot happen-before itself by using the following axiom.
%\begin{compactitem}
%\item $\hb$ is irreflexive. \hfill (irrHB) 
%\end{compactitem} 
% In addition, certain C/C++ concurrency models use relation 
% $\SC$ to obtain a global ordering as $\SC$ is a total order on all $\MOsc$ accesses. \ucomment{remove?}

\bparagraph{Causality cycles} 
A causality cycle arises in the presence of primitives that have weaker behaviors than release-acquire accesses. 
Such a cycle consists of $\po$ and $\rf$ orderings and may result in 
`out-of-thin-air' behavior in certain programs. 
To avoid such `out-of-thin-air' behavior, many consistency models 
explicitly disallow such cycles~\cite{Lahav:2022,Luo:2021}.
% \begin{compactitem}
% \item $(\po \cup \rf)$ is acyclic. \hfill (PO-RF) 
% \end{compactitem}
In the absence of $\MOrlx$ accesses, the PO-RF axiom coincides with the requirement for $\hb$ acyclicity.

%% file: consistency_models-new.tex
%!TEX root = ./main.tex

\subsection{Axiomatic Consistency Models and Consistency Checking} 

\input{figures/model_examples}

We can now present the memory models we consider in this work.
See \cref{fig:ax:models} for a summary.
%For the sake of completeness, we also state Sequential Consistency as a standard reference model.

\bparagraph{Sequential consistency}
%While we do not consider Sequential Consistency (SC) memory model in this work, we define it for the sake of completeness.
Sequential consistency ($\scmm$) enforces a global order on all memory accesses. 
This is a stronger constraint than coherence, which orders same-location memory accesses. 
In addition, $\scmm$ also guarantees atomicity. 
%The precise axiomatic specification of SC is as follows. 
% \begin{compactitem}
% \item $\ (\po \cup \rf \cup \mo \cup \fr)$ is acyclic. \hfill ($\scmm$)
% \end{compactitem}

\bparagraph{Release-Acquire and variants}
The release-acquire ($\ramm$) memory model is weaker than SC, 
and is arguably the most well-understood fragment of C11.
At the same time, $\ramm$ enables high-performance implementations of fundamental concurrency algorithms~\cite{Lahav:2016,desnoyers:11}.
Broadly, under the release-acquire semantics (including other related variants),
each read-from ordering establishes a synchronization. 
In this case $\hb$ reduces to $\hb \triangleq(\po \cup \rf)^+$.
Following~\cite{Lahav:2022}, we consider three variants of 
release-acquire models: $\ramm$, strong $\ramm$ ($\sramm$), and weak $\ramm$ ($\wramm$). 
These models coincide with standard variants of causal consistency~\cite{Burckhardt:2014, bouajjani:2017, Lloyd:2011}, which is a ubiquitous consistency model relevant also in other domains such as distributed systems.

$\sramm$ enforces a stronger coherence guarantee (namely strong-write-coherence)
on write accesses compared to $\ramm$. 
$\wramm$ does not place any restrictions on the $\mo$ ordering between same-location writes. 
Instead, the only orderings considered between same-location writes are through the $[\W];\hbloc;[\W]$ relation. 
Thus, $\wramm$ provides weaker constraints for coherence and atomicity.

% First, we consider the basic $\ramm$ memory model.
% % RA enforces coherence by (write-coherence) and (read-coherence), atomicity property by (atomicity) axiom, and global ordering by (irrHB).
% \begin{compactitem}
% \item (write-coherence) $\land$ (read-coherence) $\land$ (atomicity) \hfill ($\ramm$)
% \end{compactitem}
% Next, we consider $\sramm$, which enforces a stronger coherence guarantee
% on write accesses compared to RA and ensures strong-write-coherence. 
% Formally,
% \begin{compactitem}
% \item (strong-write-coherence) $\land$ (read-coherence) $\land$ (atomicity) \hfill ($\sramm$)
% \end{compactitem}
% Finally, we consider WRA, which does not place any restrictions on the $\mo$ ordering between same-location writes. 
% Instead, the only orderings considered between same-location writes are through the $[\W];\hbloc;[\W]$ relation. 
% As a result, WRA provides weaker constraints for coherence and atomicity by (weak-read-coherence) and (weak-atomicity) respectively:
% \begin{compactitem}
% \item (weak-read-coherence) $\land$ (weak-atomicity) $\land$ (PO-RF) \hfill ($\wramm$)
% \end{compactitem}

\bparagraph{RC20}
The recently introduced $\rcmm$ model~\cite{Margalit:2021} defines a rich fragment of the C11 model
consisting of acquire/release and relaxed accesses.
Despite the absence of $\scmm$ accesses, $\rcmm$ can express many practical synchronization algorithms, and can simulate $\scmm$ fences.

% defined for the non-$\SC$ fragment in C/C++ concurrency that handles the semantics of relaxed accesses along with release-acquire primitives~\cite{Margalit:2021}.
% It can nevertheless simulate $\scmm$ fences using release, acquire, and RMW primitives.
%RC20 provides coherence and atomicity guarantees by (write-coherence), (read-coherence), and (atomicity) axioms. 
%In addition, RC20 enforce (PO-RF) to avoid out-of-thin-air 
%behavior. 
% \begin{compactitem}
% \item (write-coherence) $\land$ (read-coherence) $\land$ (atomicity) $\land$ (PO-RF) \hfill ($\rcmm$)
% \end{compactitem}
%We also analyze WeakRC20 model where we drop the (PO-RF) axiom from RC20.

\bparagraph{Relaxed}
Finally, the relaxed fragment of $\rcmm$ contains only $\MOrlx$ accesses, resulting in $\hb = \po$.
% \begin{compactitem}
% \item (write-coherence) $\land$ (read-coherence) $\land$ (atomicity) $\land$ (PO-RF) \hfill ($\rlxmm$)
% \end{compactitem}

\bparagraph{Comparison between memory models}
The above models can be partially ordered according to the behaviors (executions) they allow as $\scmm \mmorder \sramm \mmorder \ramm \mmorder \set{\wramm, \set{  \rcmm  \mmorder \rlxmm}}$, with models towards the right
allowing more behaviors.
%according to their strength as follows:
% \[
% \inarr{
% \scmm \sqsubset \sramm \sqsubset \ramm \sqsubset \set{\wramm, \set{  \rcmm  \sqsubset \rlxmm}}
% }
% \]
All models are weaker than $\scmm$. 
%Amongst the release-acquire models, $\sramm$ is stronger than $\ramm$, which is stronger than $\wramm$. 
%$\rcmm$ and $\psmm$ are weaker than $\ramm$ as they allow relaxed accesses. 
$\rlxmm$ is weaker than $\ramm$ but incomparable with $\wramm$.
In particular, the lack of synchronization across $\rf$ in $\rlxmm$ makes 
$\hb$ weaker in $\rlxmm$ compared to $\wramm$.
On the other hand, $\wramm$ allows extra behaviors over $\rlxmm$ because it does not enforce write-coherence.
Finally, $\rcmm$ can be viewed as a combination of $\ramm$ and $\rlxmm$, where 
fences may add synchronization between relaxed accesses.
See \cref{fig:mmex} for an illustration.

\bparagraph{Extensions of the models with non-atomics}
For ease of presentation, we do not explicitly handle non-atomic accesses.
The above memory models can be straightforwardly extended to include non-atomics
with ``catch-fire'' semantics, similarly to previous works~\cite{Lahav:2022}.
Intuitively, non-atomic accesses on any given location 
must always be $\hb$-ordered, as otherwise this implies a data race, leading to undefined behavior~\cite{cstandard,cppstandard}.

\bparagraph{The reads-from consistency problem}
An execution $\ex = \tup{\E, \po, \rf, \mo}$ is \emph{consistent} in a memory model $\MemModel$, written $\ex\models \MemModel$, if it satisfies the axioms of the model.
A \emph{partial execution} $\expartial=\tup{\E, \po, \rf}$ is an abstraction of executions without the modification order.
We call $\expartial$ \emph{consistent} in $\MemModel$, written similarly as $\expartial\models \MemModel$, if there exists an $\mo$ such that $\ex\models \MemModel$, where $\ex = \tup{\E, \po, \rf, \mo}$.
Thus the problem of \emph{reads-from consistency checking} (or simply consistency checking, from now on) is to find an $\mo$ that turns $\expartial$ consistent\footnote{Except for $\wramm$, the axioms of which do not involve $\mo$.}.
%The problem is known to be $\NP$-hard for $\scmm$~\cite{Gibbons:1997}, and in polynomial-time for $\ramm$~\cite{Lahav:2015,Abdulla:2018}.
%In this paper we resolve the complexity of consistency checking for all aforementioned memory models, with tight/almost-tight bounds as outlined  in \cref{subsec:contributions}.
%The coming sections present these results in detail.
%Due to space restrictions, we relegate all proofs to \cref{sec:app_rf}.
%\ucomment{Add 2 sentences describing the naive cubic algorithm in terms of computing a fixpoint relations or edges in a graph.}

%% file: figures/model_examples.tex
%!TEX root = ../main.tex

\begin{figure}
\centering
\def\ystep{0.5}
\begin{subfigure}{0.25\textwidth}
\centering
\scalebox{0.9}{
\begin{tikzpicture}[yscale=1]
  \node (t11) at (0,0*\ystep)  {$\wt(x)$};
  \node (t12) at (0,-1.5*\ystep) {$\wt(x)$};
  \node (t13) at (0,-3*\ystep) {$\wt(y)$};
  \node (t14) at (0,-4.5*\ystep) {$\rd(y)$};
  \node (t21) at (2.5,0*\ystep) {$\wt(y)$};
  \node (t22) at (2.5,-4.5*\ystep) {$\rd(x)$};
  \draw[po] (t11) to node[left]{$\po$} (t12);
  \draw[po] (t12) to (t13);
  \draw[po] (t13) to (t14);
  \draw[po] (t21) to (t22);
  \draw[rf,bend left=0] (t11) to node[above,pos=0.2, sloped]{$\rf$} (t22);
  \draw[fr,bend left=0] (t22) to node[below,pos=0.2, sloped]{$\fr$} (t12);
  \draw[mo,bend right=0] (t13) to node[above,pos=0.8, sloped]{$\mo$} (t21);
  \draw[rf,bend right=0] (t21) to node[below,pos=0.8, sloped]{$\rf$} (t14);
\end{tikzpicture}
}
\caption{
$\sramm$
}
\label{subfig:sra:nosc}
\end{subfigure}
\hfill
\begin{subfigure}{0.25\textwidth}
\centering
\scalebox{0.9}{
\begin{tikzpicture}[yscale=1]
  \node (t11) at (0,0*\ystep)  {$\wt(y)$};
  \node (t12) at (0,-2.5*\ystep) {$\wt(x)$};
  \node (t13) at (0,-4.5*\ystep) {$\rd(x)$};
  \node (t21) at (2.5,0*\ystep) {$\wt(x)$};
  \node (t22) at (2.5,-2.5*\ystep) {$\wt(y)$};
    \node (t23) at (2.5,-4.5*\ystep) {$\rd(y)$};
  \draw[po] (t11) to (t12);
  \draw[po] (t12) to (t13);
  \draw[po] (t21) to (t22);
  \draw[po] (t22) to (t23);
  \draw[rf,bend left=0] (t11) to node[above,pos=0.9, sloped]{$\rf$} (t23);
  \draw[rf,bend right=0] (t21) to node[above,pos=0.9, sloped]{$\rf$} (t13);
  \draw[mo,bend left=10] (t12) to node[above,pos=0.8, sloped]{$\mo$} (t21);
  \draw[mo,bend right=10] (t22) to node[above,pos=0.8, sloped]{$\mo$} (t11);
\end{tikzpicture} 
}  
\caption{
$\ramm$
}
\label{subfig:ra:nosra}
\end{subfigure}
\hfill
\begin{subfigure}{0.23\textwidth}
\centering
\scalebox{0.9}{
\begin{tikzpicture}[yscale=1]
  \node (t11) at (0,0*\ystep)  {$\wt(x)$};
  \node (t12) at (0,-4.5*\ystep) {$\rd(x)$};
  \node (t21) at (2.5,0*\ystep) {$\wt(x)$};
  \node (t22) at (2.5,-4.5*\ystep) {$\rd(x)$};
  \draw[po] (t11) to (t12);
  \draw[po] (t21) to (t22);
  \draw[rf,bend left=0] (t11) to node[above,pos=0.8, sloped]{$\rf$} (t22);
  \draw[rf,bend right=0] (t21) to node[above,pos=0.8, sloped]{$\rf$} (t12);
\end{tikzpicture}   
} 
\caption{
$\wramm$
}
\label{subfig:wra:notscperloc}
\end{subfigure}
\hfill
\begin{subfigure}{0.23\textwidth}
\centering
\scalebox{0.9}{
\begin{tikzpicture}[yscale=1]
  \node (t11) at (0,0*\ystep)  {$\wt(x)$};
  \node (t12) at (0,-2.25*\ystep) {$\wt(x)$};
  \node (t13) at (0,-4.5*\ystep) {$\wt(y)$};
  \node (t21) at (2.5,0,0*\ystep) {$\rd(y)$};
  \node (t22) at (2.5,-4.5*\ystep) {$\rd(x)$};
  \draw[po] (t11) to (t12);
  \draw[po] (t12) to (t13);
  \draw[po] (t21) to (t22);
  \draw[rf,bend left=0] (t11) to node[above,pos=0.1, sloped]{$\rf$} (t22);
  \draw[rf,bend right=0] (t13) to node[above,pos=0.1, sloped]{$\rf$} (t21);
  \draw[fr,bend right=0] (t22) to node[below,pos=0.5, sloped]{$\fr$} (t12);
\end{tikzpicture}      
}
\caption{
$\rlxmm$
}
\label{subfig:relaxed}
\end{subfigure}
\caption{
Executions consistent in various memory models.
$\mo$ edges  that go along $(\po\cup \rf)$ are not shown.
}
\label{fig:mmex}
\end{figure}

%% file: helpers.tex
%!TEX root=./main.tex

\section{Auxiliary Functions, Data Structures and Observations}\label{sec:helper_functions}

Our consistency checking algorithms rely on some common notation and computations.
To avoid repetition, we present these here as auxiliary functions, while we refer to \cref{fig:helpers} for examples.
%For the sake of completeness, we present the detailed algorithms and proofs
%of these algorithms and observations in \cref{subsec:app_helper_functions},

\bparagraph{Happens-before computation}
One common component in most of our algorithms is the computation of the $\hb$-timestamp $\HB_{\event} : \E.tids \to \nats_{\geq 0}$ 
of each event $e$, declaratively defined as
\begin{align*}
%\label{eq:hb-timestamp}
\HB_e(t) = |\setpred{f}{f.\tid = t \land (f, e) \in \hb^?}|
\end{align*}
That is, $\HB_e(t)$ points to the last event of thread $t$ that is $\hb$-ordered before (and including) $e$.
The computation of all $\HB_e$ can be computed by a relatively straightforward algorithm (see e.g.,~\cite{Luo:2021}).
We will thus take the following proposition for granted in this work.
%\Andreas{provide a citation, remove the rest} under $\MOrel$/$\MOacq$ accesses,
%but becomes more challenging when those are mixed with $\MOrlx$ accesses, as $\rf$ edges along $\MOrlx$ events do not induce synchronization (see $\sw$ in \cref{subsec:executions}).
%Nevertheless, we show that these timestamps can be computed in a \emph{streaming} fashion, 
%in $O(\NumEvents \cdot \NumThreads)$ time.
%Nevertheless, in \cref{subsec:app_helper_functions} we present a streaming algorithm \getHBTS{$\E$, $\po$, $\rf$} that returns an $\E$-indexed array $\HB[\event][t]$ of $\hb$-timestamps, arriving at the following theorem.

\begin{proposition}\label{prop:HB-computation}
The happens-before relation can be computed in $O(\NumEvents \cdot \NumThreads)$ time.
\end{proposition}

\bparagraph{Last write and last read}
%Our consistency algorithms rely on efficient access to events 
%identified by their thread, position in the thread and memory location.
% $\glw$ and $\glr$ defined as follows.
Given a thread $t$, location $x$, and index $c \in \nats_{\geq 0}$, we define
\begin{align*}
\glw(t, x, c) = 
\begin{cases} \event & 
\begin{array}{l}\text{if }  \event  \text{ is the last event such that, } \event \in \locx{\W} \cup \locx{\Upd}, \\
		\event.\tid = t \text{ and } |\setpred{g}{(g, \event) \in \po^?}| \leq c \\
\end{array} \\
		\bot & \text{ if no such event exists} 
\end{cases}
\end{align*}
\begin{align*}
\glr(t, x, c) = 
\begin{cases} \event & 
\begin{array}{l}\text{if }  \event  \text{ is the last event such that, } \event \in \locx{\R} \cup \locx{\Upd}, \\
		\event.\tid = t \text{ and } |\setpred{g}{(g, \event) \in \po^?}| \leq c \\
\end{array} \\
		\bot & \text{ if no such event exists} 
\end{cases}
\end{align*}
In words, $\glw(t, x, c)$ returns the latest $\po$-predecessor $\wt(t,x)$ or $\ud(t,x)$ of the $c$-th event of thread $t$ (similarly for $\glr(t, x, c)$).
When our consistency algorithms process a read/RMW event $\event(u,x)$, they query for $\glw(t, x, c)$ and $\glr(t, x, c)$ on each thread $t$, 
where $c=\HB_e(t)$.
We call $u$ the \emph{observer thread}.
Our efficient handling of such queries is based on the insight that, along subsequent queries from the same observer thread, $c$ is monotonically increasing ($\HB$ timestamps are monotonic along $\po$-ordered events).
We develop a simple data structure for handling such queries as follows.
For each thread $t$, memory location $x$, and observer thread $u$,
we maintain lists $\view{\WtLst}{t,x}{u}$ and $\view{\RdLst}{t,x}{u}$,
each containing the sequence of write/RMW events and read/RMW events performed by $t$ on $x$,
together with their thread-local indices.
The parameterization by $u$ ensures that $u$ observes its own local
version of this list.
Answering a query $\glw(t, x, c)$ consists of iterating over $\view{\WtLst}{t,x}{u}$ until the correct event is identified.
Subsequent queries continue the iteration from the last returned position.
The total cost of traversing all these lists is $O(\NumEvents\cdot\NumThreads)$ (each event appears in $\NumThreads$ lists, one per observer thread).
In pseudocode descriptions, we will call $\apicall{\view{\WtLst}{t,x}{u}}{get}{c}$ (resp., $\apicall{\view{\RdLst}{t,x}{u}}{get}{c}$)
to access the event $\event = \glw(t, x, c)$ (resp., $\event = \glr(t, x, c)$).
This implementation of $\glw$ and $\glr$ is novel compared to prior works (e.g.,~\cite{Luo:2021}), and crucial for obtaining the linear-time bounds developed in our work.
%In particular, previous works do not employ any data structures to answer such queries.
%They simply iterate over a list of events until the correct event is identified for each such query, and 
%no bookkeeping is performed which could be utilized in subsequent queries.

\input{figures/helpers}

\bparagraph{Top of, and position in  $\rf$-chain}
All memory models we consider satisfy weak-atomicity (atomicity implies weak atomicity),
i.e., no two RMW events can have the same writer.
This implies that all write and RMW events are arranged in disjoint \emph{$\rf$-chains},
where a chain is a maximal sequence of events $\event_0, \event_1, \ldots, \event_{\ell}$ ($\ell \geq 0$), 
such that
(i)~$\event_0\in \W$ and $\event_1, \ldots, \event_\ell\in \Upd$, and
(ii)~$\rfinv(\event_{i})=\event_{i-1}$ for each $i\geq 1$.
%\ucomment{Every RMW that reads from this chain also belongs to this \emph{maximal} chain}.
In words, we have a chain of events connected by $\rf$, starting with the top write event $e_0$, and (optionally) continuing with a maximal sequence of RMW events that read from this chain.
Given an event $\event\in  (\WRMW)$, we often refer to the $\rf$-chain that contains $\event$.
Specifically, the top of the chain $\TC[\event]$
is the unique event $f$ such that $(f, \event) \in \rf^*$ and $f.\op = \wt$.
The position $\PC[\event]$ of $e$ in its $\rf$-chain is
$|\setpred{f}{(f, \event) \in \rf^+}|$.
Both $\TC$ and $\PC$ can be computed in $O(\NumEvents)$ time for all events.

\begin{wrapfigure}{r}{0.39\textwidth}
\input{figures/chain-rule}
\end{wrapfigure}

\bparagraph{Conflicting triplets}
A \emph{conflicting triplet} (or just \emph{triplet}, for short) is a triplet of distinct events $(\event_1, \event_2, \event_3)$ such that
(i)~all events access the same location $x$,
(ii)~$\event_1, \event_3 \in (\locx{\W}\cup \locx{\Upd})$ and $\event_2\in (\locx{\R}\cup \locx{\Upd})$, and
(iii)~$\rfinv(\event_2)=\event_1$.

Finally, we state a simple lemma that is instrumental throughout the paper.
This lemma identifies certain $\mo$ orderings 
implied by the basic axioms of read-coherence, write-coherence and atomicity, and thus applies to all models except $\wramm$.
\cref{fig:Chain-Rule} provides an illustration.

%\scomment{formatting}
\begin{restatable}{lemma}{lemreadcoherenceatomicity}
\label{lem:read_coherence_atomicity}
Let $\ex = \tup{\E, \po, \rf, \mo}$ be an execution that satisfies read-coherence, write-coherence and atomicity.
Let $(\wt, \rd, \wt')$ be a triplet.
%and let $\wt^{\top}=\TC[\wt]$.
%If $(\wt', \rd)\in \hb$, or
%$(\wt', \rd)\in \rf;\hb^?$, then 
If $(\wt', \rd)\in \rf^?;\hb$ and 
$(\wt', \wt)\not \in \rf^+$, then
$(\wt', \TC[\wt])\in \mo$.
% \ucomment{Triplet not defined until this point. Similarly, $\TC$ is not defined until here.}
\end{restatable}

%\begin{example}
%In \cref{fig:helpers} we compute some example helper functions on the given trace. We mark the access modes some events which are relevant for the computed helper functions. 
%On the trace the event pair $(e_3, e_6)$ is in $\sw$ relation. Hence, for event $e_7$, event $e_3$ is the latest $\hb$ predecessor in thread $T_1$, $\po$-predecessor event $e_6$ in $T_4$, and no $\hb$ predecessor in threads $T_2$ and $T-3$. 
%In thread $T_1$, given the third event $e_3$ and location $x$, the latest $\po$-predecessor write and read events are $e_1$ and $e_2$ respectively. Threads $T_2$ and $T_3$ has no access on location $x$. However, given thread $T_2$ and location $y$, 
%event $e_4$ itself is the latest write.
%On the trace, $e_3, e_4, e_5$ is the $\U$-chain and $e_3$ is the unique start event of the chain. Hence, we directly compute an $\mo$ relation from $e_3$ to $e_5$.
%\end{example}

%% file: figures/helpers.tex
%!TEX root = ../main.tex

\begin{figure}
\centering
\begin{subfigure}{0.58\textwidth}
%\centering
\scalebox{0.9}{
\begin{tikzpicture}[yscale=0.6]

\node (t1) at (0,1)  {$t_1$};
\node (t2) at (2.5,1)  {$t_2$};
\node (t3) at (5,1)  {$t_3$};
\node (t4) at (7,1)  {$t_4$};
  \node (t11) at (0,0)  {$e_1: \ \wt(x)$};
  \node (t12) at (0,-2) {$e_2: \ \rd(x)$};
  \node (t13) at (0,-4) {$e_3: \ \wt^\MOrel(y)$};
  \node (t41) at (7,0) {$e_6: \ \rd^\MOacq(y)$};
  \node (t42) at (7,-2) {$e_7: \ \rd(x)$};
  \node (t21) at (2.5,-3) {$e_4: \Upd^\MOrlx(y)$};
  \node (t31) at (5,-1) {$e_5:\Upd^\MOrlx(y)$};
 
% \path[->] (e2)  edge [bend right=10]  node[below, align=left, sloped]  {$\rf^?;\hb$}         (e3);

  \draw[po] (t11) to node[left]{$\po$} (t12);
  \draw[po] (t12) to (t13);
  \draw[po] (t41) to (t42);
  \draw[rf,bend left=20] (t13) to node[above]{$\rf$} (t21);
  \draw[rf,bend left=20] (t21.north) to node[above]{$\rf$} (t31);
  \draw[rf,bend left=20] (t31.north) to node[above]{$\rf$} (t41);
%    
  %\draw[mo,bend left=30] (t13) to node[above]{$\mo$} (t31);
\end{tikzpicture}
}
% \caption{
% Execution Trace
% }
% \label{subfig:sra:nosc}
\end{subfigure}
\hfill
\begin{subfigure}{0.4\textwidth}
% $\HB_{e_{7}}(t_1) = 3, \ \HB_{e_{7}}(t_4) = 1$
% \\[1.5pt]
% $\HB_{e_{7}}(t_2) = \HB_{e_{7}}(t_3) = 0$
% \\[1.5pt]
% $\glw(t_1, x, 3) =  e_1$
% \\[1.5pt] 
% $\glw(t_2, x, 1) =  \bot$
% \\[1.5pt] 
% $\glw(T_2, y, 1) =  e_4$
% \\[1.5pt] 
% $\glr(t_1, x, 3) =  e_2$
% \\[1.5pt]
% $\TC[e_4] = \TC[e_5] = e_3$
% \\[1.5pt]
% $\PC[e_3]=0$, $\PC[e_4]=1$, $\PC[e_5] = 2$
\begin{tikzpicture}
\node[draw=black, dashed,rounded corners=.3cm, text width=5cm,align=left, inner sep=5pt,minimum height=4cm] at (7.8,-1.6) {
    \small
    $\HB_{e_{7}}(t_1) = 3, \ \HB_{e_{7}}(t_4) = 2$
\\[1.5pt]
$\HB_{e_{7}}(t_2) = \HB_{e_{7}}(t_3) = 0$
\\[3pt]
$\glw(t_1, x, 2) =  e_1$
\\[1.5pt] 
$\glw(t_2, x, 1) =  \bot$
\\[1.5pt] 
$\glw(t_2, y, 1) =  e_4$
\\[1.5pt] 
$\glr(t_1, x, 2) =  e_2$
\\[3pt]
$\TC[e_4] = \TC[e_5] = e_3$
\\[1.5pt]
$\PC[e_3]=0$, $\PC[e_4]=1$, $\PC[e_5] = 2$
    };
\end{tikzpicture}
\end{subfigure}
\caption{
Example of a partial execution (left) and its auxiliary functions (right).
}
\label{fig:helpers}
\end{figure}

%% file: figures/chain-rule.tex
%!TEX root=../main.tex
%\vspace{-0.3cm}
% \begin{figure}
%\vspace{-0.7cm}
\scalebox{1}{
\begin{tikzpicture}
%\footnotesize

\node[] (e1) at (0, 0) {$\TC[\wt]$};
\node[] (f1) at (1.5, 0) {$\ud$};
% \node[] (f2) at (3, 0) {$\ud$};
\node[] (f3) at (3.5, 0) {$\wt$};
\node[] (e3) at (4.5, 0) {$\rd$};
\node[] (e2) at (2.25, -0.9) {$\wt'$};
\node[] (condition) at (2.25, 0.9) {$(\wt',\wt)\not \in \rf^+$};

\draw [rf] (e1) -- (f1) node [midway, above] {$\rf$};
% \draw [rf] (f1) -- (f2) node [midway, above] {$\rf$};
\draw [rf, dotted] (f1) -- (f3) node [midway, above] {$\rf^*$};
\draw [rf] (f3) -- (e3) node [midway, above] {$\rf$};
\draw[->, thick] (e2)  [bend right=10] to  node[below, align=left, sloped]  {$\rf^?;\hb$}         (e3);
%\path[mo] (e2)  edge [bend left=10]  node[below, align=left, sloped]  {$\mopartial$}         (e1);
\draw[mo] (e2) [bend left=10] to node[below, align=left, sloped]  {$\mo$}         (e1);
\end{tikzpicture}
}
\caption{
An $\mo$ ordering implied by read-coherence, write-coherence and atomicity.
%The main step of \cref{algo:rf-consistency-relaxed} for constructing a minimally coherent partial modification order $\mopartial$.
}
\label{fig:Chain-Rule}
\vspace{-0.2cm}

%% file: rf.tex
\section{Consistency Checking}\label{sec:rf}

This section presents the main results of the paper, as outlined in \cref{subsec:contributions}. 
We start with an algorithm for checking consistency under $\wramm$ in \cref{subsec:rf_wra}.
% , followed by studying consistency checking for the $\sramm$ model. 
For $\sramm$, we show in \Cref{subsec:rf_sra_lower} that the problem is $\NP$-complete in general, but has a polynomial time algorithm for the RMW-free programs as shown in  
 \Cref{subsec:rf_sra_normw}. 
\Cref{subsec:rf_rc20} and \Cref{subsec:rf_rlx} show that consistency checking is polynomial time for $\rcmm$ and linear-time for $\rlxmm$ along with the respective algorithms. 
Finally, we study the lower bound of consistency checking for RMW-free $\ramm$, $\wramm$, and $\sramm$ in \Cref{subsec:rf_lower}.

\input{rf_wra}

\input{rf_sra}
%\input{rf_sra_lower}
%\input{rf_sra_upper}
\input{rf_sra_normw}

\input{rf_rc20}

\input{rf_rlx}

\input{rf_lower}

%% file: rf_wra.tex
%!TEX root=./main.tex

\subsection{Consistency Checking for WRA}\label{subsec:rf_wra}

We start with the $\wramm$ model, which is conceptually simpler as there is no $\mo$ involved in the consistency axioms.
\cref{algo:rf-consistency-wra} checks for consistency in 
 $O(\NumEvents \cdot\NumThreads)$, towards \cref{thm:wra_rf_upper}.

Given a partial execution $\expartial=(\E, \po, \rf)$,
the  algorithm first verifies that
there are no $(\po\cup\rf)$-cycles and every write/RMW event
is read by at most one RMW event (\cref{line:rf-wra-hb-acyclic}).
Afterwards, the algorithm streams the events of $\E$ in an order consistent with
$\po$ and verifies weak-read-coherence, i.e.,
there is no triplet $(\wt, \rd, \wt')$ such that
$\set{(\wt, \wt'), (\wt', \rd)} \subseteq \hb$.
In order to check this in linear time, the algorithm first
computes the array of $\hb$-timestamps (\cref{line:rf-wra-hb-array-init})
and the last write events $\glw(t, x, c)$ for each thread $t$,
location $x$ and index $c$ (\cref{line:rf-wra-lwb-array-init}),
as defined in \cref{sec:helper_functions}.

\input{algorithms/algo-wra-rf}

In order to check that weak-read-coherence is not violated,
at a read/RMW event $\event$ with $\event.\tid = t$ and $\event.\lloc = x$, 
\cref{algo:rf-consistency-wra} checks if there is an event $\event' \in \locx{\W}\cup \locx{\Upd}$
such that $\event'$ is $\hb$-sandwiched between $\rfinv(\event)$ and $\event$.
Since $\po\subseteq \hb$, it suffices to check if for any thread $u$,
the event $\glw(u, x, \HB_e[u])$ can play the role of $\event'$ above (\cref{line:wra-rf-weak-read-coherence}).

%Due to \cref{lem:Last-Per-Thread}, this amounts
%to simply checking if for each thread $u$,
%the event $\glw(u, x, \HB_e[u])$ satisfies the above property.
%This check is performed in \cref{line:wra-rf-weak-read-coherence}

%\scomment{Explain Line 4-10. Especially the conditions on line 9-10. Refer to Fig. 2?}

The total running time on an input partial execution $\expartial$
with $\NumEvents$ events %, $\NumMemory$ memory locations
and $\NumThreads$ threads can be computed as follows.
The initialization of $\HB$ and the lists $\set{\view{\WtLst}{t,x}{u}}_{t, x, u}$, and
% The calls to function \getHBTS{}, \initWtLst{}, and 
the total cost
of all calls to $\apicall{\view{\WtLst}{t,x}{u}}{get}{$c_u$}$
takes $O(\NumEvents\cdot  \NumThreads)$ time (\cref{sec:helper_functions}).
Afterwards, the algorithm spends $O(\NumThreads)$ time at each event.
This gives a total running time of $O(\NumEvents \cdot \NumThreads)$.
We thus have the following theorem.

\thmwrarfupper*

%% file: algorithms/algo-wra-rf.tex
%!TEX root = ../main.tex

%%%%%%%%%%%%%%%%%%%%%%%%%%%%%%%%%%%%%%%%%%%%%%%%%%%%%
% Algorithm(X, po, rf)
%     {G_x}_{x \in Vars} = split(G) // add immediate po edges between two events of same var
%     for each x:
%       let V_x, E_x = G_x.nodes, G_x.egdes
%       for events e in topological ordering of G_x
%           if e.type = w:
%                LW_t := we
%           if e.type = r:
%               e_rf := rf(e)
%               add mo edge from LW_t to e_rf    
%       check if the resulting graph has a cycle.
%%%%%%%%%%%%%%%%%%%%%%%%%%%%%%%%%%%%%%%%%%%%%%%%%%%%%
% \small
\begin{algorithm*}[t]
\Input{A partial execution $\expartial=(\E, \po, \rf)$}
\BlankLine
\lIf{$(\po\cup\rf)$ is cyclic or $\rf$ violates weak-atomicity}{
    \declare `Inconsistent' \label{line:rf-wra-hb-acyclic}
} 
\Let $\HB$ be an $\E$-indexed array storing the $\hb$-timestamps of events \label{line:rf-wra-hb-array-init} \;
% $\HB \gets$ \getHBTS{$\E$, $\po$, $\rf$} \label{line:rf-wra-hb-array-init} \;
\Let $\set{\view{\WtLst}{t,x}{u}}_{t, x, u}$ be data structures implementing $\glw(\cdot, \cdot, \cdot)$ \label{line:rf-wra-lwb-array-init} \;
% $\set{\view{\WtLst}{t,x}{u}}_{t, x, u} \gets$ \initWtLst{$\E$, $\po$}  \label{line:rf-wra-lwb-array-init} \;
\ForEach{$\event \in \E$ in $\po$-order}{
    \Case{$\event = \rd(t, x)$ or $\event = \ud(t, x)$}{
        \Let $w = \rfinv(\event)$, $t' = w.\tid$ and $c = \HB[w][t']$ \;
        \ForEach{$u \in \E.tids$}{
            \Let $c_u$ = \lIf*{$(\event.\op = \ud \land u = t)$}{$\HB[\event][u] - 1$} \lElse*{$\HB[\event][u]$}\;
            \Let $w_u $= $\apicall{\view{\WtLst}{t,x}{u}}{get}{$c_u$}$ \;%$\LWB[u][x][c_u]$ \;
            \lIf*{$(\HB[w_u][t’] \geq c) \land ( (u = t’) {\Rightarrow} \HB[w_u][t’] > c)$}{
                \declare `Inconsistent' \label{line:wra-rf-weak-read-coherence}
            }
        }
    }
}
\declare `Consistent'
% \vspace*{0.25cm}
\caption{Checking Consistency for $\wramm$.
\label{algo:rf-consistency-wra}
}
\end{algorithm*}
\normalsize

%% file: rf_sra.tex
%!TEX root=./main.tex
\subsection{Consistency Checking for SRA}\label{subsec:rf_sra_lower}

We now turn our attention to $\sramm$, and prove the bounds of \cref{thm:sra_rf_lower} and \cref{thm:sra_rf_upper}.

\bparagraph{The hardness of consistency checking for $\sramm$}
First, note that consistency checking is a problem in $\NP$.
Indeed, given a partial execution $\expartial=(\E, \po, \rf)$, 
one can guess a candidate $\mo$ and verify that $\ex\models \sramm$, 
where $\ex=(\E, \po, \rf, \mo)$ is a complete execution, by checking against the axioms of $\sramm$.
% It is not hard to see that e
Each axiom can be verified in polynomial time, giving us membership in $\NP$.
Now we turn our attention to $\WComplexity{1}$-hardness (which will also imply $\NP$-hardness).
Our reduction is from the consistency problem for $\scmm$, which is known to be 
$\NP$-hard~\cite{Gibbons:1997} and more recently shown to be $\WComplexity{1}$-hard~\cite{Mathur:lics2020}.
We obtain hardness in two steps.

First, we observe that the consistency problem for $\scmm$ is $\WComplexity{1}$-hard even over instances in which every write event is observed at most once.
This can be obtained from the proof of $\WComplexity{1}$-hardness in~\cite{Mathur:lics2020}.
Towards our $\WComplexity{1}$-hardness proof for $\sramm$, we can substitute 
in such instances every read access by an RMW access without affecting the $\scmm$ consistency of the execution.
Intuitively, as any write observed by a read does not have any other readers, the write of the substituting RMW has no effect.
Formally, we have the following lemma.

\begin{restatable}{lemma}{lemschardness}
\label{lem:sc_hardness}
Consistency checking for $\scmm$ with only write and RMW events is $\WComplexity{1}$-hard in the parameter $\NumThreads$.
\end{restatable}

Given \cref{lem:sc_hardness}, we can now prove \cref{thm:sra_rf_lower}.
The key observation is that the strong-write-coherence of $\sramm$ implies a total order on all write/RMW events.
Thus, over instances where every event is either a write or an RMW, strong-write-coherence yields a total order on all events, which, in turn, implies an $\scmm$-consistent execution.
We arrive at the following theorem.

\thmsrarflower*

\bparagraph{A parameterized upper bound}
We now turn our attention to \cref{thm:sra_rf_upper}, i.e., we solve consistency checking for $\sramm$ in time $O(\NumThreads\cdot \NumEvents^{\NumThreads+1})$.
Recall that our goal is to construct an $\mo$ that witnesses the consistency of $\expartial=(\E, \po, \rf)$.
One natural approach is to enumerate 
all possible $\mo$'s and check whether any of them leads to a consistent $\ex$.
However, this leads to an exponential algorithm 
regardless of the number of threads (there are exponentially many possible $\mo$'s even with two threads) 
which is beyond the bound of \cref{thm:sra_rf_upper}.
We instead follow a different approach.
% We start with some notation that simplifies our exposition later on.

%\scomment{To provide an intuition we may start with Fig. 6 example  and then discuss the algorithm. In that case we can refer to the example while discussing the algorithm.
%}

\eparagraph{Algorithm}
Given the poset $(\E,\hb)$, a set $Y \subseteq \E$ is said to be \emph{downward-closed} if for all 
$\event_1 \in Y, \event_2 \in \E$, 
if $(\event_2, \event_1)\in \hb$ then $e_2 \in Y$. 
\newcommand{\Gg}{\mathcal{G}}
We define a (directed) \emph{downward graph} $\Gg_{\expartial}$ induced by $(\E,\hb)$, and show that the question of $\expartial\models\sramm$ reduces to checking reachability in $\Gg_{\expartial}$.
The node set of $\Gg_{\expartial}$ consists of all downward closed subsets $S$ of $\E$, with $\emptyset$ being the \emph{root node} and $\E$ being the \emph{terminal node}.
Given a node $S$ in $\Gg_{\expartial}$, we insert edges $S\to S'$
where $S'$ is obtained by extending $S$ with an event which is \emph{executable} in $S$. 
An event $\event$ executable in $S$ if the following conditions hold.
\begin{compactenum}
\item\label{item:executability1} All events $\event'$ such that  $(\event', \event)\in \hb$ are in $S$.
\item\label{item:executability2} If $\event\in (\W\cup\Upd)$ is a write/RMW event, then it
must also be \emph{enabled}. 
We say that a write/RMW event $\wt$ is enabled if the following hold.

\begin{compactitem}
\item[(a)]\label{item:enabledness1} 
%\scomment{Mix of notations e.g. we are using $\wt$ for a set of writes. We can change it to w, r, w'.}
For every triplet $(\wt, \rd,\wt')$, if $(\wt', \rd) \in \hb$, then $\wt' \in S$.
Intuitively, executing $\wt$ while $\wt'\not \in S$ represents a guess that $(\wt,\wt')\in \mo$,
which would violate read-coherence as $(\wt',\rd)\in \hb$.
%\scomment{What is the relation between w and r? is it rf(w,r)?}
\item[(b)]\label{item:enabledness2}
For every RMW event $\ud$ and triplet $(\wt', \ud, \wt)$, if $\ud\not \in S$ then $\wt'\not \in S$.
Intuitively, executing $\wt$ while $\ud\not \in S$ but $\wt' = \rfinv(\ud) \in S$ represents a 
guess that $(\wt', \wt) \in \mo$ and $(\wt, \ud)\in \mo$, which would violate atomicity as 
it would imply $(\ud,\wt)\in\fr$ and thus $(\ud, \ud)\in \fr;\mo$.
\end{compactitem}
\end{compactenum}

\begin{wrapfigure}{r}{0.45\textwidth}
\input{figures/executability}
\end{wrapfigure}

\cref{fig:exec} illustrates the above notions.
Conceptually, every path from the root $\emptyset$ to a node $S$ in $\Gg_{\expartial}$ represents an $\mo$ on the write/RMW events of $S$.
Although there can be exponentially many such paths, the node $S$ ``forgets'' their corresponding $\mo$'s.
Instead, $S$ represents a partial $\mopartial$, which orders every write/RMW event on a location $x$ of $S$ before every write/RMW event on $x$ outside $S$.
The algorithm terminates and returns that $\expartial\models \sramm$ iff the node $\E$ is reachable from the root node $\emptyset$ in $\Gg_{\expartial}$.
Observe that $\Gg_{\expartial}$ contains $O(\NumEvents^{\NumThreads})$ nodes, while each node has $\leq \NumThreads$ successors.
Deciding whether a node has a transition to another node can be easily done in $O(\NumEvents)$ time.
Hence, the total time of the algorithm is $O(\NumThreads\cdot \NumEvents^{\NumThreads+1})$.
We thus arrive at the following theorem.

\thmsrarfupper*

%\begin{proof}[(Proof Sketch)]
%We now sketch the argument for correctness, while we refer to \cref{sec:app_rf} for the formal proof.
%First, assume that $\expartial\models \sramm$, and let $\ex=(\E,\po,\rf,\mo)$ be the corresponding witness.
%Due to strong-write-coherence, $(\hb\cup \mo)$ is acyclic.
%Then we can construct a path from the root of $\Gg_{\expartial}$ to node $\E$ as follows: whenever we are in a node $S$, if $S$ has an outgoing edge labeled with a read event that is not RMW, we follow that edge to its corresponding successor.
%Otherwise, we follow the edge outgoing $S$ that is labeled with a write event that is $(\hb\cup \mo)$ minimal among all the write events that label edges outgoing $S$.
%One can verify that $S$ will always have such a write event that is enabled, otherwise $\mo$ would violate one of the coherence or atomicity axioms of $\sramm$.
%Similarly, if we have a path from the root of $\Gg_{\expartial}$ to $\E$, we can construct an $\mo$ by ordering all writes on the same variable in the order that they appear on the path.
%\end{proof}

%\hunkar{
%The practical implication of this result is that, similar to the case of $\scmm$,  
%performing consistency checking for $\sramm$ efficiently requires developing 
%practically oriented heuristics that work well in the commonly observed cases.
%}

%% file: figures/executability.tex
%!TEX root=../main.tex
\vspace{-0.45cm}
\centering
\begin{tikzpicture}[thick,
pre/.style={<-,shorten >= 2pt, shorten <=2pt, very thick},
post/.style={->,shorten >= 2pt, shorten <=2pt,  very thick},
seqtrace/.style={->, line width=1},
aux_seqtrace/.style={->, line width=1, draw=gray},
und/.style={very thick, draw=gray},
event/.style={rectangle, minimum height=3.5mm, fill=white, minimum width=6mm,   line width=1pt, inner sep=1, font={\small}},
aux_event/.style={event, fill=gray!20},
virt/.style={circle,draw=black!50,fill=black!20, opacity=0},
bad/.style={preaction={fill, white}, pattern color=red!40, pattern=north east lines},
good/.style={preaction={fill, white}, pattern color=green!60, pattern=north west lines},
isLabel/.style={rectangle, fill opacity=0.5, fill=white, text opacity=1}
]

\newcommand{\LocalTrace}{\rho}

\newcommand{\xstep}{2.45}
\newcommand{\ystep}{0.6}%70909}
\newcommand{\yaux}{0.17727}

\fill[gray!20]
(-0.25*\xstep, 0.32*\ystep) to
(-0.25*\xstep, -0.5*\ystep) to
(1*\xstep, -1.5*\ystep - \yaux) to
(2.3*\xstep, -1.5*\ystep - \yaux) to
(2.3*\xstep, 0.32*\ystep) to cycle;

\node[]       (t1)   at (0*\xstep, 0*\ystep) {$t_1$};
\node[]       (t1_end) at (0*\xstep, -4*\ystep) {};

\node[]       (t2)   at (0.5*\xstep, 0*\ystep) {$t_2$};
\node[]       (t2_end) at (0.5*\xstep, -4*\ystep) {};

\node[]       (t3)   at (1.0*\xstep, 0*\ystep) {$t_3$};
\node[]       (t3_end) at (1.0*\xstep, -4*\ystep) {};

\node[]       (t4)   at (1.5*\xstep, 0*\ystep) {$t_4$};
\node[]       (t4_end) at (1.5*\xstep, -4*\ystep) {};

\node[]       (t5)   at (2.0*\xstep, 0*\ystep) {$t_5$};
\node[]       (t5_end) at (2.0*\xstep, -4*\ystep) {};

\draw[po] (t1) to (t1_end);
\draw[po] (t2) to (t2_end);
\draw[po] (t3) to (t3_end);
\draw[po] (t4) to (t4_end);
\draw[po] (t5) to (t5_end);

\node[event]  (t1_r_x)   at (0*\xstep, -1.5*\ystep) {$\rd(x)$};
\node[event]  (t1_w2_y)   at (0*\xstep, -3*\ystep) {$\wt_1(y)$};
\node[event]  (t2_r_y1)   at (0.5*\xstep, -2*\ystep) {$\rd_1(y)$};
\node[event]  (t2_r_y2)   at (0.5*\xstep, -3*\ystep) {$\rd_2(y)$};
\node[aux_event]  (t3_w_x)  at (1.0*\xstep, -1*\ystep - \yaux) {$\wt(x)$};
\node[event]  (t3_w1_y)  at (1.0*\xstep, -3*\ystep) {$\wt_2(y)$};
\node[aux_event]  (t4_u1_z)  at (1.5*\xstep, -1*\ystep - \yaux) {$\ud_1(z)$};
\node[event]  (t4_w_z)  at (1.5*\xstep, -3*\ystep) {$\wt(z)$};
\node[event]  (t5_u2_z)  at (2.0*\xstep, -3*\ystep) {$\ud_2(z)$};

\draw[rf]   (t3_w_x) -- (t1_r_x) node [midway, pos=0.7, above,sloped] {\small$\rf$};;
\draw[rf]   (t3_w1_y) -- (t2_r_y2) node [midway, above,sloped] {\small$\rf$};;
\draw[rf]   (t1_w2_y) -- (t2_r_y1) node [midway, above,sloped] {\small$\rf$};
\draw[hb]   (t1_w2_y) -- (t2_r_y2) node [midway, below,sloped] {\small$\hb$};
\draw[rf]   (t4_u1_z) -- (t5_u2_z) node [midway, above,sloped] {\small$\rf$};;

\end{tikzpicture}
% \vspace{-0.5cm}
\caption{
Enabledness and executability of a set $S$, marked in gray. 
%Already executed events are marked grey and constitute the set $S$.
%Amongst the remaining events, 
Only $\rd(x)$ and $\ud_2(z)$ are executable events.
$\wt_1(y)$ is not executable as $\rd(x) \not\in S$.
$\rd_1(y)$ is not executable as $\wt_1(y) \not\in S$.
$\wt_2(y)$ is not enabled as $(\wt_1(y), \rd_2(y)) \in \hb$ but $\wt_1(y) \not\in S$.
$\wt(z)$ is not enabled as $\ud_2(z) \not\in S$.
}
\label{fig:exec}
\vspace{-0.5cm}
% \end{figure*}

%% file: rf_sra_normw.tex
%!TEX root = ./main.tex

\subsection{Consistency Checking for the RMW-Free Fragment of $\sramm$}\label{subsec:rf_sra_normw}

On close inspection, RMW events played a central role in the $\NP$-hardness of \cref{thm:sra_rf_lower}.
A natural question thus arises: does the hardness persist in the absence of RMWs?
Here we show that the RMW-free fragment of $\sramm$ can be handled efficiently (\cref{thm:sra_rf_normw_upper}).
Observe that \cref{lem:read_coherence_atomicity} applies to $\sramm$, as strong-write-coherence implies write-coherence.
However, the lemma admits a simplification under RMW-free $\sramm$.
In particular, as $\ex$ does not contain $\MOrlx$ accesses, $(\wt', \rd)\in \rf^?;\hb$ reduces to $(\wt', \rd)\in \hb$.
Moreover, as $\ex$ is RMW-free, we have $\TC[\wt]=\wt$, and hence $(\wt', \wt)\in \mo$.
Thus, \cref{lem:read_coherence_atomicity} reduces to the following corollary.

\begin{corollary}\label{cor:read_coherence_sra}
Consider any execution $\ex = \tup{\E, \po, \rf, \mo}$ that satisfies read-coherence and strong-write-coherence.
Consider any triplet $(\wt, \rd, \wt')$.
If $(\wt', \rd)\in \hb$ then
$(\wt', \wt)\in \mo$.
\end{corollary}

\bparagraph{Minimal coherence under $\sramm$}
\cref{cor:read_coherence_sra} identifies necessary orderings 
in any modification order that witnesses the consistency of $\expartial$.
Towards an algorithm, we must also determine if
there are non-trivial conditions \emph{sufficient} to conclude consistency.
We answer this in the positive, by capturing these conditions in the notion of \emph{minimal coherence}.
Consider a partial modification order $\mopartial=\bigcup_x \mopartial_x$, 
where each $\mopartial_x$ is a partial order.
We call $\mopartial$ \emph{minimally coherent for $\expartial$ under $\sramm$} if the following conditions hold.
\begin{compactenum}
\item\label{item:minimal_coherence_sra1}
For every triplet $(\wt, \rd, \wt')$ with $(\wt', \rd)\in \hb$, we have $(\wt', \wt)\in (\hb\cup\mopartial)^+$.
\item\label{item:minimal_coherence_sra2}
$(\hb\cup \mopartial)$ is acyclic.
\end{compactenum}

\input{figures/min-coh-sra}

%\ucomment{Should we make this a definition environment?}
\cref{fig:mincoh-sra} illustrates the notion of minimal coherence under $\sramm$.
Observe that any $\mo$ witnessing the consistency of $\expartial$ satisfies these conditions.
In the following, we show that minimally coherent modification orders are
also sufficient for witnessing consistency.
We note that for RMW-free executions, minimal coherence coincides 
with the previous notions of coherence that witnesses consistency under $\ramm$~\cite{Lahav:2015,Abdulla:2018,Luo:2021}.
However, these notions also handle RMWs, and are not directly applicable in $\sramm$, as the problem of consistency checking is $\NP$-hard for $\sramm$ with RMWs (\cref{thm:sra_rf_lower}).

\begin{restatable}{lemma}{lemsraminimalcoherence}
\label{lem:sra_minimal_coherence}
Consider any RMW-free, partial execution $\expartial=(\E,\po, \rf)$.
If there exists partial modification order $\mopartial$ that is minimally coherent for $\expartial$ under $\sramm$, then $\expartial\models\sramm$.
\end{restatable}

\input{algorithms/algo-SRA_no_RMW-rf}

\bparagraph{Algorithm}
\cref{cor:read_coherence_sra} and \cref{lem:sra_minimal_coherence} suggest a polynomial-time algorithm for deciding the $\sramm$ consistency of an RMW-free, partial execution $\expartial=(\E, \po, \rf)$.
Similarly to $\wramm$, we first verify that $(\po\cup \rf)$ is acyclic.
Then, we construct a partial modification order $\mopartial$ by identifying all conflicting triplets $(\wt, \rd, \wt')$ such that $(\wt', \rd)\in \hb$, and inserting an ordering $(\wt', \wt)\in \mopartial$.
Finally, we report that $\expartial\models \sramm$ iff $(\hb\cup \mopartial)$ is acyclic.

Although this process runs in polynomial time, it is still far from the  nearly linear bound 
we aim for (\cref{thm:sra_rf_normw_upper}).
The key extra step towards this bound comes from a closer look at minimal coherence:~based on \cref{item:minimal_coherence_sra1}, it suffices to only consider 
conflicting triplets $(\wt, \rd, \wt')$ in which $\wt'$ is $\po$-maximal among all write events $\wt''$ forming a conflicting triplet $(\wt, \rd, \wt'')$ such that $(\wt'', \rd) \in \hb$.
For each thread $t$, we thus only need to identify the $\po$-maximal write $\wt'$ in the scheme outlined above.
This concept is illustrated in~\cref{fig:mincoh-sra}.
Consider the triplets $(\wt_3, \rd_2, \wt_4)$, $(\wt_3, \rd_2, \wt_5)$, and
$(\wt_3, \rd_2, \wt_6)$. 
Only the first two satisfy the above definition (since $(\wt_4, \rd_2), (\wt_5, \rd_2) \in \hb$ but $(\wt_6, \rd_2) \not \in \hb$).
In this case, only identifying the event $\wt_5$ is sufficient as it is the po-maximal write among $\wt_4$ and $\wt_5$.

This insight is precisely formulated in
\cref{algo:rf-consistency-sra-no-rmw}.
The algorithm uses the auxiliary functions from \cref{sec:helper_functions} 
to compute the $\HB$-timestamp of each event.
Also recall that, for threads $t$ and $u$ and location $x$,
$\view{\WtLst}{t,x}{u}$ denotes (thread $u$'s copy of) the $\po$-ordered 
list of write accesses performed by $t$ on location $x$.
It then processes events in $\expartial$ in an order consistent with $\po$ 
and builds a partial modification order $\mopartial$.
When processing a read event $\rd$, the algorithm identifies 
for every thread $u$, the $\po$-maximal write $\wt'$ of $u$ that 
forms a conflicting triplet $(\wt, \rd, \wt')$ with $(\wt', \rd) \in \hb$
 (\cref{line:rf-sra-lwb}), and inserts $(\wt', \wt)\in \mopartial$ (\cref{line:rf-sra-add-mo-edge}).
Finally, it checks whether $\mopartial$ violates strong-write-coherence.

\bparagraph{Correctness and complexity}
The completeness follows directly from \cref{cor:read_coherence_sra}: every ordering inserted in $\mopartial_x$ is present in any modification order $\mo$ that witnesses the consistency of $\expartial$, while the acyclicity check in \cref{line:rf-sra-cycle} is necessary for strong-write-coherence.
Hence, if the algorithm returns ``Inconsistent'', we have $\expartial\not\models \sramm$.
The soundness comes from the fact that $\mopartial$ satisfies \cref{item:minimal_coherence_sra1} of minimal coherence at the end of the loop of \cref{line:rf-sra-loop-start}, while if the acyclicity check in \cref{line:rf-sra-cycle} passes, $\mopartial$ also satisfies \cref{item:minimal_coherence_sra2} of minimal coherence.
Thus, by \cref{lem:sra_minimal_coherence}, $\expartial\models \sramm$.

%\bparagraph{Time Complexity}
The time spent in computing $\HB$ and initializing and accessing the lists $\view{\WtLst}{t,x}{u}$
is $O(\NumEvents\cdot\NumThreads)$ (\cref{sec:helper_functions}).
The number of orderings added in $\mopartial$ is $O(\NumEvents \cdot \NumThreads)$,
taking $O(\NumEvents \cdot \NumThreads)$ total time.
Finally, the check in  \cref{line:rf-sra-cycle}
is $O(\NumEvents \cdot \NumThreads)$ time
as it corresponds to detecting a cycle on a graph
with $|\E| = \NumEvents$ nodes and $\leq \NumEvents \cdot (\NumThreads + 1)$ edges.
This gives a total running time of $O(\NumEvents\cdot\NumThreads)$.
We thus arrive at \cref{thm:sra_rf_normw_upper}.

\thmsrarfnormwupper*

%% file: figures/min-coh-sra.tex
\begin{wrapfigure}{r}{0.4\textwidth}
\centering
\vspace{-0.2cm}
\begin{tikzpicture}[node distance=3.5mm and 7mm]
  \node(r1)  at (0,0) {$\rd_1$};
  \node[below = of r1] (w1) {$\wt_1$};
  \node[below = of w1] (r2) {$\rd_2$};
  
  \node[right = of r1]  (w2) {$\wt_2$};
  \node[below = of w2] (w3) {$\wt_3$};
  
  \node[right = of w2] (w4) {$\wt_4$};
  \node[below = of w4] (w5) {$\wt_5$};
  \node[below = of w5] (w6) {$\wt_6$};

   \draw[po] (w2) to node[right]{$\po$} (w3);
   \draw[po] (r1) to node[right]{$\po$} (w1);
   \draw[po] (w1) to node[right]{} (r2);

   \draw[po] (w4) to node[right]{$\po$} (w5);
   \draw[po] (w5) to node[right]{} (w6);

   \draw[rf, bend right=38] (w5) to node[above,sloped, pos=0.4]{$\rf$} (r1);
   \draw[bend left, rf] (w3) to node[right]{$\rf$} (r2);
   
   \draw[mo] (w1) to node[below, sloped]{$\mopartial$} (w3);
   \draw[mo] (w5) to node[below, sloped]{$\mopartial$} (w3);

% %
%   \draw[po] (t21) to (t22);
% %
%   \draw[rf,bend left=0] (t11) to node[right,pos=0.2]{$\rf$} (t22);
%   \draw[fr,bend left=0] (t22) to node[left,pos=0.2]{$\fr$} (t12);
% %    
%   \draw[mo,bend right=0] (t13) to node[left,pos=0.9]{$\mo$} (t21);
%   \draw[rf,bend right=0] (t21) to node[right,pos=0.8]{$\rf$} (t14);
\end{tikzpicture}
\caption{
A minimally coherent $\mopartial$ under $\sramm$.
All events access the same location.
}
\label{fig:mincoh-sra}
\vspace{-0.3cm}
\end{wrapfigure}

%% file: algorithms/algo-SRA_no_RMW-rf.tex
%!TEX root = ../main.tex

% \small
\begin{algorithm*}[t]
\Input{Events $\E$, program order $\po$ and reads-from relation $\rf$}
\BlankLine
\lIf{$(\po \cup \rf)$ is cyclic}{
    \declare `Inconsistent' \label{line:rf-sra-hb-acyclic}
} 
\Let $\HB$ be an $\E$-indexed array storing the $\hb$-timestamps of events \label{line:rf-sra-hb-array-init} \;
% $\HB \gets$ \getHBTS{$\E$, $\hb$} \label{line:rf-sra-hb-array-init} \;
\Let $\set{\view{\WtLst}{t,x}{u}}_{t, x, u}$ be data structures implementing $\glw(\cdot, \cdot, \cdot)$ \label{line:rf-sra-lwb-array-init} \;
% $\set{\view{\WtLst}{t,x}{u}}_{t, x, u} \gets$ \initWtLst{$\E$, $\po$} \label{line:rf-sra-lwb-array-init} \;
\lForEach{$x \in \E.locs$}{
    $\mopartial_x \gets \emptyset$
}
\ForEach{$\event \in \E$ in $\po$-order}{ \label{line:rf-sra-loop-start}
    \Case{$\event = \rd(t, x)$}{
        \Let $\wt_{\mathsf{rf}} = \rfinv(\event)$ \;
        \ForEach{$u \in \E.tids$}{\label{line:rf-sra-saturation-loop}
            % \Let $\wt_u = \LWB[u][x][\HB[\event][u]]$ \label{line:rf-sra-lwb} \;
            \Let $\wt_u = \apicall{\view{\WtLst}{u,x}{t}}{get}{$\HB[\event][u]$}$ \label{line:rf-sra-lwb} \;
            \lIf*{$\wt_u \neq \wt_{\mathsf{rf}}$}{
            $\mopartial_x \gets \mopartial_x \cup \set{(\wt_u, \wt_{\mathsf{rf}})}$ \; \label{line:rf-sra-add-mo-edge}
            }
        }
    }
}
\lIf{$(\hb \cup \bigcup_{x \in \E.locs}\mopartial_x)$ is cyclic}{\declare `Inconsistent' \label{line:rf-sra-cycle}}
\lElse{\declare `Consistent'}
% \vspace*{0.25cm}
\caption{Checking consistency for the RMW-free fragment of SRA.}
\label{algo:rf-consistency-sra-no-rmw}
\end{algorithm*}
\normalsize

%% file: rf_rc20.tex
%!TEX root=./main.tex
\subsection{Consistency Checking for $\rcmm$}\label{subsec:rf_rc20}

We now turn our attention to the full $\rcmm$ memory model, which comprises a mixture of $\MOrel$, $\MOacq$ and $\MOrlx$ memory accesses.
Similarly to the RMW-free $\sramm$, we obtain a nearly linear bound (\cref{thm:rc_rf_upper}).
Note, however, that here we also allow RMW events.
As $\rcmm$ satisfies read-coherence, write-coherence and atomicity, \cref{lem:read_coherence_atomicity} applies also in this setting.
However, our earlier notion of minimal coherence under $\sramm$ is no longer applicable
as is --- \cref{lem:sra_minimal_coherence} does not hold for $\rcmm$.
Fortunately, we show this model enjoys a similar notion of coherence minimality.

\bparagraph{Minimal coherence under $\rcmm$}
Consider a partial modification order $\mopartial=\bigcup_x \mopartial_x$.
We call $\mopartial$ \emph{minimally coherent for $\expartial$ under $\rcmm$} if the following conditions hold.

\begin{compactenum}
\item\label{item:minimal_coherence_rc1} For every triplet $(\wt, \rd, \wt')$ accessing location $x$, if $(\wt', \rd)\in \rf^?;\hb$ and $(\wt', \wt)\not \in \rf^+$, then $(\wt', \TC[\wt])\in (\rf_x\cup \hb_x \cup \mopartial_x)^+$.
\item\label{item:minimal_coherence_rc2} For every two write/RMW events $\wt_1, \wt_2$ accessing location $x$,
if $(\wt_1, \wt_2)\not \in \rf^+$ and $(\wt_1, \wt_2)\in \mopartial_x$, then $(\wt_1, \TC[\wt_2])\in \mopartial_x$.
\item\label{item:minimal_coherence_rc3} $(\rf_x\cup\hb_x\cup \mopartial_x)$ is acyclic, for each $x\in\E.locs$.
\end{compactenum}
%Recall that $\TC[\wt]$ is the top of the $\rf$-chain of $\wt$.

\cref{fig:mincoh} illustrates the above definition.
Observe that any $\mo$ witnessing the consistency of $\expartial$ satisfies minimal coherence.
As before, minimal coherence is also a sufficient witness of consistency.

\begin{restatable}{lemma}{lemrcminimalcoherence}
\label{lem:rc_minimal_coherence}
Consider any partial execution $\expartial=(\E,\po, \rf)$.
If there exists partial modification order $\mopartial$ that is minimally coherent for $\expartial$ under $\rcmm$, then $\expartial\models\rcmm$.
\end{restatable}

\input{figures/min-coh}
Our algorithm for consistency checking in $\rcmm$ relies on \cref{lem:rc_minimal_coherence} to construct a minimally-coherent partial modification order that witnesses the consistency of $\expartial$.
In particular, the algorithm employs the simple inference rule of $\mopartial$ edges illustrated earlier in \cref{fig:Chain-Rule},
and is a direct application of \cref{item:minimal_coherence_rc1} of minimal coherence.
At a glance, it might come as a surprise that such a simple rule suffices to deduce consistency.
Indeed, analogous relations have been used in the past as 
consistency witnesses (e.g., the writes-before order~\cite{Lahav:2015}, saturated traces~\cite{Abdulla:2018}, or C11Tester's framework~\cite{Luo:2021}).
However, these witness relations are stronger than minimal coherence,
while the algorithms for computing them
(and thus checking consistency) have a higher polynomial complexity $O(\NumEvents^3)$ (or $O(\NumEvents^2\cdot \NumThreads)$) compared to our nearly linear bound.

%On the technical level, a crucial difference between prior witness relations (\cite{Lahav:2015,Abdulla:2018,Luo:2021})
%and a minimally coherent $\mopartial$ is the following.
%Owing to the minimality of $\mopartial$, 
%not every total extension of $(\rf_x\cup \hb_x \cup  \mopartial_x)$ 
%qualifies as the complete $\mo_x$ that witnesses consistency;
%in particular, some extensions might violate atomicity.
%For example, in \cref{fig:mincoh} observe that $(\ud_1, \wt_2) \in \mopartial$. 
%This implies that, due to atomicity, in any valid total extension $(\ud_2, \wt_2)$ must be ordered.
%However, minimal coherence does not force $(\ud_2, \wt_2)$ in $\mopartial$, in contrast to previous works.
%\hunkar{
%(Even if we had ordered $(\ud_2, \wt_2)$, certain total extension could still violate atomicity as %$\wt_3$ is unordered with the events of the RMW chain that starts with $\wt_1$. Previous works (e.g., %C11Tester) also would not order those events.)
%}

%\hunkar{
On a more technical level, not every total extension of $(\rf_x\cup \hb_x \cup  \mopartial_x)$ 
qualifies as the complete $\mo_x$ that witnesses consistency;
in particular, some extensions might violate atomicity.
This is also the case in prior witness relations~\cite{Lahav:2015,Abdulla:2018,Luo:2021}.
However, a key difference between prior work and minimal coherence is the following.
%Owing to the minimality of $\mopartial$, the number of total extensions of a $\mopartial$ that does not correspond to a complete $\mo_x$ are in general larger than of prior witness relations.
In prior witness relations, the events of an $\rf$-chain are either totally ordered or unordered with respect to any event outside this chain.
In contrast, minimal coherence allows only some events of the $\rf$-chain being ordered with outside events.
For example, in \cref{fig:mincoh} observe that $(\ud_1, \wt_2) \in \mopartial$. 
This implies that, due to atomicity, the pair $(\ud_2, \wt_2)$ must be ordered in any valid total extension of $\mopartial$.
However, minimal coherence does not force $(\ud_2, \wt_2)$ in $\mopartial$.
%}
%As a result, certain total extensions may violate atomicity due to such missing orderings.
%A naive approach would be to require $(\ud_2, \wt_2) \in \mopartial$.
%However, this would correspond to a more conservative formalization of minimal coherence.
Nevertheless, our proof of~\cref{lem:rc_minimal_coherence} shows that, as long as $\mopartial_x$ is minimally coherent, there always exists an extension $\mopartial'_x\supseteq \mopartial_x$ 
that can serve as the witnessing modification order, in the spirit of the prior notions of witness relations.
In \cref{fig:mincoh}, for example, this extension
would be $\mopartial'_x = \mopartial_x \cup (\ud_2, \wt_3)$.
%Moreover, such an extension $\mopartial'$ can be constructed
%in linear time, following the completion procedure described 
%in the proof of~\cref{lem:rc_minimal_coherence}.

\input{algorithms/algo-rc20-rf}

\bparagraph{Algorithm}
The insights made above are turned into a consistency checking procedure in \cref{algo:rf-consistency-rc20}.
This algorithm first verifies the absence of $(\po\cup \rf)$ cycles 
(which also implies that $\hb\subseteq (\po\cup \rf)^+$ is irreflexive), and that $\rf$ follows weak-atomicity (\cref{line:rf-rc20-po-rf-acyclic}).
Then, it computes auxiliary data discussed in~\cref{sec:helper_functions} (\crefrange{line:rf-rc20-hb-array-init}{line:rf-rc20-tc-pc-array-init}).
The main computation is performed in \Crefrange{line:rf-rc20-mo-infer-start}{line:rf-rc20-update-mo}, 
where the algorithm constructs a minimally coherent partial modification 
order $\mopartial_x$ for each location $x$.
The algorithm iterates over all read/RMW events $\event$ 
accessing some location $x$, and identifies $\wt_{\mathsf{rf}} = \rfinv(\event)$.
Then, it iterates over all threads $u$ and identifies 
the $\po$-maximal write/RMW event $\wt_u$ such that
either $(\wt', \event)\in \hb$ (in which case $\wt'$ is the event $\wt_u^\wt$ in \cref{line:rf-rc20-latest-writes})
or $(\wt', \event)\in \rf;\hb$ (in which case $\wt'$ is the event $\wt_u^\rd$ in \cref{line:rf-rc20-latest-writes}).
It then checks whether $(\wt_u, \wt_{\mathsf{rf}})\not\in \rf^+$,
by checking that either $\wt_u$ and $\wt_{\mathsf{rf}}$
belong to different $\rf$-chains (`$(\TC[\wt_{\mathsf{rf}}] \neq \TC[\wt_u])$'), 
or $\wt_{\mathsf{rf}}$ appears earlier than $\wt_u$
in the common $\rf$-chain (`$(\PC[\wt_{\mathsf{rf}}] < \PC[\wt_u])$'); 
see~\cref{line:rf-rc20-check-chain-condition}.
If so, the algorithm inserts an ordering $(\wt_u, \TC[\wt_{\mathsf{rf}}]) $ in $\mopartial_x$ (\cref{line:rf-rc20-update-mo}).
Finally, \cref{line:rf-rc20-cycle} verifies that $\mopartial_x$ satisfies write-coherence.
\cref{fig:mincoh-rc20-rlx-a} displays the resulting $\mopartial$ computed by \cref{algo:rf-consistency-rc20} on a partial execution.

\bparagraph{Correctness and complexity}
Completeness follows from \cref{lem:read_coherence_atomicity}: every ordering inserted in $\mopartial_x$ is present in any modification order $\mo$ that witnesses the consistency of $\expartial$, while the acyclicity check in \cref{line:rf-rc20-cycle} is necessary for write-coherence.
Hence, if the algorithm returns ``Inconsistent'', $\expartial\not\models \rcmm$.
The soundness comes from the fact that $\mopartial$ satisfies \cref{item:minimal_coherence_rc1} of minimal coherence at the end of the loop of \cref{line:rf-rc20-mo-infer-start}.
Since all orderings inserted in $\mopartial$ are to the top of an $\rf$-chain,
\cref{item:minimal_coherence_rc2} of minimal coherence is trivially satisfied at all times.
Finally, if the acyclicity check in \cref{line:rf-rc20-cycle} passes, $\mopartial$ also satisfies \cref{item:minimal_coherence_rc3} of minimal coherence.
Thus, by \cref{lem:rc_minimal_coherence}, we have $\expartial\models \rcmm$.

%\bparagraph{Time Complexity}
The time spent in computing $\HB$, $(\TC,\PC)$ and 
accessing the lists $\set{\view{\WtLst}{t,x}{u}}_{t, x, u}$ and $\set{\view{\RdLst}{t,x}{u}}_{t, x, u}$
is $O(\NumEvents\cdot \NumThreads)$ (\cref{sec:helper_functions}).
The number of orderings added in $\mopartial$ is $O(\NumEvents \cdot \NumThreads)$,
taking $O(\NumEvents \cdot \NumThreads)$ total time.
For each location $x\in \E.locs$, the acyclicity check in \cref{line:rf-rc20-cycle} can be performed in $O(\NumEvents_x\cdot \NumThreads)$ time, where $\NumEvents_x=|\locx{\W}\cup\locx{\Upd}|$.
For this, we construct a graph $G_x$ that consists of all events $(\locx{\W}\cup\locx{\Upd})$ and $O(\NumEvents_x\cdot \NumThreads)$ edges.
Given two events $\event_1=(t_1, x)$, $\event_2=(t_2, x)$ we have an edge $\event_1\to \event_2$ in $G_x$ iff $(\event_1, \event_2)\in (\rf\cup \mopartial_x)$ or 
% $\event_1=\LWB[t_1][\HB[\event_2][t_2]-1]$.
$\event_1=\glw(t_1, x, \HB[\event_2][t_1]-1)$.
We then check for a cycle in $G_x$.
Repeating this for all locations $x$, we obtain $O(\NumEvents\cdot\NumThreads)$ total time.
We thus arrive at \cref{thm:rc_rf_upper}.

\thmrcrfupper*

%The time spent in computing $\HB$, $(\TC,\PC)$ and
%accessing the lists $\set{\view{\WtLst}{t,x}{u}}_{t, x, u}$ and $\set{\view{\RdLst}{t,x}{u}}_{t, x, u}$
%is $O(\NumEvents\cdot \NumThreads)$ (\cref{sec:helper_functions}).
%The number of orderings added in $\mopartial$ is $O(\NumEvents \cdot \NumThreads)$,
%taking $O(\NumEvents \cdot \NumThreads)$ total time.
%For each location $x\in \E.locs$, the acyclicity check in \cref{line:rf-rc20-cycle} can be performed in $O(\NumEvents_x\cdot \NumThreads)$ time, where $\NumEvents_x=|\locx{\W}\cup\locx{\U}|$.
%For this, we construct a graph $G_x$ that consists of all events $(\locx{\W}\cup\locx{\U})$ and $O(\NumEvents_x\cdot \NumThreads)$ edges.
%Given two events $\event_1=(t_1, x)$, $\event_2=(t_2, x)$ we have an edge $\event_1\to \event_2$ in $G_x$ iff $(\event_1, \event_2)\in (\rf_x \cup \hb_x\cup \mopartial_x)$ or 
%$t_1=t_2=t$ and $\event_1=\glw(t, x, \HB[\event_2][t]-1)$.
%%$\event_1=\apicall{\view{\WtLst}{u,x}{t}}{get}{$\HB[\event_2][t_2]-1$}$.
%We then check for a cycle in $G_x$.
%Repeating this for all locations $x$, we obtain $O(\NumEvents\cdot\NumThreads)$ total time.
%We thus arrive at \cref{thm:rc_rf_upper}.

\input{figures/min-coh-relaxed-rc20}

%% file: figures/min-coh.tex
\begin{wrapfigure}{r}{0.4\textwidth}
\vspace{-0.45cm}
\centering
\begin{tikzpicture}[node distance=5mm and 7mm]
  \node (w1) at (0,0)  {$\wt_1$};
  \node[below = of w1] (u1) {$\ud_1$};
  \node[below = of u1] (u2) {$\ud_2$};
  \node[right = of u1] (r1) {$\rd_1$};
  \node[below = of r1] (r2) {$\rd_2$};
  \node[right = of r1] (u4) {$\ud_3$};
  \node[below = of u4] (u3) {$\ud_4$};
  \node[above = of u4] (w2) {$\wt_2$};
  \node[right = of u4] (w3) {$\wt_3$};
  \node[right = of u3] (r3) {$\rd_3$};
  
   \draw[po] (r1) to node[right]{$\po$} (r2);
   \draw[po] (w3) to node[right]{$\po$} (r3);
   \draw[rf] (w1) to node[left]{$\rf$} (u1);
   \draw[rf] (u1) to node[left]{$\rf$} (u2);
   \draw[rf] (u1) to node[below]{$\rf$} (r1);
   \draw[rf] (w2) to node[left]{$\rf$} (u4);
   \draw[rf] (u4) to node[left]{$\rf$} (u3);
   \draw[rf] (u3) to node[above]{$\rf$} (r3);
   \draw[rf] (u3) to node[above]{$\rf$} (r2);
   \draw[mo] (w3) to node[above, sloped]{$\mopartial$} (u3);
   \draw[mo] (w3) to node[above, sloped]{$\mopartial$} (w2);
   \draw[mo] (u1) to node[above, sloped]{$\mopartial$} (w2);
   
% %
%   \draw[po] (t21) to (t22);
% %
%   \draw[rf,bend left=0] (t11) to node[right,pos=0.2]{$\rf$} (t22);
%   \draw[fr,bend left=0] (t22) to node[left,pos=0.2]{$\fr$} (t12);
% %    
%   \draw[mo,bend right=0] (t13) to node[left,pos=0.9]{$\mo$} (t21);
%   \draw[rf,bend right=0] (t21) to node[right,pos=0.8]{$\rf$} (t14);
\end{tikzpicture}
\caption{
A minimally coherent $\mopartial$ under $\rcmm$.
All events access the same location.
}
\label{fig:mincoh}
\vspace{-0.5cm}
\end{wrapfigure}

%% file: algorithms/algo-rc20-rf.tex
%!TEX root = ../main.tex

\begin{algorithm*}[t]
\Input{Events $\E$, program order $\po$ and reads-from relation $\rf$}
\BlankLine
\lIf{$(\po \cup \rf)$ is cyclic or $\rf$ violates weak-atomicity}{
    \declare `Inconsistent' \label{line:rf-rc20-po-rf-acyclic}
} 
% $\HB \gets$ \getHBTS{$\E$, $\po$, $\rf$} \label{line:rf-rc20-hb-array-init}\;
\Let $\HB$ be an $\E$-indexed array storing the $\hb$-timestamps of events \label{line:rf-rc20-hb-array-init} \;
% $(\LRB, \LWB) \gets$ \getLRWB{$\E$, $\po$} \label{line:rf-rc20-lwb-array-init} \;
\Let $\set{\view{\WtLst}{t,x}{u}}_{t, x, u}$ and 
$\set{\view{\RdLst}{t,x}{u}}_{t, x, u}$ be data structures implementing $\glw()$ and $\glr()$ \label{line:rf-rc20-lwb-array-init} \;
% $(\set{\view{\WtLst}{t,x}{u}}_{t, x, u},\set{\view{\RdLst}{t,x}{u}}_{t, x, u}) \gets$ \initWtRdLst{$\E$, $\po$} \label{line:rf-rc20-lwb-array-init} \;
\Let $\TC$ and $\PC$ be $\E$-indexed arrays denoting the top and position of events in their $\rf$-chains \label{line:rf-rc20-tc-pc-array-init}\;
% $(\TC, \PC) \gets$ \getTCPC{$\E$, $\rf$}  \label{line:rf-rc20-tc-pc-array-init} \;
\lForEach{$x \in \E.locs$}{
    $\mopartial_{x} \gets \emptyset$;
}
\ForEach{$\event \in \E$ in $\po$-order}{
\label{line:rf-rc20-mo-infer-start}
    \Case{$\event = \rd(t, x)$ or $\event = \ud(t, x)$}{
        \Let $\wt_{\mathsf{rf}} = \rfinv(\event)$ \; %, $t_{\mathsf{rf}} = \wt_{\mathsf{rf}}.\tid$\\% and $c_{\mathsf{rf}} = \HB[\wt_{\mathsf{rf}}][t_{\mathsf{rf}}]$ \;
        \ForEach{$u \in \E.tids$}{
            % \Let $c_u = (\event.\op = \ud \land u = t)$ ? $\HB[\event][u] - 1$ : $\HB[\event][u]$ \;
            \Let $c_u$ = \lIf*{$(\event.\op = \ud \land u = t)$}{$\HB[\event][u] - 1$} \lElse*{$\HB[\event][u]$} \;
            % \Let $\wt_u^\wt = \LWB[u][x][c_u]$ and \Let $\wt_u^\rd = \rf(\LRB[u][x][c_u])$ \label{line:rf-rc20-latest-writes} \;
            \Let $\wt_u^\wt = \apicall{\view{\WtLst}{u,x}{t}}{get}{$c_u$}$ and \Let $\wt_u^\rd = \rfinv(\apicall{\view{\RdLst}{u,x}{t}}{get}{$c_u$})$ \label{line:rf-rc20-latest-writes} \;
            \For{$\wt_u \in \set{\wt_u^\wt, \wt_u^\rd}$}{ \label{line:rf-rc20-loop-glr-glw}
                \uIf{$(\TC[\wt_{\mathsf{rf}}] \neq \TC[\wt_u])$ or $(\PC[\wt_{\mathsf{rf}}] < \PC[\wt_u])$}{ \label{line:rf-rc20-check-chain-condition}
                    $\mopartial_{x} \gets \mopartial_{x} \cup \set{(\wt_u, \TC[\wt_{\mathsf{rf}}])}$ \label{line:rf-rc20-update-mo}
                }
            }
        }
    }
}
% \lIf{$(\hb \cup \bigcup_{x \in \E.locs}\mo_{\partialtext, x})$ is cyclic}{\declare `Inconsistent' \linelabel{rf-rc20-cycle}}
% \lIf{$\hb;(\bigcup_{x \in \E.locs}\mo_{\partialtext, x})$ is cyclic}{\declare `Inconsistent' \linelabel{rf-rc20-cycle}}
\ForEach{$x \in \E.locs$}{
    \lIf{$(\rf_x \cup \hb_x \cup \mopartial_{x})$ is cyclic}{\declare `Inconsistent' \label{line:rf-rc20-cycle}}
}
\declare `Consistent'
\caption{
\label{algo:rf-consistency-rc20}
Checking consistency for $\rcmm$.
}
%\hunkar{We never seem to use $c_\rf$ and $t_\rf$ in the algorithm}
\end{algorithm*}
\normalsize

%% file: figures/min-coh-relaxed-rc20.tex
%!TEX root = ../main.tex

\begin{figure}[t!]
\centering
\def\ystep{0.5}

\begin{subfigure}[b]{0.45\textwidth}
\centering
\scalebox{0.9}{
\begin{tikzpicture}[node distance=4.25mm and 7mm]
  \node (w1) at (0,0) {$\wt_1(x)$};
  \node[below = of w1] (rmw1) {$\ud_1(x)$};
  \node[below = of rmw1] (rmw2) {$\ud_2(x)$};

  \node[right = of w1] (w2) {$\wt_2(x)$};

  \node[right = of w2] (w3) {$\wt_3(x)$};
  \node[below = of w3] (w4) {$\wt_4(y)$};

  \node[right = of w3] (r1) {$\rd_1(y)$};
  \node[below = of r1] (r2) {$\rd_2(x)$};
  \node[below = of r2] (r3) {$\rd_3(x)$};

   \draw[po] (w3) to node[right]{$\po$} (w4);

   \draw[po] (r1) to node[right]{$\po$} (r2);
   \draw[po] (r2) to node[right]{} (r3);

   \draw[rf] (w4) to node[below, sloped]{$\rf$} (r1);
   \draw[rf,bend right=15] (rmw1) to node[below,pos=0.45]{$\rf$} (r2);
   \draw[bend right, rf] (w2) to node[below, sloped,pos=0.7]{$\rf$} (r3);

    \draw[rf] (w1) to node[left]{$\rf$} (rmw1);
    \draw[rf] (rmw1) to node[left]{$\rf$} (rmw2);

   \draw[mo] (w3) to node[below]{$\mopartial$} (w2);
   \draw[mo,bend right=20] (w3) to node[above, sloped]{$\mopartial$} (w1);
   \draw[mo] (rmw1) to node[below, sloped,pos=0.6]{$\mopartial$} (w2);

\end{tikzpicture}
}
\caption{
All accesses are $\MOrel$/$\MOacq$.
}
\label{fig:mincoh-rc20-rlx-a}
\end{subfigure}
\hfill
\begin{subfigure}[b]{0.45\textwidth}
\centering
\scalebox{0.9}{
\begin{tikzpicture}[node distance=4.25mm and 7mm]
  \node (w1) at (0,0) {$\wt_1(x)$};
  \node[below = of w1] (rmw1) {$\ud_1(x)$};
  \node[below = of rmw1] (rmw2) {$\ud_2(x)$};

  \node[right = of w1] (w2) {$\wt_2(x)$};

  \node[right = of w2] (w3) {$\wt_3(x)$};
  \node[below = of w3] (w4) {$\wt_4(y)$};

  \node[right = of w3] (r1) {$\rd_1(y)$};
  \node[below = of r1] (r2) {$\rd_2(x)$};
  \node[below = of r2] (r3) {$\rd_3(x)$};

   \draw[po] (w3) to node[right]{$\po$} (w4);

   \draw[po] (r1) to node[right]{$\po$} (r2);
   \draw[po] (r2) to node[right]{} (r3);

   \draw[rf] (w4) to node[below, sloped]{$\rf$} (r1);
   \draw[rf,bend right=15] (rmw1) to node[below,pos=0.45]{$\rf$} (r2);
   \draw[bend right, rf] (w2) to node[below, sloped,pos=0.7]{$\rf$} (r3);

    \draw[rf] (w1) to node[left]{$\rf$} (rmw1);
    \draw[rf] (rmw1) to node[left]{$\rf$} (rmw2);

   %\draw[mo] (w2) to node[above]{$\mopartial$} (w1);
   %\draw[mo] (w2) to node[above]{$\mopartial$} (w3);
   \draw[mo] (rmw1) to node[below, sloped,pos=0.6]{$\mopartial$} (w2);
\end{tikzpicture}
}
\caption{
All accesses are $\MOrlx$.
}
\label{fig:mincoh-rc20-rlx-b}
\end{subfigure}
\caption{
A minimally coherent $\mopartial$ computed by \cref{algo:rf-consistency-rc20} (\subref{fig:mincoh-rc20-rlx-a}) and \cref{algo:rf-consistency-relaxed} (\subref{fig:mincoh-rc20-rlx-b}).
}
\label{fig:mincoh-rc20-rlx}
\end{figure}

%% file: rf_rlx.tex
%!TEX root=./main.tex

\subsection{Consistency Checking for $\rlxmm$}\label{subsec:rf_rlx}
We now turn our attention to the $\rlxmm$ fragment.
As a strict subset of $\rcmm$ (where $\hb=\po$), consistency checking for this model can be performed in $O(\NumEvents\cdot\NumThreads)$ time by \cref{thm:rc_rf_upper}.
Although this bound is nearly linear time, here we show that the $\rlxmm$ fragment
enjoys a \emph{truly} linear time consistency checking, independent of $\NumThreads$ (\cref{thm:rlx_rf_upper}).
This improvement is based on two insights.

As $\rf$ edges do not induce any synchronization in this fragment, our first insight is that the input partial execution $\expartial=(\E, \po, \rf)$ can be partitioned into separate executions $\locx{\expartial}=(\locx{\E}, \locx{\po}, \locx{\rf})$, one for each location $x\in \E.locs$.
Indeed, we have $\expartial\models \rlxmm$ iff $\acy(\po \cup \rf)$ and $\locx{\expartial}\models \rlxmm$ for each $x\in \E.locs$.
Thus, without loss of generality, we may assume that $\expartial$ consists of a single location.
% This reduces the complexity of \cref{algo:rf-consistency-rc20} from $O(\NumEvents\cdot(\NumThreads+\NumMemory))$ to $O(\NumEvents\cdot\NumThreads)$.
Our second insight comes from the simplified formulation of minimal coherence under $\rlxmm$.

\bparagraph{Minimal coherence under $\rlxmm$}
Let us revisit the concept of minimal coherence under $\rcmm$.
Focusing on the $\rlxmm$ fragment, we have $\hb=\po$.
Thus, the first and third conditions of minimal coherence are reduced to the following.
\begin{compactenum}[label=(\arabic*$'$)]
\item \label{item:minimal_coherence_rlx1} 
For every triplet $(\wt, \rd, \wt')$ accessing location $x$, if $(\wt', \rd)\in \rf^?;\po$ and $(\wt', \wt)\not \in \rf^+$, we have $(\wt', \TC[\wt])\in (\rf_x\cup \po_x \cup \mopartial_x)^+$.
\addtocounter{compactenumi}{1}
\item \label{item:minimal_coherence_rlx3} 
$(\rf_x\cup\po_x\cup \mopartial_x)$ is acyclic, for each $x\in\E.locs$.
\end{compactenum}
Similarly to $\rcmm$, \cref{fig:mincoh} also serves as an illustration of the above definition.
The key insight towards a truly linear-time algorithm is as follows.
Consider the execution of \cref{algo:rf-consistency-rc20} on a partial execution $\expartial$ under $\rlxmm$ semantics.
Further, consider a read/RMW event $\event$ processed by the algorithm with $\wt_{\mathsf{rf}}=\rfinv(\event)$.
Among all events $\wt'$ forming a triplet $(\wt_{\mathsf{rf}}, \event, \wt')$ and such that $(\wt', \event)\in \rf^?;\po$, there exists one that is $(\rf\cup\po\cup \mopartial_x)^+$-maximal.
In particular, if the immediate $\po$-predecessor of $\event$ is a read event $\rd$, then the event $\rfinv(\rd)$ is this maximal $\wt'$.
Otherwise, the immediate $\po$-predecessor of $\event$ is a write/RMW event $\wt''$, which is also the maximal $\wt'$.
Thus, it suffices to keep track of this information on-the-fly,
and only insert $(\wt', \TC[\wt_{\mathsf{rf}}])$ in $\mopartial_x$, if necessary, to make $\mopartial_x$ minimally-coherent.
As we now do not have to compute $\HB$-timestamps or iterate over all threads during the processing of $\event$, we have a truly linear-time algorithm.
%\hunkar{
%These concepts are illustrated in \cref{fig:mincoh-relaxed}.
%Moreover, the utilization of these concepts in this example results in a $\mopartial$ that is weaker than of \cref{algo:rf-consistency-rc20}.
%In particular, \cref{algo:rf-consistency-rc20} would strengthen this $\mopartial$ by including the orderings $(\wt_1, \wt_4)$ and $(\ud_1, \wt_4)$.
%}

%\Andreas{Add a reference to previous RC20 figure}

\input{algorithms/algo-relaxed-rf}

\bparagraph{Algorithm}
The above insights are turned into an algorithm in \cref{algo:rf-consistency-relaxed}.
The algorithm first verifies that $(\po\cup \rf)$ is acyclic and $\rf$ satisfies weak-atomicity (\cref{line:rf-relaxed-hb-acyclic}).
%However, the algorithm next does not compute $\HB$ timestamps, but only the $(\TC, \PC)$ arrays.
Then, it performs a separate pass for each location $x$ and constructs the minimally coherent $\mopartial_x$.
To this end, it keeps track in $\LW_{t,x}$ the unique $(\rf_x\cup\po_x\cup \mopartial_x)^+$-maximal write/RMW event that has an $\rf^?;\po$ path to the current event of thread $t$.
When a read/RMW event $\event$ is processed, the algorithm potentially updates $\mopartial_x$ with an ordering $(\LW_{t,x}, \TC[\wt_{\mathsf{rf}}])$ (\cref{line:rf-relaxed-check-read}), using the same condition as in \cref{algo:rf-consistency-rc20}.
%, as opposed to computing the events $\wt_u^\wt$ and $\wt_u^\rd$ that \cref{algo:rf-consistency-rc20} does.
%The condition for inserting the ordering remains the 
\cref{fig:mincoh-rc20-rlx} contrasts the $\mopartial$ computed by \cref{algo:rf-consistency-relaxed} to the $\mopartial$ computed by \cref{algo:rf-consistency-rc20} on the same partial execution but with different access levels.
We arrive at the following theorem.

\thmrlxrfupper*

%% file: algorithms/algo-relaxed-rf.tex
%!TEX root = main.tex

%%%%%%%%%%%%%%%%%%%%%%%%%%%%%%%%%%%%%%%%%%%%%%%%%%%%%
% Algorithm1(X, po, rf)
%     {G_x}_{x \in Vars} = split(G) // add immediate po edges between two events of same var
%     for each x:
%       let V_x, E_x = G_x.nodes, G_x.egdes
%       for events e in topological ordering of G_x
%           if e.type = w:
%                LW_t := we
%           if e.type = r:
%               e_rf := rf(e)
%               add mo edge from LW_t to e_rf    
%       check if the resulting graph has a cycle.
%%%%%%%%%%%%%%%%%%%%%%%%%%%%%%%%%%%%%%%%%%%%%%%%%%%%%
% \small
\begin{algorithm*}[t]
\Input{Events $\E$, program order $\po$ and reads-from relation $\rf$}
\BlankLine
\lIf{$(\po \cup \rf)$ is cyclic or $\rf$ violates weak-atomicity}{
    \declare `Inconsistent' \label{line:rf-relaxed-hb-acyclic}
} 
\Let $\TC$ and $\PC$ be $\E$-indexed arrays denoting the top and position of events in their $\rf$-chains \;
% $(\TC, \PC) \gets$ \getTCPC{$\E$, $\rf$} \;
\ForEach{$x \in \E.locs$}{
    \lForEach{$t \in \locx{\E}.tids$}{ \label{line:rf-relaxed-lw-init}
        $\LW_{t,x} \gets \nil$
    }
    $\mopartial_{x} \gets \emptyset$ \;
    \ForEach{$\event \in \locx{\E}$ in $(\po_x \cup \rf_x)$-order}{\label{line:rf-rlx-mo-infer-start}
        \Case{$\event = \wt(t, x)$}{
            $\LW_{t,x} \gets \event$ \; \label{line:rf-relaxed-write}
        }
        \Case{$\event = \rd(t, x)$}{
            \If{$(\TC[\rfinv_x(\event)] \neq \TC[\LW_{t,x}])$ or $(\PC[\rfinv_x(\event)] < \PC[\LW_{t,x}])$}{ \label{line:rf-relaxed-check-read}
                $\mopartial_{x} \gets \mopartial_{x} \cup \set{(\LW_{t,x}, \TC[\rfinv_x(\event)])}$ \label{line:rf-relaxed-updated-mo-read}
            }
            $\LW_{t,x} \gets \rfinv(\event)$ \;
        }
        \Case{$\event = \ud(t, x)$}{
            Execute \crefrange{line:rf-relaxed-check-read}{line:rf-relaxed-updated-mo-read} followed by \cref{line:rf-relaxed-write} \label{line:rf-relaxed-end} \;
        }
    }
    \lIf{$(\po_x \cup \rf_x \cup \mopartial_{x})$ is cyclic}{\declare `Inconsistent'  \label{line:rf-relaxed-mo-cycle}} 
}
\declare `Consistent'
% \vspace*{0.25cm}
\caption{Checking consistency for $\rlxmm$.}
\label{algo:rf-consistency-relaxed}
\end{algorithm*}
\normalsize

%% file: rf_lower.tex
\subsection{A Super-Linear Lower Bound for RMW-Free $\ramm$, $\wramm$, and $\sramm$}\label{subsec:rf_lower}

Finally, we address the existence of a truly linear-time algorithm for any model other than $\rlxmm$.
We show that this is unlikely, by proving the two lower bounds of \cref{thm:rf_lower}.
The proof is via a \emph{fine-grained} reduction from the problem of checking triangle freeness in undirected graphs, which suffers the same lower bounds. That is, there is no algorithm (resp. combinatorial algorithm, under the BMM hypothesis) for checking triangle-freeness in time $O(n^{\omega/2 - \epsilon})$ (resp. $O(n^{3/2 - \epsilon})$) for any fixed $\epsilon > 0$~\cite{Williams2018}, where $n$ is the number of nodes in the graph.

\newcommand{\varN}{y}
\newcommand{\varE}{x}
\newcommand{\varH}{z}
\newcommand{\juncEvRd}{\sf\rd Jxn}
\newcommand{\juncEvWt}{\sf\wt Jxn}
\newcommand{\hbEv}{\sf H}
\newcommand{\conn}[2]{{#1}\to{#2}}

\bparagraph{Reduction}
Given a graph $G = (V_G, E_G)$ of $n$ vertices, we construct an RMW-free partial execution $\expartial = \tup{\E, \po, \rf}$ with $|\E|=O(n)$
such that $\expartial$ is consistent with any of $\ramm$,
$\wramm$ and $\sramm$ iff $G$ is triangle-free.
For simplicity, we let $V_G=\{1,\dots, n\}$.

\eparagraph{Events and memory locations}
We start with the event set $\E$.
For the moment, all events belong to different threads, while we only define the memory location of an event when relevant.
For every node $\alpha \in V_G$, $\expartial$ contains
(i)~a location $\varN_\alpha$ and a write event $\wt_{\alpha}(\varN_\alpha)$ and
(ii)~auxiliary ``junction'' events $\juncEvRd_{\alpha}$ and $\juncEvWt_{\alpha}$ on fresh locations.
For every edge $(\alpha, \beta) \in E_G$ with $\alpha < \beta$, $\expartial$ contains
(i)~an event $e_{(\alpha, \beta)}$ that accesses a fresh location,
(ii)~a read event $\rd^{\alpha}_{\beta}(\varN_{\beta})$, and
(iii)~a write event $\wt^{\beta}_{\alpha}(\varN_{\alpha})$.

\eparagraph{Relations $\rf$ and $\hb$}
We now define the $\rf$ relation.
Our construction also makes certain events $\hb$ ordered.
This can be done trivially by introducing auxiliary events with an $\rf$ relation between them, while $\expartial$ remains of size $O(n)$.
In particular, every $(\event_1, \event_2)\in\hb$ edge can be simulated using fresh events $\rd, \wt$ such that 
(i)~$(\wt, \rd) \in \rf$
(ii)~$(\event_1,\wt)\in \po$, and
(iii)~$(\rd, \event_2)\in \po$.
For simplicity of presentation, we do not mention these events explicitly, but rather directly the $\hb$ relation they result in.
For every edge $(\alpha, \beta) \in E_G$ with $\alpha < \beta$, we have the following relations:
\[
(\wt_{\beta}, \rd^{\alpha}_{\beta}) \in \rf
\qquad \qquad\quad\quad
(\wt_{\alpha}, \wt^{\beta}_{\alpha}) \in \hb
\qquad\qquad\quad\quad  
(\wt^{\beta}_{\alpha}, \juncEvWt_{\beta}) \in \hb
\]
\[
(\juncEvWt_{\beta}, e_{(\alpha, \beta)}) \in \hb 
\qquad\qquad 
(e_{(\alpha, \beta)}, \juncEvRd_{\alpha}) \in \hb
\qquad\qquad
(\juncEvRd_{\alpha}, \rd^{\alpha}_{\beta}) \in \hb
\]

\input{figures/super-linear-lower-bound}

\cref{fig:super-linear-lower-bound} illustrates the above construction for a slice of the constructed partial execution $\expartial$.
We conclude with a proof sketch of \cref{thm:rf_lower}, and refer to\begin{arxiv}~\cref{subsec:app_rf_lower}\end{arxiv}\begin{pldi}~\cite{arxiv}\end{pldi} for the full proof.

\bparagraph{Correctness and time complexity}
If there is a triangle $(\alpha, \beta, \gamma)$ in $G$ with $\alpha<\beta,\gamma$, then $(\wt^\gamma_{\beta}, \rd^{\alpha}_{\beta}) \in \hb$
because of the sequence of $\hb$ edges:
$(\wt^{\gamma}_{\beta}, \juncEvWt_{\gamma})$,
$(\juncEvWt_{\gamma}, e_{(\alpha, \gamma)})$,
$(e_{(\alpha, \gamma)}, \juncEvRd_{\alpha})$,
$(\juncEvRd_{\alpha}, \rd^{\alpha}_{\beta})$.
Together with $(\wt_{\beta}, \rd^{\alpha}_{\beta}) \in \rf$ and $(\wt_{\beta}, \wt^{\gamma}_{\beta}) \in \hb$
by construction, we obtain a weak-read-coherence violation.
In the other direction, if there are no triangles in $G$, then the modification order
$\mo = \bigcup_{\alpha \in V_G} \mo_{y_\alpha}$ where
$\mo_{y_\alpha}$ orders $\wt_{\alpha}$ before every other write $\wt^{\beta}_{\alpha}$
on $\varN_\alpha$, makes $\ex = \tup{\expartial.\E, \expartial.\po, \expartial.\rf, \mo}$ $\sramm$-
(and thus also $\ramm$- and $\wramm$-) consistent.
Such an $\mo$ ensures that $(\hb \cup \mo)$ is acyclic.
Triangle-freeness ensures read-coherence --- 
a violation of read coherence implies that there are three events
$e_1 = \wt_{\beta}, e_2 = \wt^{\gamma}_{\beta}, e_3 = \rd^{\alpha}_{\beta}$
such that $(e_1, e_3) \in \rf$, $(e_1, e_2) \in \mo$
and $(e_2, e_3) \in \hb$, implying a triangle $(\alpha, \beta, \gamma)$ in $G$.
The total time to construct $\expartial$ is $O(|V_G| + |E_G|)$. Further, our reduction is completely combinatorial.
We thus arrive at the following theorem.

\thmrflower*
% \hunkar{I think we never state what BMM corresponds to.}
% \ucomment{Correct. We use triangle free-ness for the reduction. We have to say that under the BMM hypothesis (which is so and so), triangle freeness cannot be solved in sublinear time. This means that consistency checking also has a super linear lower bound under the same BMM hypothesis. Change in Section 1 contributions}

%% file: figures/super-linear-lower-bound.tex
%!TEX root=../main.tex

\begin{figure}
\scalebox{0.9}{
\begin{tikzpicture}[thick,
pre/.style={<-,shorten >= 2pt, shorten <=2pt, very thick},
post/.style={->,shorten >= 3pt, shorten <=3pt,   thick},
seqtrace/.style={line width=2},
und/.style={very thick, draw=gray},
event/.style={rectangle, minimum height=5mm, minimum width=8mm,  line width=0.5pt, inner sep=1},
virt/.style={circle,draw=black!50,fill=black!20, opacity=0}]
%\footnotesize

\begin{scope}[shift={(-8,-1)},scale=0.7]
\node[circle,draw,inner sep=2pt,minimum size=4pt] (one) at (0, 0) {$1$};
\node[circle,draw,inner sep=2pt,minimum size=4pt] (two) at (1, -1.2) {$2$};
\node[circle,draw,inner sep=2pt,minimum size=4pt] (three) at (2, 0) {$3$};
\path (one) edge (two);
\path (one) edge (three);
\path (two) edge (three);
\end{scope}

\newcommand{\xstep}{1.7}
\newcommand{\ystep}{0.7}
\newcommand{\yaux}{0.17727}

\fill[gray!20]
(-2.5, 0.5*\ystep) to
(-2.5, -2.5*\ystep) to
(2.5, -2.5*\ystep) to
(2.5, 0.5*\ystep) to cycle;

\node[event] (e13) at (0, 0*\ystep) {$e_{(1,3)}$};

\node[event] (e1r) at (-2, -0.5*\ystep) {$\juncEvRd_{1}$};
\node[event] (e3w) at (2, -0.5*\ystep) {$\juncEvWt_{3}$};

\node[event] (w2) at (0, -1.5*\ystep) {$\wt_{2}$};
\node[event] (r1y2) at (-2, -2*\ystep) {$\rd^{1}_{2}$};
\node[event] (w3y2) at (2, -2*\ystep) {$\wt^{3}_{2}$};

% =============

\node[event] (w3y1) at (4, -2*\ystep) {$\wt^{3}_{1}$};
\node[event] (e23) at (4, 0*\ystep) {$e_{(2,3)}$};
\node[event] (r1y3) at (-4, -2*\ystep) {$\rd^{1}_{3}$};
\node[event] (e12) at (-4, 0*\ystep) {$e_{(1,2)}$};

% =============

\draw [hb] (e13) -- (e1r) node [midway, above, sloped] {$\hb$};
\draw [hb] (e3w) -- (e13) node [midway, above, sloped] {$\hb$};
\draw [hb] (e1r) -- (r1y2) node [midway, right] {$\hb$};
\draw [rf] (w2) -- (r1y2) node [midway, above, sloped] {$\rf$};
\draw [hb] (w2) -- (w3y2) node [midway, above, sloped] {$\hb$};
\draw [hb] (w3y2) -- (e3w) node [midway, left] {$\hb$};

% =============

\draw [hb] (w3y1) -- (e3w) node [midway, below, sloped] {$\hb$};
\draw [hb] (e3w) -- (e23) node [midway, above, sloped] {$\hb$};
\draw [hb] (e1r) -- (r1y3) node [midway, below, sloped] {$\hb$};
\draw [hb] (e12) -- (e1r) node [midway, above, sloped] {$\hb$};

% =============

% \draw [rf] (f1) -- (f2) node [midway, fill=white, below] {$\rf$};
% \draw [rf, dotted] (f2) -- (f3) node [midway, fill=white, below] {$\rf^*$};
% \draw [rf] (f3) -- (e3) node [midway, fill=white, below] {$\rf$};
% \path[->] (e2)  edge [bend right=20]  node[below, align=left, sloped]  {$\rf^?;\hb$}         (e3);
% \path[->,color=colorMO] (e2)  edge [bend left=20, double]  node[below, align=left, sloped]  {$\mo$}         (e1);
\end{tikzpicture}
}
\caption{
\emph{Left:} A graph $G$ with three nodes $V_G = \set{1,2,3}$ containing a triangle. 
\emph{Right:} A slice of the partial execution $\expartial$ for $G$.
We have $(\wt_{2}, \wt^{3}_{2})\in \hb$ and $(\wt^{3}_{2},\rd^{1}_{2})\in \hb$,
thus violating weak read coherence in $\expartial$.
\label{fig:super-linear-lower-bound}
}
\end{figure}

%% file: experiments-new.tex
%!TEX root=./main.tex

\section{Experimental Evaluation}\label{sec:experiments}
We implemented our consistency-checking algorithms for $\ramm$/$\rcmm$ and evaluated their performance on two standard settings of program analysis, namely,
(i)~\emph{stateless model checking}, using \Trust~\cite{Kokologiannakis:2022}, and
(ii)~\emph{online testing}, using \celeventester~\cite{Luo:2021}.
These tools are designed to handle 
variants of C11 
including SC accesses, and performing race-detection, which are beyond the scope of this work.
Here, we focus on the consistency-checking component for the $\ramm$/$\rcmm$ fragment, which is common in these tools and our work.
We conducted our experiments on a machine running Ubuntu 22.04 with 2.4GHz CPU and 64GB of memory.

%The key difference between the two approaches, and the corresponding tools, is as follows.
%Stateless model checking analyzes a program by exhaustively enumerating all executions, up to a certain bound in length.
%As the number of executions can be large, the length of each execution is usually small.
%Nevertheless, this approach offers coverage guarantees of all possible program behaviors as witnessed by executions of the given bound.
%On the other hand, online testing explores randomly generated executions. 
%This allows to explore executions of larger length, but relinquishes coverage guarantees.
%In their core, both techniques manipulate partial executions $\expartial=(\E, \po, \rf)$, together with a partial order on the write/RMW events $\mopartial$.
%Hence, (reads-from) consistency checking is a fundamental problem that each technique has to solve repeatedly, and thus as efficiently as possible.

\bparagraph{Benchmarks} 
We used standard benchmark programs from prior state-of-the-art verification and testing papers~\cite{NorrisDemsky:2013,Abdulla:2018,Luo:2021,Kokologiannakis:cav:2021},
as well as the applications \bname{Silo}, \bname{GDAX}, \bname{Mabain}, and \bname{Iris}~\cite{Luo:2021} for online testing.
%CDSChecker~\cite{NorrisDemsky:2013}, 
%tsan11rec~\cite{Lidbury:2017}, Tracer~\cite{Abdulla:2018}, genmc~\cite{Kokologiannakis:cav:2021}, C11Tester~\cite{Luo:2021}.
%For online testing, we have also used the application benchmarks --  used in~\cite{Luo:2021}.
% and $9$ data structure benchmarks used in \cite{Kokologiannakis:cav:2021}. 
%This resulted in a total of $33$ benchmarks.
These benchmarks use C11 concurrency primitives extensively. 
For thorough evaluation, we have scaled up some of them, 
when their baseline versions were too small, by increasing the number of threads or loop counters.
Some benchmarks also contain accesses outside our scope; we converted those accesses to access modes applicable for our experiments, in line with the evaluation in prior works~\cite{Abdulla:2018,Lahav:2019}.

\input{experiments_offline}

\input{experiments_online}

%% file: experiments_offline.tex
%!TEX root=./main.tex

\subsection{Stateless Model Checking}\label{subsec:experiments_offline}

%We start with the setting of stateless model checking inside \Trust.
The \Trust model checker explores all behaviors of a bounded program by 
enumerating executions, making use of different strategies to avoid redundant exploration.
One such strategy is to enumerate partial executions $\expartial$
%, i.e., 
%without carrying an explicit description of a modification order $\mo$.
and perform a consistency check for the maximal ones, to verify that they represent valid program behavior.
As the number of explored executions is typically large, it is imperative that consistency checks are performed as fast as possible.

\bparagraph{Consistency checking inside \Trust}
The algorithm for consistency checking in \Trust constructs a writes-before order $\wb$~\cite{Lahav:2015}, which is a partial modification order that serves as a witness of consistency.
The time taken to construct $\wb$ is $O(\NumEvents^3)$, which has been identified as a bottleneck in the model checking task~\cite{Kokologiannakis:2019,Kokologiannakis:2022}.
We replaced \Trust's $\wb$ algorithm for consistency checking with the $\mopartial$ computation of \cref{algo:rf-consistency-rc20}, and measured
(i)~the speedup realized for consistency checking, and
(ii)~the effect of this speedup on the overall model-checking task.
%To this end, and with the aim to evaluate (i) and (ii) on many benchmarks,
We executed \Trust on several benchmarks, each with a time budget of $2$ hours, measuring the average time for consistency checking (for evaluating (i)) as well as the total number of executions explored (for evaluating (ii)).
Finally, we note that \Trust employs a number of simpler consistency checks during the exploration.
Although we expect that our new algorithm can improve those as well, 
we have left them intact as it was unclear to us how they interact with the rest of the tool,
% and if changing them will affect soundness of the obtained results.
and in order to maintain soundness of the obtained results.
% Our evaluation results show promise, nevertheless.

\bparagraph{Experimental results}
Our results are shown in \cref{tab:expr-results-smc}.
We mark with $\dagger$ benchmarks on which the model checker found an error and halted early.
We observe that \cref{algo:rf-consistency-rc20} is always faster, typically by a significant margin.
The maximum speedup for consistency checking is $162 \times$, and the geometric mean of speedups is $36 \times$.
Regarding the number of executions, the model checker explores $4.3 \times$ more on (geometric) average, and as high as $71.6 \times$ more, when using~\cref{algo:rf-consistency-rc20}.
In some benchmarks, the two approaches observe a similar number of executions.
This is due to consistency checking being only part of the overall model-checking procedure,
which also consists of other computationally intensive tasks such as backtracking.
As consistency checking appears now to not be a bottleneck, it is meaningful to focus further optimization efforts on these other tasks.
For \bname{ttaslock}, we noticed a livelock that blocks the model checker.
Overall, our experiments highlight that the benefit of the new, nearly linear time property of consistency checking leads to a measurable speedup that positively impacts the overall efficiency of model checking.
We refer to\begin{arxiv}~\cref{subsec:app_experiments_smc}\end{arxiv}\begin{pldi}~\cite{arxiv}\end{pldi} for experiments on $\rcmm$, which lead to the same qualitative conclusions.

\input{tab_experiments4}

%than the speed-ups on consistency checking.
% 
% lead to 
%  exploration time, these large speedups are translate to smaller, but in most cases considerable, improvements in the speed of the overall exploration.
% using the proper algorithm, consistency checking can be performed much more efficiently than previously thought,
% with a clear, and 
% typically significant impact on the speed of model checking.

%% file: tab_experiments4.tex
%!TEX root = main.tex

\begin{table}[H]
	\caption{
		Impact on model checking. 
  		Columns $2$ and $3$ denote the average time (in seconds) spent in consistency checking by resp. \Trust and  our algorithm.
		Columns $5$ and $6$ denote the total number of executions explored by resp. \Trust and  our algorithm.
		Columns $4$ and $7$ denote the respective speedups and ratios.
		%All times are in seconds.
		\label{tab:expr-results-smc}
	}
	\setlength\tabcolsep{0.6pt}
	\renewcommand{\arraystretch}{1.0}
	\centering
	\scalebox{0.925}{
{\small
		\begin{tabular}{|c|c|c|c|c|c|c||c|c|c|c|c|c|c|}
			\hline
% 			1 & 2 & 3 & 4 & 5 & 6 & 7 & 8 & 1 & 2 & 3 & 4 & 5 & 6 & 7 & 8 \\
1 & 2 & 3 & 4 & 5 & 6 & 7 & 1 & 2 & 3 & 4 & 5 & 6 & 7\\
			\hline
% 			\multirow{2}{*}{\textbf{Benchmark}} & 
% 			\multirow{2}{*}{$k$} &
% 			\multirow{2}{*}{$d$} &
% 			\multicolumn{2}{c|}{ {\rule{0pt}{1em} \celeventester} } & 
% 			\multicolumn{2}{c|}{ {\cref{algo:ra-rf-on-the-fly}} } &
% 			\multirow{2}{*}{SpdUp} &
% 			\multirow{2}{*}{\textbf{Benchmark}} & 
% 			\multirow{2}{*}{$k$} &
% 			\multirow{2}{*}{$d$} &
% 			\multicolumn{2}{c|}{ {\rule{0pt}{1em} \celeventester} } & 
% 			\multicolumn{2}{c|}{ {\cref{algo:ra-rf-on-the-fly}} } &
% 			\multirow{2}{*}{SpdUp}
% 			\\
            % \hline
			\multirow{2}{*}{\textbf{Benchmark}}& 

			\multicolumn{3}{c|}{ {\textbf{Avg. Time}} }  &
		\multicolumn{3}{c||}{ {\textbf{Executions}} }  &
		
		\multirow{2}{*}{\textbf{Benchmark}}
			&
			 \multicolumn{3}{c|}{ {\textbf{Avg. Time}} }
			& \multicolumn{3}{c|}{ {\textbf{Executions}} }
			
			 \\
			\cline{2-4}
			\cline{5-7}
			\cline{9-11}
			\cline{11-14}
			
			&\textbf{TruSt} & 
		\textbf{Our Alg.} &
		\textbf{S}
		&
		\textbf{TruSt} & 
		\textbf{Our Alg.} &
		\textbf{R} &
		       &
			\textbf{TruSt} &
			\textbf{Our Alg.} &
			\textbf{S} &
			\textbf{TruSt} & 
		\textbf{Our Alg.} &
		\textbf{R}
			\\
% 			\cline{4-7}
% 			\cline{12-15}
% 			 & & & \textsf{P} & Total & \textsf{P} & Total &  &
% 			 & & & \textsf{P} & Total & \textsf{P} & Total & 
% 			\\
			\hline
\rule{0pt}{1em} 
barrier & 0.5  & 0.01 & 45.0  & 14K  & 81K  & 5.98  & ms-queue$\dagger$ & 0.07 & 0.01 & 7.0  & 208  & 208  & 1\\
buf-ring & 1.6  & 0.02 & 79.5  & 3K  & 6K  & 2.52  & mutex & 0.8  & 0.01 & 78.0  & 7K  & 218K  & 29.67 \\
chase-lev$\dagger$ & 0.0 & 0.0 & -  & 2  & 2  & 1 & peterson & 0.3  & 0.01 & 26.0  & 4K  & 5K  & 1.16 \\
control-flow & 0.08 & 0.01 & 8.0  & 88K  & 1M  & 16.56  & qu & 0.1  & 0.01 & 12.0  & 2K  & 2K  & 1.04 \\
dekker & 0.2  & 0.01 & 24.0  & 26K  & 64K  & 2.45  & seqlock & 0.2  & 0.01 & 24.0  & 24K  & 171K  & 7.02 \\
dq & 0.1  & 0.01 & 13.0  & 40K  & 163K  & 4.11  & sigma & 0.2  & 0.01 & 22.0  & 4K  & 62K  & 5.97\\
exp-bug & 0.1  & 0.01 & 13.0  & 55K  & 2M  & 29.98  & spinlock & 1.0  & 0.01 & 98.0  & 5K  & 31K  & 6.04 \\
fib-bench & 1.6  & 0.01 & 162.0  & 4K  & 315K  & 71.63  & szymanski & 0.4  & 0.01 & 44.0  & 2K  & 2K  & 1.16 \\
gcd & 0.3  & 0.01 & 31.0  & 18K  & 80K  & 4.46  & ticketlock & 0.9  & 0.02 & 45.5  & 7K  & 89K  & 12.47 \\
lamport & 0.4  & 0.01 & 36.0  & 17K  & 93K  & 5.54  & treiber & 0.6  & 0.01 & 58.0  & 403  & 403  & 1\\
linuxrwlocks & 0.4  & 0.01 & 40.0  & 12K  & 32K  & 2.75  & ttaslock & 1.0  & 0.01 & 97.0  & 401  & 401  & 1\\
mcs-spinlock & 1.1  & 0.01 & 108.0  & 5K  & 38K  & 7.19  & twalock & 0.5  & 0.01 & 52.0  & 8K  & 30K  & 3.58 \\
mpmc & 0.7  & 0.01 & 66.0  & 10K  & 77K  & 8.01  &  &   &   &  &   &   & \\
\hline
\hline
\textbf{Totals} & - & - & - & - & - & - & - & \textbf{14.5} & \textbf{0.26} & - & \textbf{356K} & \textbf{4.6M} & - \\
\hline
%\textbf{Totals} & - & - & - & - & - & - & - & \textbf{196M} & - & - & \textbf{286.4} & \textbf{170.8} & -\\
% \hline\hline\textbf{Total} & \textbf{58K} & \textbf{115M} & \textbf{47M} & \textbf{35M} & \textbf{110.8} & \textbf{118.9} & \textbf{22.5} & \textbf{34.5} & \textbf{286.6} & \textbf{95.8} & \textbf{18.9} & \textbf{8.0} & \textbf{27.6} & \textbf{171.1} & \textbf{}\\
			% \hline
		\end{tabular}
}  
	}
\end{table}

%% file: experiments_online.tex
%!TEX root=./main.tex

\subsection{Online Testing}\label{subsec:experiments_online}

We now turn our attention to the online testing setting using \celeventester's framework.
%The tool is designed to handle a more general C11 
%\hunkar{C11 model?} \scomment{yes} 
%memory model, including SC accesses and race-detection on non-atomics, which are beyond the scope of this work.
%Instead, here we focus on the consistency-checking component for the $\ramm$ fragment.
In \celeventester's setting, a partial execution $\expartial$ is constructed incrementally, by iteratively
(i)~revealing a randomly chosen new read/RMW event $\rd$,
(ii)~choosing a valid writer $\rf(\rd)$, and
(iii)~continuing the execution of the program until the next read/RMW events.
Hence, every iteration requires a consistency check.
Although we could use \cref{algo:rf-consistency-rc20} from scratch at each step, this would result in unnecessary recomputations of $\mopartial$.
Instead, we follow a different approach here --- we maintain $\mopartial$ on-the-fly, in a way that incremental consistency checks can be done more efficiently.

\bparagraph{Incremental consistency checking}
Our incremental algorithm constructs a similar minimally coherent partial modification order $\mopartial$ as our offline algorithm (\cref{algo:rf-consistency-rc20}).
However, unlike the offline setting, 
we need efficient incremental consistency checks.
For this, we maintain a per-location order 
$\hbmo_x$ on write/RMW events that satisfies following invariants:
(i)~$\hbmo_x\subseteq (\hb_x \cup \mopartial_x)^+$ and
(ii)~$(\hb_x\cup (\mopartial_x;\po_x^?))\subseteq \hbmo_x$.
In order to decide whether a new read/RMW event $\rd(x)$ can observe a write/RMW event $\wt(x)$, 
we must determine if there exists another write/RMW event $\wt'(x)$ such that
$(\wt', \rd) \in \hb_x$ and $(\wt', \wt) \in (\hb_x \cup \mopartial_x)^+$,
as this would lead to a consistency violation.
Using $\hbmo_x$, this check is performed as follows:
(a)~for each thread $u$, we identify the $\po$-maximal write/RMW event $\wt'(u,x)$ for which $(\wt', \rd) \in \hb$, and
(b)~we update 
$
\hbmo_x \gets \hbmo_x\cup \setpred{(\wt'', \wt')}{(\wt'', \wt') \in \hbmo_x^+}
$.
Due to invariant (ii), at this point we are guaranteed that, for all write events $\wt''$, we have $(\wt'', \wt')\in \hbmo_x$ iff $(\wt'', \wt')\in (\hb_x\cup \mopartial_x)^+$.
We can now test whether $\wt$ is a valid writer for $\rd$ by checking whether $(\wt,\wt')\in \hbmo_x$, for one of the aforementioned write/RMW events $\wt'$.
%\hunkar{if space permits maybe we could also mention that in practice this check can also be performed before computing the transitive closure as well as we may be able to discover certain inconsistent executions even without performing the transitive closure.}
We refer to\begin{arxiv}~\cref{subsec:app_experiments_c11}\end{arxiv}\begin{pldi}~\cite{arxiv}\end{pldi} for implementation details.

\bparagraph{Main differences with \celeventester}
The consistency-checking algorithm implemented inside \celeventester also infers $\mo$ orderings as implied by read and write coherence.
The two key differences between that approach and our incremental algorithm described above are the following:
(i)~\celeventester's $\mo$ is stronger than our minimally coherent $\mopartial$ that is contained in $\hbmo$, and
(ii)~this $\mo$ is always maintained transitively-closed.
These two differences are expected to make $\hbmo$ computationally cheaper to maintain than \celeventester's $\mo$.
Although we also have to compute transitive paths $\hbmo^+$ when encountering read/RMW operations (step (b) above), in our experience, these paths typically touch a small part of the input, leading to an efficient computation.

\input{tab_experiments2}

\bparagraph{Experimental results}
Our results are shown in \cref{tab:expr-results-c11}.
For robust measurements, we report averages over $10$ executions per benchmark, focusing on benchmarks for which at least one algorithm took $\geq 1$s.
Our approach achieves a maximum speedup of $104.2 \times$ 
and a geometric speed-up of $2 \times$.
In more detail, we observe significant improvement in the first $5$ benchmarks, and consistent speedups of at least $1.2\times$ on $18$ benchmarks.
We have encountered only 2 benchmarks, \bname{gcd} and \bname{szymanski}, on which the new algorithm is arguably slower.
These benchmarks contain no RMWs and only a small number of write events.
This results in the computation of very small modification orders, diminishing the benefit of our algorithm and results in a marginal slowdown.

%% file: tab_experiments2.tex
%!TEX root = main.tex

\begin{table}[H]
	\caption{
		Impact on online testing. 
		Columns $2$, $3$ and $4$ give the average number of events, threads and locations in each benchmark. 
		Columns $5$ and $6$ denote the average times in seconds to check for consistency by resp. \celeventester and our algorithm.
		Column 6 denotes the speedup.
		\label{tab:expr-results-c11}
	}
	\setlength\tabcolsep{0.7pt}
	\renewcommand{\arraystretch}{1.0}
	\centering
	\scalebox{0.95}{
{\small
		\begin{tabular}{|c|c|c|c|c|c|c||c|c|c|c|c|c|c|c|c|c|c|}
			\hline
% 			1 & 2 & 3 & 4 & 5 & 6 & 7 & 8 & 1 & 2 & 3 & 4 & 5 & 6 & 7 & 8 \\
1 & 2 & 3 & 4 & 5 & 6 & 7 & 1 & 2 & 3 & 4 & 5 & 6 & 7\\
			\hline
% 			\multirow{2}{*}{\textbf{Benchmark}} & 
% 			\multirow{2}{*}{$k$} &
% 			\multirow{2}{*}{$d$} &
% 			\multicolumn{2}{c|}{ {\rule{0pt}{1em} \celeventester} } & 
% 			\multicolumn{2}{c|}{ {\cref{algo:ra-rf-on-the-fly}} } &
% 			\multirow{2}{*}{SpdUp} &
% 			\multirow{2}{*}{\textbf{Benchmark}} & 
% 			\multirow{2}{*}{$k$} &
% 			\multirow{2}{*}{$d$} &
% 			\multicolumn{2}{c|}{ {\rule{0pt}{1em} \celeventester} } & 
% 			\multicolumn{2}{c|}{ {\cref{algo:ra-rf-on-the-fly}} } &
% 			\multirow{2}{*}{SpdUp}
% 			\\
            \textbf{Benchmark} & 
            $\mathbf{n}$ &
			$\mathbf{k}$ &
			$\mathbf{d}$ &
			%\textbf{\celeventester} &
			\textbf{C11Test.} & 
		\textbf{Our Alg.}	%{\cref{algo:ra-rf-on-the-fly}} 
		&
			\textbf{SpeedUp} &
			\textbf{Benchmark} &
			$\mathbf{n}$ &
			$\mathbf{k}$ &
			$\mathbf{d}$ &
			%\textbf{\celeventester} &
			\textbf{C11Test.} & 
		    \textbf{Our Alg.}	%{\cref{algo:ra-rf-on-the-fly}} 
			&
			\textbf{SpeedUp}
			\\
% 			\cline{4-7}
% 			\cline{12-15}
% 			 & & & \textsf{P} & Total & \textsf{P} & Total &  &
% 			 & & & \textsf{P} & Total & \textsf{P} & Total & 
% 			\\
			\hline
\rule{0pt}{1em} 
control-flow & 52K  & 25  & 3  & 4.2  & 0.04 & 104.25 & mutex & 15M  & 11  & 11  & 20.9  & 16.8  & 1.24 \\
sigma & 36K  & 10  & 9  & 2.4  & 0.03 & 80  & gdax & 11M  & 5  & 46K  & 7.0  & 5.7  & 1.23 \\
dq & 599K  & 4  & 2  & 22.1  & 0.5  & 46.78  & spinlock & 5M  & 11  & 10  & 7.0  & 5.9  & 1.18 \\
iris-1 & 1M  & 12  & 45  & 28.4  & 1.8  & 16.25  & ticketlock & 14M  & 6  & 20  & 15.7  & 13.2  & 1.18 \\
seqlock & 478K  & 17  & 20  & 14.4  & 1.4  & 10.38  & ttaslock & 5M  & 11  & 10  & 7.8  & 6.6  & 1.18 \\
exp-bug & 2M  & 4  & 2  & 1.6  & 0.7  & 2.35  & fib-bench & 6M  & 3  & 2  & 3.3  & 2.8  & 1.18 \\
chase-lev & 7M  & 5  & 2  & 12.4  & 5.4  & 2.3  & qu & 1M  & 10  & 29  & 1.6  & 1.3  & 1.18 \\
linuxrwlocks & 7M  & 6  & 10  & 10.9  & 5.9  & 1.84  & treiber & 1M  & 6  & 11  & 1.1  & 1.0  & 1.12 \\
mabain & 5M  & 6  & 18  & 7.3  & 4.0  & 1.82  & silo & 8M  & 4  & 4K  & 5.3  & 4.8  & 1.1 \\
iris-2 & 12M  & 3  & 12  & 9.9  & 5.7  & 1.72  & barrier & 8M  & 5  & 20  & 7.8  & 7.4  & 1.04 \\
mcs-lock & 10M  & 11  & 30  & 17.8  & 12.2  & 1.45  & mpmc & 9M  & 10  & 3  & 12.9  & 12.5  & 1.03\\
lamport & 6M  & 3  & 5  & 3.4  & 2.5  & 1.36  & indexer & 2M  & 17  & 128  & 1.6  & 1.5  & 1.03 \\
peterson & 5M  & 3  & 4  & 2.9  & 2.2  & 1.29  & buf-ring & 5M  & 9  & 12  & 7.0  & 6.8  & 1.02 \\
spsc & 10M  & 3  & 699  & 6.4  & 5.0  & 1.28  & ms-queue & 4M  & 11  & 13  & 8.2  & 8.2  & 0.99 \\
dekker & 16M  & 3  & 3  & 9.2  & 7.2  & 1.27  & gcd & 5M  & 3  & 2  & 2.2  & 2.5  & 0.89 \\
twalock & 4M  & 11  & 4K  & 8.9  & 7.0  & 1.26  & szymanski & 4M  & 3  & 3  & 1.8  & 2.3  & 0.81 \\
\hline
\textbf{Totals} & - & - & - & - & - & - & - & \textbf{196M} & - & - & \textbf{286.4} & \textbf{170.8} & -\\
% \hline\hline\textbf{Total} & \textbf{58K} & \textbf{115M} & \textbf{47M} & \textbf{35M} & \textbf{110.8} & \textbf{118.9} & \textbf{22.5} & \textbf{34.5} & \textbf{286.6} & \textbf{95.8} & \textbf{18.9} & \textbf{8.0} & \textbf{27.6} & \textbf{171.1} & \textbf{}\\
			\hline
		\end{tabular}
}  
	}
\end{table}
\vspace{-0.2in}

%% file: conclusion.tex
\section{Conclusion}\label{sec:conclusion}

Checking the reads-from consistency of concurrent executions is a fundamental computational task in the development of formal concurrency semantics, program verification and testing.
In this paper we have addressed this problem in the context of C11-style weak memory models, 
for which this problem is both highly meaningful, and intricate.
We have developed a collection of algorithms and complexity results that are either optimal or nearly-optimal, and thus accurately characterize the complexity of the problem in this setting.
Further, our experimental evaluation indicates that the new algorithms have a measurable, and often significant, impact on the consistency-checking tasks that arise in practice.
Thus our algorithms enable the development of more performant and scalable program analysis tools in this domain.
This work is focused on non-SC fragments of C11, as otherwise, consistency checking inherits the $\NP$-hardness of $\scmm$ consistency checking.
For applications having an abundance of SC accesses, however, a meaningful direction for future work is to combine our techniques with heuristics developed for checking $\scmm$ consistency (e.g., \cite{Abdulla:2018, Pavlogiannis2019}), and apply them on programs that mix all types of C11 accesses.

%% file: artifact.tex
\section*{Data and Software Availability Statement}

The artifact developed for this work is available~\cite{artifact}, which contains all source codes and experimental data necessary to
reproduce our evaluation in~\cref{sec:experiments}.

%We provide an artifact~\cite{artifact} which contains all source codes and experimental data necessary to reproduce our evaluation in~\cref{sec:experiments}.
%In particular, the extended versions of the \celeventester~\cite{Luo:2021} and \Trust~\cite{Kokologiannakis:2022} tools with our new algorithms, as well as all the benchmarks used in our evaluation are included.
%The artifact contains scripts that fully automate the process of reproducing~\cref{tab:expr-results-smc} and~\cref{tab:expr-results-c11}.
%Instructions on the usage of these scripts and customization of the evaluation are also provided with the artifact.

%% file: app_rf.tex
%!TEX root=./main.tex
\section{Proofs} \label{sec:app_rf}

\input{app_models}
\input{app_rf_wra}

%\input{app_sra_rf_np_hard}
\input{app_rf_sra}
\input{app_rf_sra_normw}
\input{app_rf_rc20}

\input{app_rf_rlx}

\input{app_rf_lower}

%% file: app_models.tex
\subsection{Proofs from {\cref{sec:helper_functions}}}\label{subsec:app_helper_functions}

%We now present the proof of one of our core observation stated in \cref{lem:read_coherence_atomicity}.

\lemreadcoherenceatomicity*
\begin{proof}
First we have $(\wt', \wt)\in \mo$, due to read-coherence.
Since $(\wt', \wt)\not \in \rf^+$, we have that $\wt'$ is not in the $\rf$-chain of $\wt$, otherwise we would have $(\wt,\wt')\in \rf^+$, which together with $(\wt', \wt)\in \mo$, would violate write-coherence.
Due to atomicity, we then have $(\wt', \wt_1) \in \mo $, where $\wt_1=\rfinv(\wt)$.
The final statement follows by applying atomicity inductively on the $\rf$-chain, i.e., for each $\wt_i=\rfinv(\wt_{i-1})$, we have $(\wt', \wt_i)\in \mo$.
\end{proof}

%% file: app_rf_wra.tex
%!TEX root=./main.tex

\subsection{Proofs from {\cref{subsec:rf_wra}}}\label{subsec:app_rf_wra}

%\LastPerThread*
%
%\begin{proof}
%The statement follows because $\hb$ is closed under composition with $\po^?$.
%That is, observe that $(e_3, e^\textsf{last}_3) \in \po^?$.
%This means $(e_1, e^\textsf{last}_3) \in \hb;\po^? = \hb$.
%\end{proof}

\thmwrarfupper*

\begin{proof}
We argue about both correctness and running time.

\bparagraph{Correctness of \cref{algo:rf-consistency-wra}}
We first argue that whenever the algorithm reports a violation, 
the input $\expartial$ is not WRA consistent.
If $\hb$ is acyclic or $\rf$ violates weak-atomicity (\cref{line:rf-wra-hb-acyclic}),
then clearly $\expartial$ violates one of irrHB or  weak-atomicity and thus violates
$\wramm$ consistency.
Next, suppose a violation is reported on \cref{line:wra-rf-weak-read-coherence} when processing event $e$.
Observe that $e.\op \in \set{\rd, \ud}$, and let $x = e.\lloc$, as well as $w = \rfinv(e)$, with $t' = w.\tid$.
There must be a thread $u$ and another event $w_u$
such that $w_u \not\in \set{w, e}$
and further, $(w, w_u) \in \hb$ and $(w_u, e) \in \hb$.
This clearly violates weak-read-coherence since
$(w, w_u) \in \hbloc$, $(w_u, e) \in \hb$ and $(e, w) \in \rf^{-1}$
and thus
$(\hbloc;[\W \cup \Upd];\hb;\rf^{-1})$ is reflexive.

Now, let's assume that no violation was reported.
Then clearly $\hb$ is irreflexive and $\rf$ satisfies weak-atomicity.
Now assume on the contrary that weak-read-coherence was violated by some triplet $(w_u, w, r)$
where $w_u.\lloc = w.\lloc = r.\lloc = x$, $(w, r) \in \rf$, $w_u \not\in \set{w_u, w}$,
$w_u.\tid = u$, $w.\tid = t'$ and $\set{(w, w_u),(w_u, r)} \subseteq \hb$.
Then, when processing event $r$, we will consider the event
$w_u^\textsf{last}$ such that $w_u^\textsf{last}$ is the last event
that is both $\po^?$-after $w_u$ and is $\hb$-before $r$,
%From \cref{lem:Last-Per-Thread}, we have that
yielding $\set{(w, w_u^\textsf{last}),(w_u^\textsf{last}, r)} \subseteq \hb$.
This means the algorithm will report a violation, which is a contradiction.

\bparagraph{Complexity}
First, checking for reflexivity of $\hb$ takes $O(\NumEvents)$ time --- this can be 
done using a cycle detection algorithm on a graph with $|\E| = \NumEvents$ nodes
and with $O(\NumEvents)$ edges (corresponding to $\rf$ and $\po_\imm$ orderings).
Next, checking for weak-atomicity of $\rf$ also takes $O(\NumEvents)$ --- scan over all write events
and check if there are more than one RMW events that read from any given write.
Next, the computation of $\hb$ timestamps takes $O(\NumEvents\cdot\NumThreads)$ time (\cref{prop:HB-computation})
and the computation involved in initializing and accessing the lists $\set{\view{\WtLst}{t,x}{u}}_{t, x, u}$ takes $O(\NumEvents\cdot \NumThreads)$ time.
Finally, at each read event, the algorithm spends $O(\NumThreads)$ time.
This gives a total running time of $O(\NumEvents\cdot\NumThreads)$.
\end{proof}

%% file: app_rf_sra.tex
%!TEX root=./main.tex

\subsection{Proofs from {\cref{subsec:rf_sra_lower}}}\label{subsec:app_sra_rf_upper}

\subsubsection{Proof of {\cref{thm:sra_rf_lower}}}
We first present the proof of the lower bound.

\lemschardness*
\begin{proof}
Consider two partial executions $\expartial_1=(\E_1, \po_1, \rf_1)$, $\expartial_2=(\E_2, \po_2, \rf_2)$, where $\E_1$ consists of only read/write events, $\rf_1$ is bijective, and $\expartial_2$ is identical to $\expartial_1$ except that every read event of $\E_1$ is replaced by an RMW event in $\E_2$.
We argue that $\expartial_1\models \scmm$ iff $\expartial_2\models \scmm$.
First, consider an $\mo_2$ witnessing the consistency of $\expartial_2$.
Take $\mo_1$ be identical to $\mo_2$, restricted on the write events of $\E$.
Since the write events of $\E_1$ are a subset of the write/RMW events of $\E_2$, we have
$\mo_1\subseteq \mo_2$ and $\fr_1\subseteq \fr_2$,
where $\fr_i$ is the corresponding from-reads relation derived from $\mo_i$.
Since $(\po_2\cup \rf_2 \cup \mo_2 \cup \fr_2)$ is acyclic we have that $(\po_1\cup \rf_1 \cup \mo_1 \cup \fr_1)$ is also acyclic,
thus $\mo_1$ witnesses the $\scmm$ consistency of $\expartial_1$.

For the inverse direction, consider an $\mo_1$ witnessing the consistency of $\expartial_1$,
and let $\mo_2=\mo_1\cup \fr_1$.
Since there are no read events in $\E_2$, we have $\fr_2\subseteq \mo_2$.
Since $(\po_1\cup\rf_1\cup \mo_1 \cup \fr_1)$ is acyclic, we have that $(\po_2\cup\rf_2\cup \mo_2 ) = (\po_2\cup\rf_2\cup \mo_2 \cup \fr_2)$ is also acyclic.
Thus $\mo_2$ witnesses the $\scmm$ consistency of $\expartial_2$.
\end{proof}

\thmsrarflower*
\begin{proof}
Let $\expartial=(\E, \po, \rf)$ be an instance of the consistency problem for $\scmm$.
In light of \cref{lem:sc_hardness}, we assume that
$\E$ consists of only write/RMW events.
Since $\scmm\sqsubseteq \sramm$, if $\expartial\models \scmm$ then $\expartial\models \sramm$.
For the inverse direction, assume that $\expartial\models \sramm$ and let $\ex=(\E, \po, \rf, \mo)$ be a witness execution, i.e., $\ex\models \sramm$.
%Strong-write-coherence states that $(\hb\cup \mo)=(\po\cup \rf\cup\mo)$ is acyclic.
Consider the derived relation $\fr$.
Since there are no read events in $\E$, we have $\fr\subseteq \mo$.
Thus $(\po\cup \rf\cup\mo)=(\po\cup \rf\cup\mo\cup \fr)$, and since the former is acyclic, due to strong-write-coherence, so is the latter.
Hence $\ex\models \scmm$, thereby serving as a witness for $\expartial\models \scmm$.
%The desired result follows.
\end{proof}

\subsubsection{Proof of {\cref{thm:sra_rf_upper}}}

We now turn our attention to~\cref{thm:sra_rf_upper}.
We argue separately about soundness (\cref{lem:reach2co}) and completeness (\cref{lem:co2reach}).

\begin{restatable}{lemma}{lemreachtoco}
\label{lem:reach2co}
If node $\E$ is reachable  in $\Gg_{\expartial}$, then $\expartial\models \sramm$.
\end{restatable}
\begin{proof}
Consider a path $\rho$ in $\Gg_{\expartial}$ from the root node $\emptyset$ to the terminal node $\E$.
We represent $\rho$ as a sequence of events, which are the ones executed along the corresponding edges traversed by $\rho$.
We construct a modification order $\mo$ such that, for any two write/RMW events $\wt_1, \wt_2$ on the same location,
we have $(\wt_1,\wt_2)\in \mo$ iff $\wt_1$ is executed before $\wt_2$ in $\rho$.
We argue that the resulting execution $\ex=(\E, \po,\rf,\mo)$ satisfies the $\sramm$ axioms.
\begin{compactenum}
\item \emph{Strong-write-coherence:}~Due to \cref{item:executability1} of the executability conditions, no event $\event$ of $\rho$ is $\hb$-ordered before any event appearing before $\event$ in $\rho$.
Moreover, by construction, no event $\event$ of $\rho$ is $\mo$-ordered before any event appearing before $\event$ in $\rho$.
Thus $\rho$ represents a linear extension of $(\hb\cup \mo)$, and hence write-coherence is satisfied.

\item \emph{Read-coherence:}~Consider any triplet $(\wt, \rd, \wt')$ such that $(\wt', \rd)\in\hb$.
It suffices to argue that $\wt'$ appears in $\rho$ before $\wt$, implying that $(\wt', \wt)\in \mo$, and thus read-coherence is not violated for $\rd$.
Indeed, \cref{item:enabledness1} of the enabledness condition makes $\wt$ not enabled in any node $S$ for which $\wt'\not \in S$.
Thus any path from the root $\emptyset$ to a node that contains both $\wt$ and $\wt'$ must execute $\wt'$ before $\wt$.

\item \emph{Atomicity:}~Consider any trilpet $(\wt, \ud, \wt')$, where $\ud\in \Upd$ is an RMW event.
We have to argue that $\wt'$ does not appear between $\wt$ and $\ud$ in $\rho$.
Indeed, \cref{item:enabledness2} of the enabledness condition makes $\wt'$ not enabled in any node $S$ for which $\wt\in S$ and $\ud\not \in S$.
Thus, any path from the root $\emptyset$ to a node that contains both $\wt$ and $\wt'$ must execute $\wt'$ before $\wt$.
\end{compactenum}

The desired result follows.
\end{proof}

\begin{restatable}{lemma}{lemco2reach}\label{lem:co2reach}
If $\expartial\models \sramm$, then node $\E$ is reachable in $\Gg_{\expartial}$.
\end{restatable}
\begin{proof}
Consider any witness complete execution $\ex=(\E, \po, \rf, \mo)$ with $\ex\models \sramm$.
We show the existence of a path $\rho$ from the root $\emptyset$ to the terminal $\E$ incrementally, with the property that every node $S$ in $\rho$ is downward-closed wrt $\hb\cup \mo$.
The statement clearly holds for the root $\emptyset$.
Now assume that the statement holds for some $S$, and we argue that there is an executable event that takes us to $S'$, with the statement holding in $S'$.

If there is a read event $\rd$ such that $S\cup \{ \rd \}$ is downward-closed wrt to $\hb$, 
then $\rd$ is enabled in $S$, and we execute it to arrive at $S'=S\cup \{ \rd\}$.
Since $\mo$ does not relate read events, we have that $S'$ is also downward-closed wrt $(\hb\cup \mo)$.

Otherwise, all $\hb$-minimal events in $\E\setminus S$ are write/RMW events.
Consider any such event $\wt$ that is also $(\hb\cup \mo)$-minimal, and we argue that $\wt$ is executable in $S$.

First, consider any triplet $(\wt, \rd, \wt')$ with $(\wt',\rd)\in \hb$.
Read-coherence on $\ex$ implies that $(\wt', \wt)\in\mo$.
Since $\wt$ is $(\hb\cup \mo)$-minimal in $\E\setminus S$ , we have that $\wt'\in S$.
Thus the enabledness condition in \cref{item:enabledness1} holds.

Second, consider any triplet $(\wt'', \ud, \wt)$, where $\ud\in\Upd$ is an RMW event, and such that $\wt''\in S$.
Since $S$ is downward-closed wrt $(\hb\cup \mo)$, we have that $(\wt'',\wt)\in \mo$.
Note that $\ud\in S$ as well, as otherwise, by the minimality of $\wt$, we would have $(\wt,\ud)\in \mo$, which would violate atomicity.
Thus the enabledness condition in \cref{item:enabledness2} also holds, thereby making $\wt$ executable in $S$.
Finally, observe that, by construction, the set $S'=S\cup\{ \wt \}$ is downward-closed wrt $(\hb\cup \mo)$.

The desired result follows.
\end{proof}

%% file: app_rf_sra_normw.tex
\subsection{Proof of {\cref{subsec:rf_sra_normw}}}\label{subsec:app_sra_rf_normw}

\lemsraminimalcoherence*
\begin{proof}
Since $\mopartial$ is minimally coherent for $\expartial$, we have that $(\hb\cup \mopartial)$ is acyclic.
Let $\mo$ be any linear extension of $(\hb\cup \mopartial)^+$ projected on $\W$, and $\ex=(\E, \po, \rf, \mo)$ the resulting execution.
Clearly $\ex$ satisfies strong-write-coherence.
Now consider any triplet $(\wt, \rd, \wt')$ with $(\wt', \rd)\in \hb$.
As $\mopartial$ is minimally coherent for $\ex$, we have $(\wt', \wt)\in (\hb\cup\mopartial)^+$.
Thus $(\wt', \wt)\in \mo$, and $\ex$ also satisfies read-coherence.
Hence $\ex\models \sramm$, witnessing that $\expartial\models \sramm$, as desired.
\end{proof}

\thmsrarfnormwupper*
\begin{proof}
We start with soundness.
First, the algorithm verifies that $\po\cup\rf$ is acyclic (\cref{line:rf-sra-hb-acyclic}).
At the end of the algorithm, for every triplet $(\wt, \rd, \wt')$ such that $(\wt', \rd)\in \hb$ and $\wt'$ is $\po$-maximal among all such write/RMW events, we have $(\wt', \wt)\in \mopartial$.
It follows that, at the end of the loop of \cref{line:rf-sra-loop-start}, 
for every triplet $(\wt, \rd, \wt')$ with $(\wt',\rd)\in \hb$, we have $(\wt', \wt)\in (\hb\cup \mopartial)$, thus $\mopartial$ satisfies \cref{item:minimal_coherence_sra1} of minimal coherence.
Finally, the algorithm verifies that $(\hb \cup \bigcup_{x \in \E.locs}\mopartial_x)$ is acyclic,
thus $\mopartial$ satisfies \cref{item:minimal_coherence_sra2} of minimal coherence. 
In turn, \cref{lem:sra_minimal_coherence} guarantees that indeed $\expartial\models\sramm$.

We now turn our attention to completeness.
Clearly, if the algorithm returns ``Inconsistent'' due to $(\po\cup \rf)$ being cyclic (\cref{line:rf-sra-hb-acyclic}), we have $\expartial\not \models \sramm$.
Otherwise, observe that, whenever the algorithm inserts an ordering $(\wt_u, \wt_{\mathsf{rf}})\in \mopartial$, we have a triplet $(\wt_{\mathsf{rf}}, \rd, \wt_u)$ and $(\wt_u, \rd)\in \hb$. 
But then \cref{cor:read_coherence_sra} guarantees that, for any $\mo$ witnessing the consistency of $\expartial$, we have $\mopartial\subseteq \mo$.
Hence if the algorithm returns ``Inconsistent'' in \cref{line:rf-sra-cycle}, we have that $(\hb \cup \bigcup_{x \in \E.locs}\mo_x)$ is cyclic, contradicting $\mo$ as a consistency witness.
Thus $\expartial\not \models \sramm$.
\end{proof}

%% file: app_rf_rc20.tex
\subsection{Proofs from {\cref{subsec:rf_rc20}}}\label{subsec:app_rf_rc20}

We begin with the proof of \cref{lem:rc_minimal_coherence}.
Consider any location $x\in \E.locs$, and we extend $\mopartial_x$ to a total modification order $\mo_x$ that satisfies read-coherence, write-coherence and atomicity.
We build $\mo_x$ in three phases, as follows.
\begin{compactenum}
\item Initially, we set $\mo_x=\mopartial_x$.
\item While there exist two (distinct, maximal) $\rf$-chains
$\wt^1_1, \wt^1_2,\dots, \wt^1_{i_1}$ and $\wt^2_1, \wt^2_2,\dots, \wt^2_{i_2}$
such that $(\wt^2_1,\wt^1_{i_1})\not \in (\rf_x\cup\hb_x\cup \mo_x)^+$,
we insert an ordering $(\wt^1_{i_1}, \wt^2_1)$ in $\mo_x$.
\item Finally we set $\mo_x$ to be $(\rf_x\cup \mo_x)^+$ projected on $(\WRMW)$.
\end{compactenum}

Towards \cref{lem:rc_minimal_coherence}, we establish three intermediate lemmas.

\begin{restatable}{lemma}{lemcompletionacyclicity}\label{lem:completion_acyclicity}
We have that $(\rf_x\cup \hb_x \cup \mo_x)$ is acyclic at the end of phase~2.
\end{restatable}
\begin{proof}
Since $\mopartial$ is minimally coherent, the statement holds for $\mo_x$ at the beginning of phase~2.
Moreover, in every step of phase~2 inserting an ordering $(\wt^1_{i_1}, \wt^2_1)$ in $\mo_x$, we have 
$(\wt^2_1,\wt^1_{i_1})\not \in (\rf_x\cup\hb_x\cup \mo_x)^+$.
Hence $(\rf_x\cup \hb_x \cup \mo_x)$ remains acyclic after the ordering has been inserted in $\mo_x$.
\end{proof}

\begin{restatable}{lemma}{lemcompletionatomicity}\label{lem:completion_atomicity}
At all times, for every $\rf$-chain $\wt_1, \wt_2,\dots, \wt_{i}$
and write/RMW event $\wt$, all accessing the same location $x$,
if $\TC[\wt]\neq \wt_1$ and $(\wt, \wt_j)\in (\rf_x\cup \hb_x\cup \mo_x)^+$ for some $j\in\set{1,\dots, i}$, then $(\wt, \wt_1)\in (\rf_x\cup\hb_x\cup\mo_x)^+$.
\end{restatable}
\begin{proof}
We first prove that the statement holds at the end of phase~1, i.e., it holds for $\mopartial_x$ instead of $\mo_x$.
Consider a path $\wt\Path \wt_j$ in $(\rf_x\cup \hb_x \cup \mopartial_x)^+$, and let $\wt'$ be the first event in this path that belongs to the $\rf$-chain (i.e., $\TC[\wt']=\wt_1$).
Hence $\wt'\neq \wt$.
If $\wt'=\wt_1$, we are done.
Otherwise, let $\wt''$ be the immediate predecessor of $\wt'$ in this path.
We distinguish cases for the edge connecting $\wt''$ and $\wt'$ in the path.
\begin{compactenum}
\item $(\wt'',\wt')\in \mopartial_x$.
Since $\mopartial$ is minimally coherent, \cref{item:minimal_coherence_rc2} implies that $(\wt'', \wt_1)\in \mopartial_x$.
\item $(\wt'',\wt')\in (\rf_x\cup \hb_x)$.
Then $(\wt'', \wt')\in \rf^?;\hb$.
Since $\mopartial$ is minimally coherent, \cref{item:minimal_coherence_rc1} implies that $(\wt'', \wt_1)\in (\rf_x\cup \hb_x \cup \mopartial_x)^+$.
\end{compactenum}
Thus in both cases $(\wt,\wt_1)\in (\rf_x\cup \hb_x\cup \mopartial_x)^+$, as desired.
Second, observe that in every step of phase~2, we insert an $\mo_x$ ordering towards the top of an $\rf$-chain.
Hence the statement holds along every step of phase~2.
\end{proof}

\begin{restatable}{lemma}{lemcompletiontotality}\label{lem:completion_totality}
At the end of phase~2, $(\rf_x\cup \mo_x)^+$ projected on $(\locx{\W}\cup\locx{\Upd})$ is total.
\end{restatable}
\begin{proof}
Consider any two (distinct, maximal) $\rf$-chains $\wt^1_1, \wt^1_2,\dots, \wt^1_{i_2}$ and $\wt^2_1, \wt^2_2,\dots, \wt^2_{i_2}$ accessing the same location $x$, and such that $(\wt^1_{i_1}, \wt^2_1)\not \in \mo_x$ at the end of phase~2.
Since that phase completed, we have $(\wt^2_1,\wt^1_{i_1})\in (\rf_x\cup \hb_x\cup \mo_x)^+$,
and thus $(\wt^2_1,\wt^1_{1})\in (\rf_x\cup \hb_x\cup \mo_x)^+$ by \cref{lem:completion_atomicity}.
If also $(\wt^2_{i_2}, \wt^1_1)\not \in \mo_x$, we similarly have $(\wt^1_1,\wt^2_{1})\in (\rf_x\cup \hb_x\cup \mo_x)^+$.
But then $(\rf_x\cup\hb_x\cup \mo_x)^+$ would be cyclic, violating \cref{lem:completion_acyclicity}.
Thus $(\wt^2_{i_2}, \wt^1_2) \in \mo_x$.
It follows that $(\rf_x\cup \mo_x)^+$ projected on $(\locx{\W}\cup\locx{\Upd})$ is total. 
\end{proof}

We are now ready to prove \cref{lem:rc_minimal_coherence}.

\lemrcminimalcoherence*
\begin{proof}
Consider the modification order $\mo$ constructed above, and we argue that $\mo$ witnesses the consistency of $\expartial$.
Write-coherence follows from \cref{lem:completion_acyclicity}.
For read-coherence, consider any triplet $(\wt, \rd, \wt')$ such that $(\wt', \rd)\in \rf^?;\hb$.
If $(\wt', \wt)\in \rf^+$, we have $(\wt, \wt')\not \in \mo_x$ due to write-coherence.
Otherwise, since $\mopartial$ is minimally coherent, we have $(\wt', \TC[\wt]) \in \mopartial$, and due to write-coherence, we also have $(\wt, \wt')\not \in \mo_x$.
Hence $(\rd,\rd)\not \in \fr;\rf^?;\hb$, satisfying read-coherence.
Finally, we argue about atomicity.
Consider any RMW event $\RMW$, and any triplet $(\wt, \RMW, \wt')$ with $(\wt', \RMW)\in \mo_x$.
If $\TC[\wt']=\TC[\RMW]$, due to read/write-coherence, we have that $\PC[\wt']<\PC[\wt]$.
Thus $(\wt', \wt)\in \rf_x^+$, and hence $(\wt, \wt')\not \in \mo_x$.
Otherwise, due to \cref{lem:completion_atomicity}, we have that $(\wt', \TC[\RMW])\in (\rf_x\cup\hb_x\cup \mo_x)^+$, and by write-coherence, we again have $(\wt, \wt')\not \in \mo_x$.
Thus, in both cases, $(\RMW, \wt')\not \in \fr$.
The desired result follows.
\end{proof}

\thmrcrfupper*
\begin{proof}
We start with soundness.
First, the algorithm verifies that $(\po\cup\rf)$ is acyclic (\cref{line:rf-rc20-po-rf-acyclic}).
At the end of the algorithm, for every triplet $(\wt, \rd, \wt')$ such that $(\wt',\rd)\in \rf^?;\hb$ and $\wt'$ is $\po$-maximal among all such write/RMW events, if $(\wt', \wt)\not \in \rf^+$, we have $(\wt', \wt)\in \mopartial_x$.
It follows that, at the end of loop of \cref{line:rf-rc20-mo-infer-start}, for every triplet $(\wt, \rd, \wt')$ such that $\wt'\in \rf^?;\hb$, if $(\wt', \wt)\not \in \rf^+$, then $(\wt', \wt)\in \hb^?;\mopartial$.
Thus $\mopartial$ trivially satisfies \cref{item:minimal_coherence_rc1} of minimal coherence.
Moreover, since every $\mopartial$ ordering goes to the top of an $\rf$-chain, $\mopartial$ also satisfies \cref{item:minimal_coherence_rc2} of minimal coherence.
Finally, the algorithm verifies that $(\rf_x \cup \hb_x \cup \mopartial_{x})$ is acyclic for every location $x$, hence $\mopartial$ satisfies \cref{item:minimal_coherence_rc3} of minimal coherence.
Thus, \cref{lem:rc_minimal_coherence} guarantees that indeed $\expartial\models \rcmm$.

We now turn our attention to completeness.
Clearly, if the algorithm returns ``Inconsistent'' due to $(\po\cup \rf)$ being cyclic (\cref{line:rf-rc20-po-rf-acyclic}), we have $\expartial\not \models \rcmm$.
Otherwise, observe that, whenever the algorithm inserts an ordering $(\wt_u, \wt)$ in $\mopartial$ in the loop of \cref{line:rf-rc20-mo-infer-start},
we have $(\wt_u, \rd)\in \rf^?;\hb$, where $\rd$ is a read/RMW event such that there exists a triplet $(\wt_{\mathsf{rf}}, \rd, \wt_u)$ and $(\wt_u, \wt_{\mathsf{rf}})\not \in \rf^+$.
But then \cref{lem:read_coherence_atomicity} guarantees that, for any $\mo$ witnessing the consistency of $\expartial$, we have $\mopartial\subseteq \mo$.
Hence if the algorithm returns ``Inconsistent'' in \cref{line:rf-rc20-cycle}, we have that $(\rf_x\cup \hb_x\cup \mo_x)$ is cyclic for some location $x$, contradicting $\mo$ as a consistency witness.
Thus $\expartial\not \models \rcmm$.
\end{proof}

%% file: app_rf_rlx.tex
\subsection{Proofs from {\cref{subsec:rf_rlx}}}\label{subsec:app_rf_rlx}

Towards \cref{thm:rlx_rf_upper}, we first prove the following intermediate lemma

\begin{restatable}{lemma}{lemrlxlastwrite}\label{lem:rlx_last_write}
Before \cref{algo:rf-consistency-relaxed} processes a read/RMW event $\event(t,x)$, 
for every triplet $(\wt_{\mathsf{rf}}, \event, \wt')$,
if $(\wt', \event)\in \rf^?;\po$ then $\wt'=\LW_{t,x}$ or $(\wt', \LW_{t,x})\in (\rf\cup \po\cup\mopartial)^+$.
\end{restatable}
\begin{proof}
The proof is by induction on the events of thread $t$ and location $x$ processed by the algorithm.
The statement clearly holds for $\event$ being the first event of thread $t$, as there is no $\wt'$ with $(\wt', \event)\in \rf^?;\po$, while $\LW_{t,x}=\bot$.

Now consider that the algorithm processes an event $\event$, and the statement holds for the event $\event'$ that is an immediate $\po$-predecessor of $\event$ in $t$.
We split cases based on the type of $\event'$.
\begin{compactenum}
\item $\event'\in (\locx{\W}\cup \locx{\Upd})$.
Then $\LW_{t,x}=\event'$, while clearly, either $\wt'=\event'$, or $(\wt', \event')\in \rf^?;\po$. 
\item $\event'\in \locx{\R}$.
Then $\LW_{t,x}=\rfinv(\event')$, while clearly, either $\wt'=\rfinv(\event')$, or $(\wt', \event')\in \rf^?;\po$.
But for every event $\wt'$ with $\wt'\neq \rfinv(\event')$ and $(\wt', \event')\in \rf^?;\po$,
the algorithm has either inserted an ordering $(\wt', \rfinv(\event'))\in\mopartial_x$, or infered that $(\wt',\rfinv(\event))\in \rf^+$.
Thus the statement holds for $\event$.
\end{compactenum}
\end{proof}

\thmrlxrfupper*
\begin{proof}
We start with soundness.
First, the algorithm verifies that $(\po\cup\rf)$ is acyclic (\cref{line:rf-relaxed-hb-acyclic}).
Due to \cref{lem:rlx_last_write},
at the end of the algorithm, for every triplet $(\wt, \rd, \wt')$ such that $(\wt', \rd)\in \rf^?;\po$ if $(\wt', \wt)\not \in \rf^+$, we have $(\wt', \wt)\in \mopartial_x$.
Thus $\mopartial$ satisfies \cref{item:minimal_coherence_rlx1} of minimal coherence.
Moreover, since every $\mopartial$ ordering goes to the top of an $\rf$-chain, $\mopartial$ also satisfies \cref{item:minimal_coherence_rc2} of minimal coherence.
Finally, the algorithm verifies that $(\rf_x \cup \hb_x \cup \mopartial_{x})$ is acyclic for every location $x$ in \cref{line:rf-relaxed-mo-cycle}, hence $\mopartial$ satisfies \cref{item:minimal_coherence_rlx3} of minimal coherence.
Thus, \cref{lem:rc_minimal_coherence} guarantees that indeed $\expartial\models \rlxmm$.

We now turn our attention to completeness.
Clearly, if the algorithm returns ``Inconsistent'' due to $(\po\cup \rf)$ being cyclic (\cref{line:rf-relaxed-hb-acyclic}), we have $\expartial\not \models \rlxmm$.
Otherwise, observe that, whenever the algorithm inserts an ordering $(\wt_u, \wt)$ in $\mopartial$ in the loop of \cref{line:rf-rlx-mo-infer-start},
we have $(\LW_{t,x}, \rd)\in \rf^?;\po$, where $\rd$ is a read/RMW event such that there exists a triplet $(\wt_{\mathsf{rf}}, \rd, \LW_{t,x})$, and $(\LW_{t,x}, \wt_{\mathsf{rf}})\not \in \rf^+$.
But then \cref{lem:read_coherence_atomicity} guarantees that, for any $\mo$ witnessing the consistency of $\expartial$, we have $\mopartial\subseteq \mo$.
Hence if the algorithm returns ``Inconsistent'' in \cref{line:rf-relaxed-mo-cycle}, we have that $(\rf_x\cup \po_x\cup \mo_x)$ is cyclic for some location $x$, contradicting $\mo$ as a consistency witness.
Thus $\expartial\not \models \rlxmm$.

Finally, we turn our attention to complexity.
The computation of $(\TC, \PC)$ in \cref{line:rf-relaxed-hb-acyclic} takes $O(\NumEvents)$ time (\cref{sec:helper_functions}).
For every location $x$, the loop in \cref{line:rf-rlx-mo-infer-start} takes $O(\locx{\NumEvents})$ time, where $\locx{\NumEvents}=|\locx{\E}|$.
The acyclicity check in \cref{line:rf-relaxed-mo-cycle} can also easily performed in $O(\NumEvents_x)$ time, viewed as a cycle detection problem on a graph of $O(\locx{\NumEvents})$ edges.
Summing over all locations $x$, we obtain total time $\sum_x O(\locx{\NumEvents})=O(\NumEvents)$.
\end{proof}

%% file: app_rf_lower.tex
%!TEX root=./main.tex

\subsection{Proofs of {\cref{subsec:rf_lower}}}\label{subsec:app_rf_lower}

\bparagraph{Correctness}
Let us now argue about the correctness of the construction i.e., $G$ has a triangle iff $\expartial$ is not consistent under $\wramm$, $\ramm$ or $\sramm$. 
The precise statements we prove are as follows.
First, if $G$ has a triangle then $\expartial\not\models \wramm$, implying $\expartial\not\models \ramm$ and $\expartial\not\models \sramm$.
Second, we show that if $G$ does not have a triangle then $\expartial\models \sramm$, implying $\expartial\not\models \ramm$ and $\expartial\not\models \wramm$. 

We first observe that $\hb$ is indeed acyclic.
This is because the event set $\E$ can be partitioned into $6$ classes
$\E_1, \E_2, \ldots, \E_6$, and the $\hb$ edges only go from $\E_i$ to $\E_{j}$ for $i<j$.
Specifically, $\E_1 = \setpred{\wt_{\alpha}}{\alpha \in V_G}$,
$\E_2 = \setpred{\wt^\beta_{\alpha}}{\set{\alpha, \beta} \in E_G, \alpha < \beta}$,
$\E_3 = \setpred{\juncEvWt_{\alpha}}{\alpha \in V_G}$,
$\E_4 = \setpred{e_{(\alpha, \beta)}}{\set{\alpha, \beta} \in E_G, \alpha < \beta}$, 
$\E_5 = \setpred{\juncEvWt_{\alpha}}{\alpha \in V_G}$, and
$\E_6 = \setpred{\rd^{\alpha}_{\beta}}{\set{\alpha, \beta} \in E_G, \alpha < \beta}$.
Observe that the $\hb$ edges only go from $\E_1$ to $\E_2$, $\E_2$ to $\E_3$,
$\E_3$ to $\E_4$, $\E_4$ to $\E_5$, $\E_5$ to $\E_6$ and also from $\E_1$ to $\E_6$.
% Here we do not explicitly mention the `connection' events;
% it is easy to accommodate them here by increasing the number of partitions, 
% by adding one partition for each kind of connection events.

\noindent
($\Rightarrow$) Assume that there is a triangle in $G$, and we show that there is a violation of weak-read-coherence.
Consider the triangle $(\alpha, \beta, \gamma)$ in $G$, with $\alpha<\beta<\gamma$.
Consider the events $\wt_{\beta}$, $\wt^\gamma_{\beta}$ and $\rd^{\alpha}_{\beta}$,
which access the same memory location $\varN_{\beta}$.
The construction ensures that
$(\wt_{\beta}, \rd^{\alpha}_{\beta}) \in \rf$, and also$(\wt_{\beta}, \wt^{\gamma}_{\beta}) \in \hb$.
Now observe the sequence of $\hb$ edges: 
$(\wt^{\gamma}_{\beta}, \juncEvWt_{\gamma})$,
$(\juncEvWt_{\gamma}, e_{(\alpha, \gamma)})$,
$(e_{(\alpha, \gamma)}, \juncEvRd_{\alpha})$,
$(\juncEvRd_{\alpha}, \rd^{\alpha}_{\beta})$.
Clearly, all of these together imply that
$(\wt^\gamma_{\beta}, \rd^{\alpha}_{\beta}) \in \hb$.
This is a violation of weak-read-coherence.
See \cref{fig:super-linear-lower-bound} for illustration.

\newcommand{\cyc}{\mathcal{C}}

\noindent
($\Leftarrow$)
Assume that $G$ is triangle-free, and we show that $\expartial\models \sramm$.
Consider the per-location $\mo = \bigcup_{\alpha \in V_G} \mo_{\varN_\alpha}$ 
such that $\mo_\alpha$ is the total order
$\wt_{\alpha} < \wt^{\beta_1}_{\alpha} < \wt^{\beta_2}_{\alpha} \ldots < \wt^{\beta_k}_{\alpha}$,
where $\beta_1 < \beta_2 < \ldots < \beta_k$ are all the neighbors of $\alpha$
(recall that each write event accesses location $\varN_a$).
We will show that $\ex = \tup{\expartial.\E, \expartial.\po, \expartial.\rf, \mo}$
is $\sramm$-consistent.
First, we argue that $\ex$ satisfies strong-write-coherence.
Assume on the contrary that there is a $(\hb \cup \mo)^+$ cycle $\cyc$.
Since $\hb$ and $\mo$ are both acyclic, 
there must be at least one $\mo$ edge and at least one $\hb$ edge in $\cyc$.
Thus, there must be three distinct events $\event_1, \event_2, \event_3$ in the cycle such that
$(\event_1, \event_2) \in \hb$ and $(\event_2, \event_3) \in \mo$, and let $\varN_\alpha = e_2.\lloc = e_3.\lloc$.
Our construction ensures that the only $\hb$ edge incident to a write event of $\varN_\alpha$
emanates from $\wt_{\alpha}$, and thus $\event_1 = \wt_{\alpha}$.
Our construction further ensures that there is no incoming $\hb$ edge into $\wt_{\alpha}$,
while our $\mo$ also does not have incoming edges into $\wt_{\alpha}$,
which is a contradiction as $\event_1$ cannot participate in a $(\hb \cup \mo)^+$ cycle.
Let us now argue that $\ex$ satisfies read-coherence.
Assume on the contrary that it does not.
Then, there are distinct events $\event_1, \event_2, \event_3$
such that $(\event_1, \event_3) \in \rf$, $(\event_1, \event_2) \in \mo$
and $(\event_2, \event_3) \in \hb$.
The only memory locations with $> 1$ read accesses are those on 
the locations associated with a node in the graph.
Thus, there is a node $\beta \in V_G$ such that
$\event_1 = \wt_{\beta}$ and two nodes
$\alpha$ and $\gamma$ such that
$\event_3 = \rd^{\alpha}_{\beta}$ and $\event_2 = \wt^{\gamma}_{\beta}$.
It follows that $\alpha < \beta < \gamma$,
and that $(\alpha, \beta) \in E_G$ and $(beta, \gamma) \in E_G$.
Now consider the $\hb$ path from $\event_2$ to $\event_3$.
The only edge emanating from $\event_2 = \wt^{\gamma}_{\beta}$ is
$(\wt^{\gamma}_{\beta}, \juncEvWt_{\gamma}) \in \hb$.
Similarly, the only non-$\rf$ edge incident onto $\event_3 = \rd^{\alpha}_{\beta}$
is the edge $(\juncEvRd_{\alpha}, \rd^{\alpha}_{\beta}) \in \hb$.
Thus, it must be that $( \juncEvWt_{\gamma}, \juncEvRd_{\alpha}) \in \hb^+$.
This is possible only when the event $\event_{(\alpha, \gamma)}$ exists
as the only edges that are incident onto $\juncEvRd_{\alpha}$
emanate from an event of the form $\event_{(\alpha, *)}$, 
and the only edges that emanate from
$\juncEvWt_{\gamma}$ are incident on events of the form $\event_{(*, \gamma)}$.
Thus, $(\alpha, \gamma) \in E_G$ by the construction, resulting to a triangle $(\alpha, \beta,\gamma)$, which is a contradiction.
Thus $\ex$ satisfies all the axioms of $\sramm$.

\bparagraph{Time Complexity}
The total time to construct $\expartial$ is proportional to $O(|V_G| + |E_G|)$.
Also, the size of $\expartial$ is $O(|V_G| + |E_G|)$.
Thus, if there is an algorithm that solves consistency under any of 
$\sramm$, $\ramm$ or $\wramm$
in time $O(\NumEvents^{\omega/2 - \epsilon})$, then there is an algorithm that checks for triangle-freeness in time $O(\NumEvents^{\omega/2 - \epsilon})$ on graphs of $\NumEvents$ nodes.
Moreover, since our reduction is completely combinatorial, if there is a combinatorial algorithm that solves consistency in $O(\NumEvents^{3/2-\epsilon})$ time,
then there is an algorithm that checks for triangle-freeness in time $O(\NumEvents^{3/2-\epsilon})$ on graphs of $O(\NumEvents)$ edges.

%% file: app_experiments.tex
\section{Experimental Details}\label{sec:app_experiments}

\input{app_experiments_smc}
\input{app_experiments_c11}

%% file: app_experiments_smc.tex
\subsection{Further Experiments with \Trust}\label{subsec:app_experiments_smc}

Here we report further experiments inside the \Trust model checker on benchmarks that contain $\MOrlx$ accesses.   
$\Trust$ follows the RC11 semantics \cite{Lahav:2017} which has subtle difference from the RC20 model \cite{Lahav:2022} in the definition of release-sequence as it uses $\MOrlx$ accesses.  
The experimental results are shown in \cref{tab:expr-results-smc-rc20}.

\input{tables/tab-experiments-rc20}

%% file: tables/tab-experiments-rc20.tex
%!TEX root = main.tex

\begin{table}[H]
	\caption{
		Performance comparison. 
		Columns $2$ and $3$ denote the total number of executions explored by resp. \Trust and  our algorithm.
		Columns $5$ and $6$ denote the average time (in seconds) spent in consistency checking by resp. \Trust and  our algorithm.
		Columns $4$ and $7$ denote the respective ratios and speedups.
		%All times are in seconds.
		\label{tab:expr-results-smc-rc20}
	}
	\setlength\tabcolsep{0.6pt}
	\renewcommand{\arraystretch}{1.0}
	\centering
	\scalebox{0.95}{
{\small
		\begin{tabular}{|c|c|c|c|c|c|c||c|c|c|c|c|c|c|}
			\hline
% 			1 & 2 & 3 & 4 & 5 & 6 & 7 & 8 & 1 & 2 & 3 & 4 & 5 & 6 & 7 & 8 \\
1 & 2 & 3 & 4 & 5 & 6 & 7 & 1 & 2 & 3 & 4 & 5 & 6 & 7\\
			\hline
% 			\multirow{2}{*}{\textbf{Benchmark}} & 
% 			\multirow{2}{*}{$k$} &
% 			\multirow{2}{*}{$d$} &
% 			\multicolumn{2}{c|}{ {\rule{0pt}{1em} \celeventester} } & 
% 			\multicolumn{2}{c|}{ {\cref{algo:ra-rf-on-the-fly}} } &
% 			\multirow{2}{*}{SpdUp} &
% 			\multirow{2}{*}{\textbf{Benchmark}} & 
% 			\multirow{2}{*}{$k$} &
% 			\multirow{2}{*}{$d$} &
% 			\multicolumn{2}{c|}{ {\rule{0pt}{1em} \celeventester} } & 
% 			\multicolumn{2}{c|}{ {\cref{algo:ra-rf-on-the-fly}} } &
% 			\multirow{2}{*}{SpdUp}
% 			\\
            % \hline
			\multirow{2}{*}{\textbf{Benchmark}}& 

			\multicolumn{3}{c|}{ {\textbf{Avg. Time}} }  &
		\multicolumn{3}{c||}{ {\textbf{Executions}} }  &
		
		\multirow{2}{*}{\textbf{Benchmark}}
			&
			 \multicolumn{3}{c|}{ {\textbf{Avg. Time}} }
			& \multicolumn{3}{c|}{ {\textbf{Executions}} }
			 \\
			\cline{2-4}
			\cline{5-7}
			\cline{9-11}
			\cline{11-14}
			
			&\textbf{TruSt} & 
		\textbf{Our Alg.} &
		\textbf{R}
		&
		\textbf{TruSt} & 
		\textbf{Our Alg.} &
		\textbf{S} &
		       &
			\textbf{TruSt} &
			\textbf{Our Alg.} &
			\textbf{R} &
			\textbf{TruSt} & 
		\textbf{Our Alg.} &
		\textbf{S}
			\\
% 			\cline{4-7}
% 			\cline{12-15}
% 			 & & & \textsf{P} & Total & \textsf{P} & Total &  &
% 			 & & & \textsf{P} & Total & \textsf{P} & Total & 
% 			\\
			\hline
\rule{0pt}{1em} 
barrier & 0.5  & 0.01 & 45.0  & 14K  & 76K  & 5.65  & mutex & 0.8  & 0.01 & 76.0  & 8K  & 55K  & 7.17 \\
buf-ring & 1.6  & 0.05 & 31.6  & 3K  & 5K  & 1.83  & peterson & 0.3  & 0.01 & 26.0  & 4K  & 5K  & 1.16 \\
chase-lev$\dagger$ & 0.0 & 0.0 & -  & 2  & 2  & 1 & qu & 0.2  & 0.02 & 9.5  & 34  & 2K  & 71.62 \\
dekker & 0.3  & 0.1 & 2.7  & 24K  & 234K  & 9.88  & seqlock & 0.2  & 0.01 & 24.0  & 26K  & 118K  & 4.58 \\
dq & 0.1  & 0.01 & 13.0  & 43K  & 117K  & 2.72  & spinlock & 1.0  & 0.02 & 49.0  & 6K  & 17K  & 2.91 \\
fib-bench & 1.6  & 0.02 & 80.5  & 4K  & 194K  & 43.94  & szymanski & 0.5  & 0.01 & 50.0  & 2K  & 2K  & 1.17 \\
lamport & 0.4  & 0.01 & 36.0  & 17K  & 91K  & 5.41  & ticketlock & 0.9  & 0.03 & 30.3  & 7K  & 62K  & 8.52 \\
linuxrwlocks & 0.4  & 0.02 & 20.0  & 12K  & 21K  & 1.8  & treiber & 0.6  & 0.03 & 19.3  & 403  & 403  & 1\\
mcs-spinlock & 1.1  & 0.02 & 54.0  & 6K  & 23K  & 4.12  & ttaslock & 1.0  & 0.02 & 48.5  & 401  & 401  & 1\\
mpmc & 0.7  & 0.03 & 22.0  & 10K  & 65K  & 6.68  & twalock & 0.6  & 0.02 & 29.0  & 8K  & 13K  & 1.7 \\
ms-queue$\dagger$ & 0.07 & 0.01 & 7.0  & 210  & 210  & 1 &  & & & &   &   & \\
\hline
\hline
\textbf{Totals} & - & - & - & - & - & - & - & & \textbf{12.58} & \textbf{0.46} & \textbf{192K} & \textbf{1.1M} & -\\
\hline
%\textbf{Totals} & - & - & - & - & - & - & - & \textbf{196M} & - & - & \textbf{286.4} & \textbf{170.8} & -\\
% \hline\hline\textbf{Total} & \textbf{58K} & \textbf{115M} & \textbf{47M} & \textbf{35M} & \textbf{110.8} & \textbf{118.9} & \textbf{22.5} & \textbf{34.5} & \textbf{286.6} & \textbf{95.8} & \textbf{18.9} & \textbf{8.0} & \textbf{27.6} & \textbf{171.1} & \textbf{}\\
			% \hline
		\end{tabular}
}  
	}
\end{table}

%% file: app_experiments_c11.tex
\subsection{Incremental Algorithm for Consistency Checking}\label{subsec:app_experiments_c11}

In this section we present in detail our increlemental algorithm for consistency checking that is evaluated in \cref{subsec:experiments_online}.
Overall, \cref{algo:ra-rf-on-the-fly} has a similar working mechanism as \cref{algo:rf-consistency-rc20}, i.e., it computes the same minimally coherent partial modification order $\mopartial$.
In contrast to that algorithm however, as we need efficient, on-the-fly consistency checks,
this algorithm maintains a per-location order $\hbmo_x$ on write/RMW events that satisfies following invariants:
(i)~$\hbmo_x\subseteq (\hb_x \cup \mopartial_x)^+$ and
(ii)~$(\hb_x\cup (\mopartial_x;\po_x^?))\subseteq \hbmo_x$.

The algorithm is given in pseudocode in \cref{algo:ra-rf-on-the-fly}.
Its main data structure is $\HBMO\colon \E\to \nats_{\geq 0}$,  that is used to store the $\hbmo$ relation.
When a new read or write event is observed, the functions $\rdhandler$ and $\wthandler$ are called, respectively. 
RMW events are handled by first calling $\rdhandler$ followed by $\wthandler$.
The algorithm does not precompute 
the array $\HB$ ($\hb$-timestamps) or the data structures $\set{\view{\WtLst}{t,x}{u}}_{t, x, u}$ (for implementing $\glw(\cdot, \cdot, \cdot)$) and $\set{\view{\RdLst}{t,x}{u}}_{t, x, u}$ (for implementing $\glr(\cdot, \cdot, \cdot)$).
% utilize the helper functions $\getHBTS$ and $\initWtRdLst$. 
Instead, it maintains the same information on-the-fly (Lines \ref{line:ra-otf-hb}, \ref{line:ra-otf-getlw}, \ref{line:ra-otf-trans-closure-2}, \ref{line:ra-otf-write-hb-update}-\ref{line:ra-otf-write-glw}).
%Overall, \cref{algo:ra-rf-on-the-fly} has a similar working mechanism as \cref{algo:rf-consistency-rc20}. 
%It infers a necessary and sufficient set of partial modification orderings, for which
%there are guaranteed total extensions that are compliant with the $\ramm$ model.
%When a new event is observed, \cref{algo:ra-rf-on-the-fly} aims to infer new $\mo$ edges based on \cref{lem:read_coherence} and \cref{lem:RC20-RF-Consistency-Chain-Rule} and then performs certain checks to ensure that consistency is maintained.

In more detail, read  and write events are handled as follows.
\begin{compactitem}

\item 
For a new read event $\event = \rd(x)$ (or $\event = \ud(x)$) with $w_{\mathsf{rf}} = \rfinv(\event)$, in the first part of the procedure (\crefrange{line:ra-otf-hb}{line:trans-closure-end}) the algorithm performs certain checks (Lines \ref{line:ra-otf-first-check-begin},  \ref{line:ra-otf-third-check-begin2}, \ref{line:ra-otf-trans-closure-3}) 
to ensure that extending the execution with $\event$ results in a consistent execution.
These checks are performed by means of an acyclicity check and they ensure 
that the write coherence and atomicity axioms are not violated.

A core component of the algorithm is the $\HBMO$ clocks, which are essential for performing efficient consistency checks incrementally.
Each write event in the execution has an associated $\HBMO$ vector clock.
RMW events in the same chain are represented with a single $\HBMO$ clock.
Intuitively, $\HBMO$ clock of a write event $\wt$ captures the incoming $\mo$ edges to $\wt$ from the writes in the other threads. 

In order to decide whether the new read/RMW event $\rd(x)$ can observe a write/RMW event $\wt(x)$, 
we must determine if there exists another write/RMW event $\wt'(x)$ such that
$(\wt', \rd) \in \hb_x$ and $(\wt', \wt) \in (\hb_x \cup \mopartial_x)^+$,
as this would lead to a consistency violation (\crefrange{line:ra-otf-hb}{line:ra-otf-check-end}).

An important feature of \cref{algo:ra-rf-on-the-fly} is that to optimize for performance, it does not keep the $\HBMO$ clocks transitively closed.
Due to this reason, the algorithm may fail to reject an inconsistent execution in these first checks.
To aleviate this,
in \crefrange{line:trans-closure-init}{line:trans-closure-end}, a transitive closure procedure is conducted and the above condition is checked again (\cref{line:ra-otf-trans-closure-3}).
If this step is passed, then the algorithm updates the state of the execution (\crefrange{line:ra-otf-update-hbmo}{line:ra-otf-last}) and declares that the execution is consistent.
While the algorithm is updating the state of the execution, it propagates the new inferred $\mo$ edges on  $\wt_{\mathsf{rf}}$ into the writes that are $\po$ ordered after $\TC[\wt_{\mathsf{rf}}]$ (\cref{line:ra-otf-po-prop}).

\item
For a new write event $\event = \wt(t, x)$ (or $\event = \ud(t, x)$), the main goal of the procedure is to initialize or update the $\HBMO$ clock of $\event$.
In particular, the algorithm first fetches the write event $\wt$ which corresponds to the last executed write event in thread $t$ on location $x$.
The algorithm then joins $\HBMO[\TC[\event]]$ with $\HBMO[\wt]$ and $\HB[\event]$ (\cref{line:ra-otf-write-rmw-hbmo}) and propagates the new inferred $\mo$ edges into the writes that are $\po$ ordered after $\TC[\event]$.
The latter step is only relevant for the write events that are associated with an $\rf$-chain. 
%Why this step? Saturation would take care of it?
\end{compactitem}

%Why getlastwriteoffline?

\input{algorithms/algo-ra-rf-on-the-fly}

\bparagraph{Scalablity experiments}
We complement our experimental results of \cref{tab:expr-results-c11} with some controlled scalability experiments that better illustrate differences in consistency checking between \cref{algo:ra-rf-on-the-fly} and the algorithm implemented inside \celeventester.

The benchmarks consists of patterns shown in \cref{fig:ra_rf_on_the_fly_scalability}.
They are created to exhibit certain behaviors which exploit contrasting aspects of the algorithms.
These benchmarks provide a better understanding of the algorithm by demonstrating when and why \cref{algo:ra-rf-on-the-fly} ensues a better performance.

%The standard benchmarks found in the literature typically contain complex and convoluted behaviors which prevent from conducting such detailed analysis.    
%The examples provided in \cref{sec:synthetic-bench-rmw-rw} and \cref{sec:synthetic-bench-rmw-wr} harness the features discussed in \cref{sec:efficient-rmw-handling}, and the examples provided in \cref{sec:synthetic-bench-r-partial-satur} and \cref{sec:synthetic-bench-no-w-satur} harness the features discussed in  \cref{sec:no-prop-across-threads} and \cref{sec:no-saturation-write}, respectively. The example discussed in \cref{sec:synthetic-bench-hb-aware} discusses one of the disadvantages of \cref{algo:ra-rf-on-the-fly}.

%\input{app_scalability_tex}

\input{figures/ra_rf_on_the_fly_scalability}

\input{figures/scalability}

%% file: algorithms/algo-ra-rf-on-the-fly.tex
%!TEX root = ../main.tex

\begin{algorithm*}
%\begin{multicols*}{2}
	%\Input{The set of threads $tids$}
	%\Input{Events $\E$, program order $\po$ and reads-from relation $\rf$, }
	%\BlankLine
	%\myfun{\init{tids}}{
	%	\lForEach{$t, t' \in tids$}{            %\label{line:hb-clock-init}
	%		$\Hh_{t}[t'] \gets 0$
	%	}
	%	$\HB \gets []$; $\HBMO \gets []$ \;
	%}

	%\myfun{\wthandler{$\E$, $\rf$}}{
		%\lIf{$(\po \cup \rf)$ is cyclic or $\rf$ violates %weak-atomicity}{
			%    \declare `Inconsistent' %\label{line:rf-rc20-po-rf-acyclic}
			%} 
		%$\HB \gets$ \getHBTS{$\E$, $\po$, $\rf$} % \label{line:rf-rc20-hb-array-init} 
		%\;
		%$(\LRB, \LWB) \gets$ \getLRWB{$\E$, $\po$} %\label{line:rf-rc20-lwb-array-init} 
		%\;
		%$(\TC, \PC) \gets$ \getTCPC{$\E$, $\rf$}  %\label{line:rf-rc20-tc-pc-array-init}
		%\;
		%\lForEach{$x \in \E.locs$}{
			%    $\mo_{\partialtext, x} \gets \emptyset$
			%}
		
		%\myfun{\init{}}{
			%	\lForEach{$t, t' \in tids$}{            %\label{line:hb-clock-init}
				%		$\Hh_{t}[t'] \gets 0$
				%	}
			%	$\HB \gets []$ \;
			%	$\HBMO \gets []$ \;
			%}
				
		%\BlankLine	
		
		\myfun{\rdhandler{$\event, w_{\mathsf{rf}}$}}{
			%{\tcp{$\event$ denotes the current event}} 			
			$\HB[\event] \gets \Hh_{t}[t] + 1$ ; $\PHBMO \gets []$ \label{line:ra-otf-hb} \;
			
			\ForEach{$t' \in tids$}{
				\label{line:ra-otf-first-loop-start}
				\Let $w = $ \glw{($t', x, \Hh_{t}[t']$)} \label{line:ra-otf-getlw} \;
				\If{$w \neq w_{\mathsf{rf}}$} {
					\label{line:ra-otf-check-begin}
					\If {$\TC[w] = \TC[w_{\mathsf{rf}}]$} {
%						\lplabel{line:ra-otf-first-check-begin}
						\If{$\PC[w_{\mathsf{rf}}] < \PC[w]$} {
							\label{line:ra-otf-first-check-begin}
							\declare `Inconsistent'
						}
%						\Else {
%							\Continue \label{line:ra-otf-first-check-end}
%						}
						
					} 
%					\If{$\HBMO[w][w_{\mathsf{rf}}.\tid] \geq \HBMO[w_{\mathsf{rf}}][w_{\mathsf{rf}}.\tid]$} {
%						\label{line:ra-otf-second-check-begin}
%						\declare `Inconsistent'
%						\label{line:ra-otf-second-check-end}
%					}
					%\If{$w_{\mathsf{rf}}.\op = \ud$} {
					%	\label{line:ra-otf-third-check-begin}
						%\Let $w' = w_{\mathsf{rf}}.\chainhead$ \;
					\Else {
						\If {$\HBMO[\TC[w]][\TC[w_{\mathsf{rf}}].\tid] \geq \HB[\TC[w_{\mathsf{rf}}]][\TC[w_{\mathsf{rf}}].\tid]$}  {
							\label{line:ra-otf-third-check-begin2}
							\declare `Inconsistent'
						}
					\label{line:ra-otf-third-check-end}
					%}
					$\PHBMO \gets \PHBMO \mx \HBMO[w]$
					\label{line:ra-otf-check-end}
					 %\label{line:ra-otf-mo-infer} 
				}
				}
			}
		
			\Let $L = tids$; \Let $changedThreads = \emptyset$ \label{line:trans-closure-init} \;
			\While{$L \neq \emptyset$} { 
				\ForEach{$t' \in L$}{
					\If{$\HB[e][t'] < \PHBMO[t']$} {
						\label{line:ra-otf-trans-closure-1}
						\Let $w' = $ \glw{($t', x, \PHBMO[t']$)} \;
						\label{line:ra-otf-trans-closure-2}
						\If {$\HBMO[\TC[w']][\TC[w_{\mathsf{rf}}].\tid] \geq \HBMO[\TC[w_{\mathsf{rf}}]][\TC[w_{\mathsf{rf}}].\tid]$} {
							\label{line:ra-otf-trans-closure-3}
							\declare `Inconsistent'
							\label{line:ra-otf-trans-closure-4}
						}
						$\PHBMO\gets\joinchangedindex{$\PHBMO, \HBMO[w']$}$ \label{line:ra-otf-trans-closure-5}\;
						$changedThreads \gets changedThreads \cup \PHBMO.\mathsf{changedIndexes}$ \label{line:ra-otf-trans-closure-6} %\tcp{$\PHBMO.\mathsf{changedIndexes}$ denotes the indexes that are changed in the last pairwise maximum operation on the $\PHBMO$ clock}  
					}
				}
				
				$L \gets changedThreads$; $changedThreads \gets \emptyset$ \label{line:trans-closure-end}
			}
		
			$\HBMO[\TC[w_{\mathsf{rf}}]] \gets \HBMO[\TC[w_{\mathsf{rf}}]] \mx \PHBMO$ \label{line:ra-otf-update-hbmo} \;
			$\HB[\event] \gets \HB[\event] \mx \HB[w_{\mathsf{rf}}]$; 		$\Hh_{t}[t] \gets \Hh_{t}[t] + 1$ \label{line:ra-otf-update-hb}\;	
			$\poprop(\TC[w_{\mathsf{rf}}])$
			\label{line:ra-otf-po-prop}	
%			\If{$w_{\mathsf{rf}}.\op = \ud$} {
%				\label{line:ra-otf-po-prop-start}
%				$\poprop(\TC[w_{\mathsf{rf}}])$
%				\label{line:ra-otf-po-prop-1}
%			} \Else{ 
%				$\poprop(w_{\mathsf{rf}})$
%				\label{line:ra-otf-po-prop-end}
%			} 
			
			\declare `Consistent' \label{line:ra-otf-last}
		}
	
		\BlankLine
		\myfun{\wthandler{$\event$}}{
			%\tcp{$\event$ denotes the current event}
			\If{$\event.\op \not = \ud$} {
				$\Hh_{t}[t] \gets \Hh_{t}[t] + 1$;
				$\HB[\event] \gets \Hh_{t}[t]$;
				$\HBMO[\event] \gets []$ \label{line:ra-otf-write-hb-update}\;
			}
			\Let $w = $ \glw{($e.\tid, e.\lloc, \HB[\event]-1$)} \label{line:ra-otf-write-glw} \;
			% Last executed event $e'$ s.t. $e'.\op =\wt \lor \ud$, $e'.\tid = t$ and $e'.\lloc = x$ 
			$\HBMO[\TC[\event]] \gets \HBMO[\TC[\event]] \mx \HB[\event] \mx \HBMO[w]$ \label{line:ra-otf-write-rmw-hbmo} \;
			$\poprop(\TC[\event])$ \;
			\label{line:ra-otf-write-rmw-hbmo-po}
%			\If{$\event.\op = \ud$} {
%				$\HBMO[\TC[\event]] \gets \HBMO[\TC[\event]] \mx \HB[\event] \mx \HBMO[w]$ \label{line:ra-otf-write-rmw-hbmo} \;
%				$\poprop(\TC[\event])$
%				\label{line:ra-otf-write-rmw-hbmo-po}
%			}
%			\Else {
%				$\HBMO[\event] \gets \HB[\event] \mx \HBMO[w]$ 
%				\label{line:ra-otf-write-hbmo}
%			}
			\declare `Consistent'
		}
		% \ForEach{$\event \in \E$ in $(\po \cup \rf)$-order}{
			% %\label{line:rf-ra-mo-infer-start}
			%     \Case{$\event = \wt(t, x)$}{
				%         $\Hh_{t}[t] \gets \Hh_{t}[t] + 1$ \;
				%         $\HB[\event] = \Hh_{t}[t]$ \;
				%         $\HBMO[\event] = \HB[\event] \sqcup \HBMO[\getLW{t, x}]$
				%     }
			%     \Case{$\event = \rd(t, x)$}{
				%         \Let $w_{\mathsf{rf}} = \rf(\event)$, $t_{\mathsf{rf}} = w_{\mathsf{rf}}.\tid$ and $c_{\mathsf{rf}} = \HB[w_{\mathsf{rf}}][t_{\mathsf{rf}}]$ \;
				%         \ForEach{$u \in \E.tids$}{
					%             \Let $c_u = (\event.\op = \ud \land u = t)$ ? $\HB[\event][u] - 1$ : $\HB[\event][u]$ \;
					%             \Let $w_u^\wt = \LWB[u][x][c_u]$ and \Let $w_u^\rd = \rf(\LRB[u][x][c_u])$ %\label{line:rf-rc20-latest-writes} 
					%             \;
					%             \For{$w_u \in \set{w_u^\wt, w_u^\rd}$}{
						%                 \If{$(\TC[w_{\mathsf{rf}}] \neq \TC[w_u])$ or $(\PC[w_{\mathsf{rf}}] < \PC[w_u])$}{ %\label{line:rf-rc20-check-chain-condition}
							%                     $\mo_{\partialtext, x} \gets \mo_{\partialtext, x} \cup \set{(w_u, \TC[w_{\mathsf{rf}}])}$ %\label{line:rf-rc20-update-mo}
							%                 }
						%             }
					%         }
				%     }
			% }
		% %\ForEach{$x \in \E.locs$}{
			%    \lIf{$\rf_x \cup \hb_x \cup \mo_{\partialtext, x}$ is cyclic}{\declare `Inconsistent' %\label{line:rf-rc20-cycle}
				%    }
			%}
		\BlankLine
	%\end{multicols*}
		
		\caption{Incremental consistency checking for $\ramm$.}
		\label{algo:ra-rf-on-the-fly}
	\end{algorithm*}

%	\normalsize

%% file: figures/ra_rf_on_the_fly_scalability.tex
\begin{figure}[ht]
\centering
\begin{subfigure}{0.45\textwidth}
\centering
\begin{tikzpicture}[yscale=0.7]
  \node (t11) at (0,0) {$\event_1=\wt(t, x)$};
  \node (t12) at (0,-1.5) {$\event_6=\ud(t, x)$};
  %\node (t13) at (0,-3) {$\ud(t, x)$};
  \node (t13) at (0,-2.85) {$\vdots$};
  \node (t14) at (0,-4.5) {$\event_7=\ud(t, x)$};

  \node (t21) at (2.2,0)  {$\event_2=\rd(t_1, x)$};
  \node (t22) at (2.2,-1.5) {$\event_3=\wt(t_1, x)$};
  
  \node (t31) at (3.3,-0.85) {$\hdots$};
  
  \node (t41) at (4.6,0)  {$\event_4=\rd(t_k, x)$};
  \node (t42) at (4.6,-1.5) {$\event_5=\wt(t_k, x)$};
 
  %\draw[po] (t11) to (t12);
  \draw[po] (t21) to node[left]{$\po$} (t22);
  \draw[po] (t41) to (t42);
%
%  \draw[po] (t21) to (t22);
%
  \draw[rf,bend left=0] (t11) to node[right,pos=0.5]{$\rf$} (t12);
  \draw[rf,bend right=0] (t12) to node[right,pos=0.8]{$\rf$} (0, -2.6);
  \draw[rf,bend right=0] (t11) to node[above,pos=0.8]{$\rf$} (t21);
  \draw[rf,bend left=35] (t11) to node[above,pos=0.8]{$\rf$} (t41);
  \draw[rf,bend right=0] (t13) to node[right,pos=0.8]{$\rf$} (t14);
  %\draw[rf,bend right=0] (0,-4.4) to node[right,pos=0.8]{$\rf$} (t14);
  
\end{tikzpicture}
\caption{
Fixed $\rf$-Chain Read-Write
}
\label{subfig:fixed-rmw-chain-rw}
\end{subfigure}
\hfill
\begin{subfigure}{0.45\textwidth}
	\centering
	\begin{tikzpicture}[yscale=0.7]
		\node (t11) at (0,0) {$\event_1=\wt(t, x)$};
		\node (t12) at (0,-1.5) {$\event_2=\ud(t, x)$};
		%\node (t13) at (0,-3) {$\ud(t, x)$};
		\node (t13) at (0,-2.85) {$\vdots$};
		\node (t14) at (0,-4.5) {$\event_3=\ud(t, x)$};
		
		\node (t21) at (2.2,0)  {$\event_4=\wt(t_1, x)$};
		\node (t22) at (2.2,-1.5) {$\event_5=\rd(t_1, x)$};
		
		\node (t31) at (3.5,-0.85) {$\hdots$};
		
		\node (t41) at (4.8,0)  {$\event_6=\wt(t_k, x)$};
		\node (t42) at (4.8,-1.5) {$\event_7=\rd(t_k, x)$};
		%
		%\draw[po] (t11) to (t12);
		\draw[po] (t21) to node[left]{$\po$} (t22);
		\draw[po] (t41) to (t42);
		%
		%  \draw[po] (t21) to (t22);
		%
		\draw[rf,bend left=0] (t11) to node[right,pos=0.5]{$\rf$} (t12);
		\draw[rf,bend right=0] (t12) to node[right,pos=0.8]{$\rf$} (0, -2.6);
		\draw[rf,bend right=0] (t14) to node[above,pos=0.75]{$\rf$} (t22);
		\draw[rf,bend right=0] (t14) to node[above,pos=0.8]{$\rf$} (t42);
		\draw[rf,bend right=0] (t13) to node[left,pos=0.8]{$\rf$} (t14);
		%\draw[rf,bend right=0] (0,-4.4) to node[right,pos=0.8]{$\rf$} (t14);
	\end{tikzpicture}
	\caption{
		Fixed $\rf$-chain Write-Read.
	}
	\label{subfig:fixed-rmw-chain-wr}
\end{subfigure}
\hfill
\begin{subfigure}{0.55\textwidth}
	\centering
	\begin{tikzpicture}[yscale=0.7]
		\node (t11) at (0,0) {$\event_3=\rd(t, y)$};
		\node (t12) at (0,-1.5) {$\event_4=\wt(t, x)$};
		\node (t13) at (0,-3) {$\event_5=\wt(t, x)$};
		\node (t14) at (0,-4.5) {$\event_6=\wt(t, x)$};
		\node (t15) at (0,-5.1) {$\vdots$};
		
		\node (t21) at (2.2,0)  {$\event_1=\wt(t', x)$};
		\node (t22) at (2.2,-1.5) {$\event_2=\wt(t', y)$};
		
		\node (t31) at (4.4,0)  {$\event_7=\wt(t_1, x)$};
		\node (t32) at (4.4,-1.5) {$\event_8=\rd(t_1, x)$};

		\node (t41) at (5.5,-0.85) {$\hdots$};
		
		\node (t51) at (6.8,0)  {$\event_9=\wt(t_k, x)$};
		\node (t52) at (6.8,-1.5) {$\event_{10}=\rd(t_k, x)$};

		\draw[po] (t11) to node[left]{$\po$} (t12);
		\draw[po] (t12) to (t13);
		\draw[po] (t13) to (t14);
		\draw[po] (t21) to (t22);
		\draw[po] (t31) to (t32);
		\draw[po] (t51) to (t52);
		%\draw[po] (t14) to (t15);
		%
		%  \draw[po] (t21) to (t22);
		%
		%\draw[rf,bend left=0] (t11) to node[right,pos=0.5]{$\rf$} (t12);
		%\draw[rf,bend right=0] (t12) to node[right,pos=0.8]{$\rf$} (0, -2.6);
		%\draw[rf,bend right=0] (t14) to node[above,pos=0.75]{$\rf$} (t22);
		%\draw[rf,bend right=0] (t14) to node[above,pos=0.8]{$\rf$} (t42);
		\draw[rf,bend right=0] (t22) to node[above,pos=0.5]{$\rf$} (t11);
		\draw[rf,bend right=0] (t21) to node[above,pos=0.5]{$\rf$} (t32);
		\draw[rf,bend left=41] (t21) to node[above,pos=0.5]{$\rf$} (t52);
		%\draw[rf,bend right=0] (0,-4.4) to node[right,pos=0.8]{$\rf$} (t14);
	\end{tikzpicture}
	\caption{
		New read events.
	}
	\label{subfig:partial-satur-read}
\end{subfigure}
\hfill
\begin{subfigure}{0.35\textwidth}
	\centering
	\begin{tikzpicture}[yscale=0.7]
		\node (t11) at (0,0) {$\event_1=\wt(t, x)$};
		\node (t12) at (0,-1.5) {$\event_2=\wt(t, y)$};
		\node (t13) at (0,-3) {$\event_3=\wt(t, x)$};
		\node (t14) at (0,-4.5) {$\event_4=\wt(t, x)$};
		\node (t15) at (0,-5.1) {$\vdots$};
		
		\node (t21) at (2.2,0)  {$\event_5=\rd(t', y)$};
		\node (t22) at (2.2,-1.5) {$\event_6=\wt(t', x)$};
		\node (t23) at (2.2,-3) {$\event_7=\wt(t', x)$};
		\node (t24) at (2.2,-3.6) {$\vdots$};

		\draw[po] (t11) to node[left]{$\po$} (t12);
		\draw[po] (t12) to (t13);
		\draw[po] (t13) to (t14);
		\draw[po] (t21) to (t22);
		\draw[po] (t22) to (t23);
		%\draw[po] (t14) to (t15);
		%
		%  \draw[po] (t21) to (t22);
		%
		%\draw[rf,bend left=0] (t11) to node[right,pos=0.5]{$\rf$} (t12);
		%\draw[rf,bend right=0] (t12) to node[right,pos=0.8]{$\rf$} (0, -2.6);
		%\draw[rf,bend right=0] (t14) to node[above,pos=0.75]{$\rf$} (t22);
		%\draw[rf,bend right=0] (t14) to node[above,pos=0.8]{$\rf$} (t42);
		\draw[rf,bend right=0] (t12) to node[above,pos=0.5]{$\rf$} (t21);
		%\draw[rf,bend right=0] (0,-4.4) to node[right,pos=0.8]{$\rf$} (t14);
	\end{tikzpicture}
	\caption{
		New write events.
	}
	\label{subfig:no-write-satur}
\end{subfigure}
\begin{subfigure}{0.45\textwidth}
	\centering
	\begin{tikzpicture}[yscale=0.7]
		\node (t11) at (5.2,0) {$\event_3=\rd(t, x)$};
		\node (t1112d) at (5.2,-1) {$\vdots$};
		\node (t12) at (5.2,-2.5) {$\event_4=\rd(t, x)$};
		\node (t13) at (5.2,-4) {$\event_5=\wt(t, x)$};
		\node (t14) at (5.2,-5.5) {$\event_6=\ud(t, x)$};
		\node (t15) at (5.2,-7) {$\event_7=\ud(t, x)$};
		\node (t16) at (5.2,-7.6) {$\vdots$};
		
		\node (t21) at (0,0)  {$\event_1=\wt(t_1, x)$};
		
		\node (t31) at (1.4,-0.1) {$\hdots$};
		
		\node (t41) at (2.8,0)  {$\event_2=\wt(t_k, x)$};

		\draw[po] (t11) to node[left]{$\po$} (5.2, -0.73);
		\draw[po] (t1112d) to (t12);
		\draw[po] (t12) to (t13);
		%\draw[po] (t13) to (t14);
		%\draw[po] (t14) to (t15);
		%
		%  \draw[po] (t21) to (t22);
		%
		%\draw[rf,bend left=0] (t11) to node[right,pos=0.5]{$\rf$} (t12);
		%\draw[rf,bend right=0] (t12) to node[right,pos=0.8]{$\rf$} (0, -2.6);
		%\draw[rf,bend right=0] (t14) to node[above,pos=0.75]{$\rf$} (t22);
		%\draw[rf,bend right=0] (t14) to node[above,pos=0.8]{$\rf$} (t42);
		\draw[rf,bend left=0] (t21) to node[above,pos=0.5]{$\rf$} (t12);
		\draw[rf,bend left=0] (t41) to node[above,pos=0.5]{$\rf$} (t11);
		\draw[rf,bend left=0] (t13) to node[right,pos=0.5]{$\rf$} (t14);
		\draw[rf,bend left=0] (t14) to node[right,pos=0.5]{$\rf$} (t15);
		%\draw[rf,bend right=0] (0,-4.4) to node[right,pos=0.8]{$\rf$} (t14);
	\end{tikzpicture}
	\caption{
		HB-aware.
	}
	\label{subfig:hb-aware}
\end{subfigure}
\caption{
Illustration of the partial executions used in our scalability experiments in \cref{fig:scalability_results}.
}
\label{fig:ra_rf_on_the_fly_scalability}
\end{figure}

%% file: figures/scalability.tex
%!TEX root = main.tex

\begin{figure}[!ht]
	%\vspace{0.5cm}
	\def\scatterscale{0.27}
	\def\scatterwidth{0.28\textwidth}
	
	\centering
	\begin{subfigure}[b]{\scatterwidth}
		\includegraphics[scale=\scatterscale]{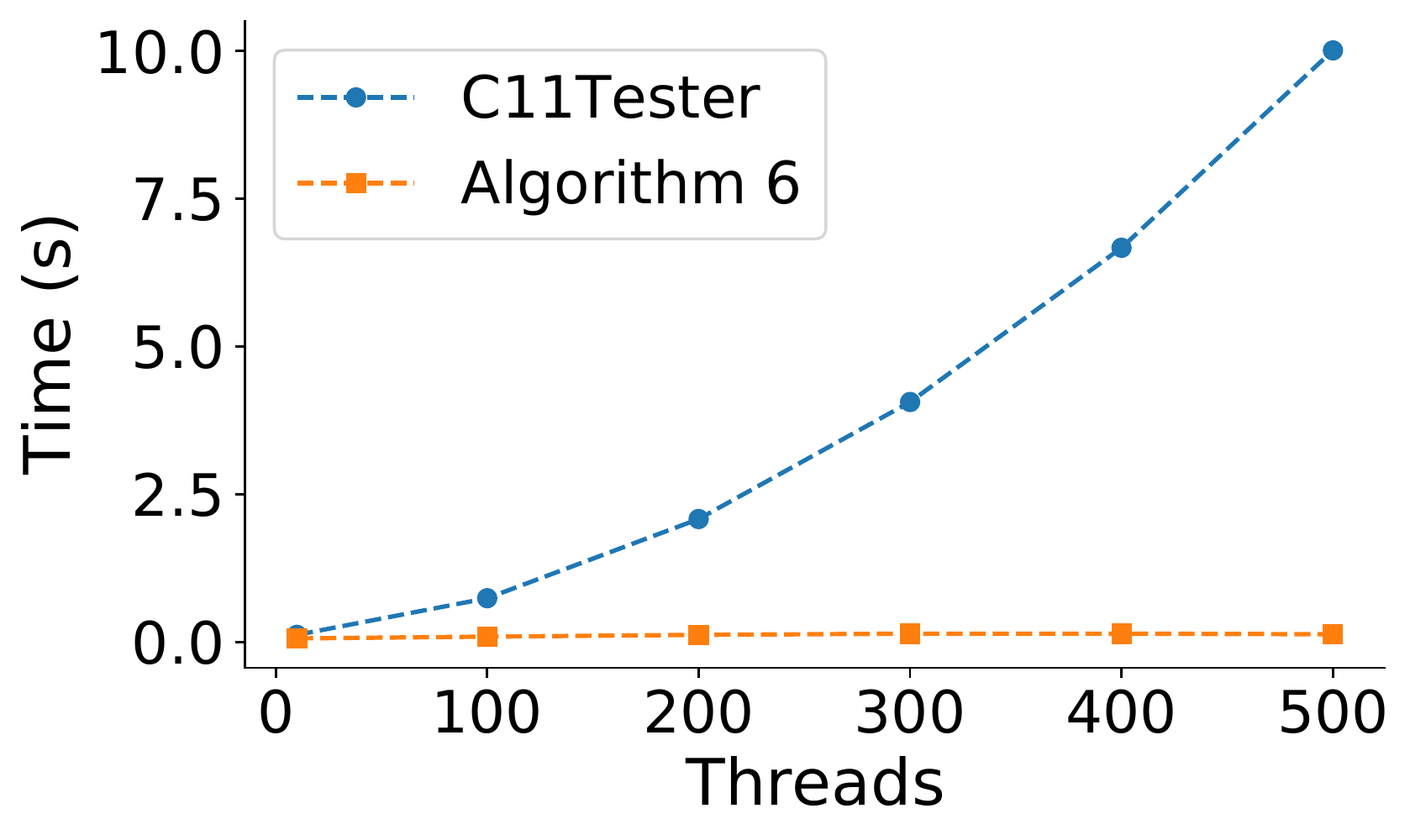}
		\label{subfig:fixed_rmw_chain_rw_result}
		\caption{Fixed $\rf$-Chain Read-Write.}
	\end{subfigure}
	\hfill
	\begin{subfigure}[b]{\scatterwidth}
		\includegraphics[scale=\scatterscale]{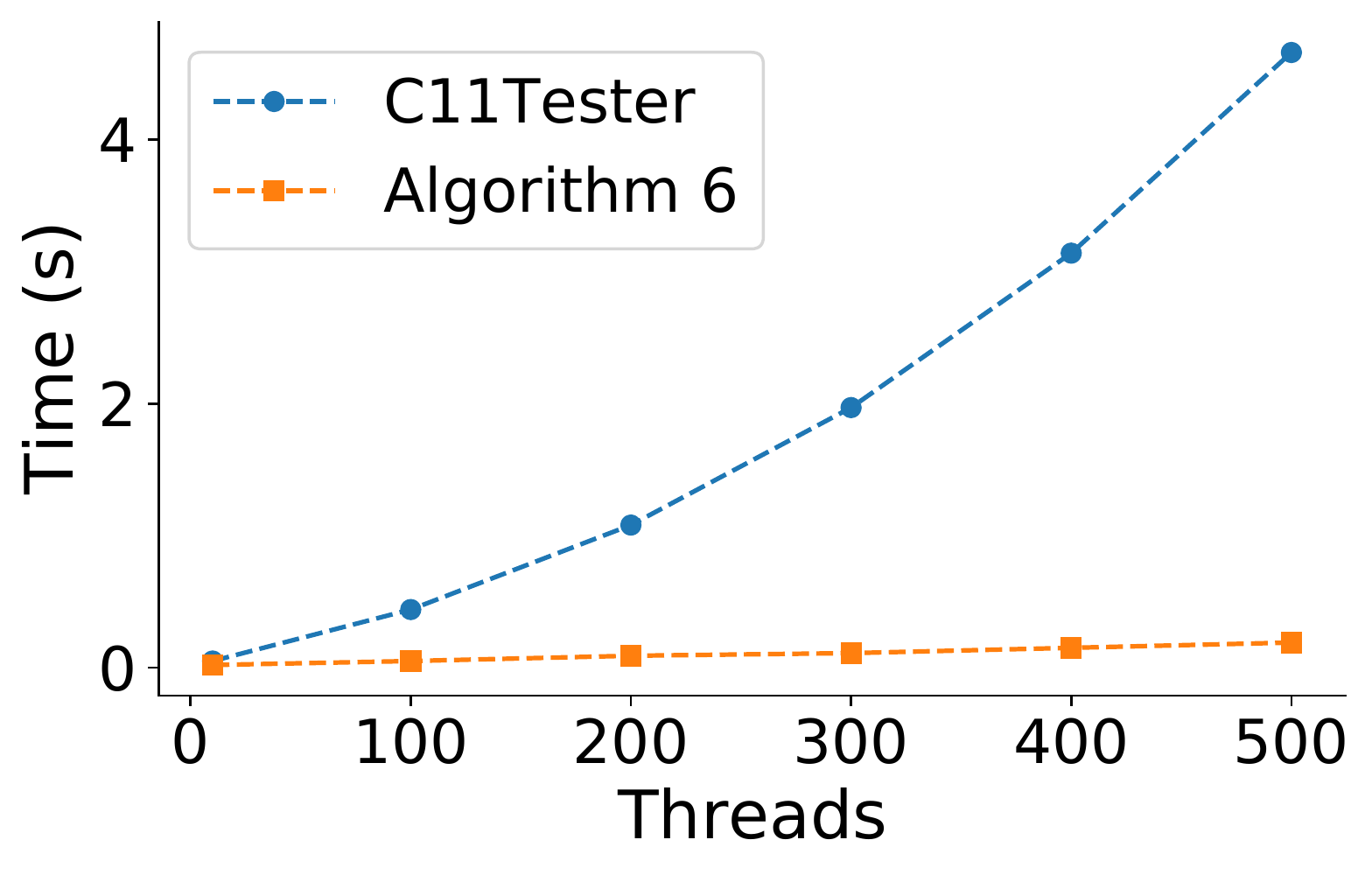}
		\label{subfig:fixed_rmw_chain_wr_result}
		\caption{Fixed $\rf$-Chain Write-Read.}
	\end{subfigure}
	\hfill
	\begin{subfigure}[b]{\scatterwidth}
		\includegraphics[scale=\scatterscale]{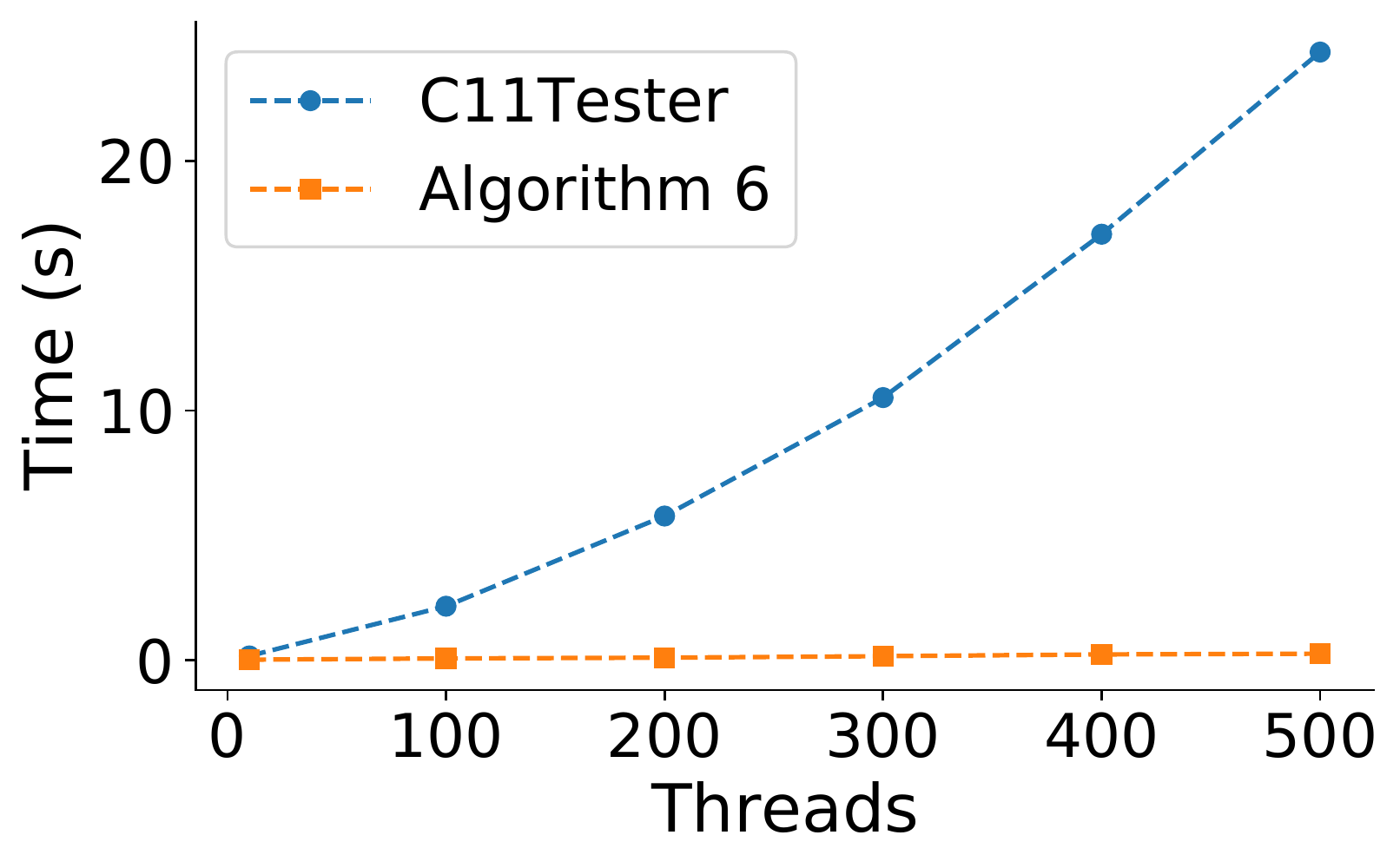}
		\label{subfig:partial_satur_read_result}
		\caption{New read events.}
	\end{subfigure}
\hfill
	\begin{subfigure}[b]{\scatterwidth}
		\includegraphics[scale=\scatterscale]{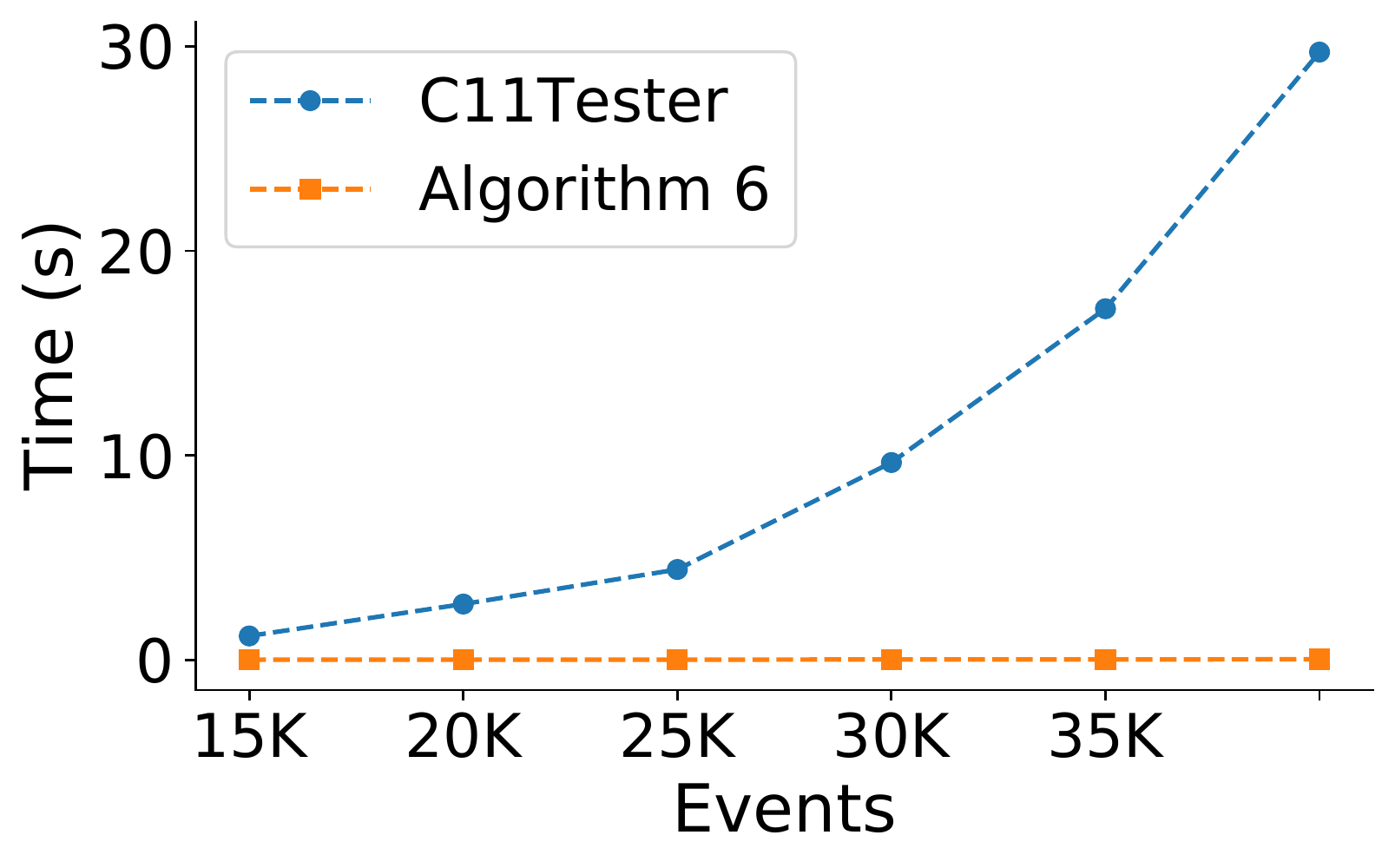}
		\label{subfig:no_w_satur_result}
		\caption{New write events.}
	\end{subfigure}
\hfill
	\begin{subfigure}[b]{\scatterwidth}
		\includegraphics[scale=\scatterscale]{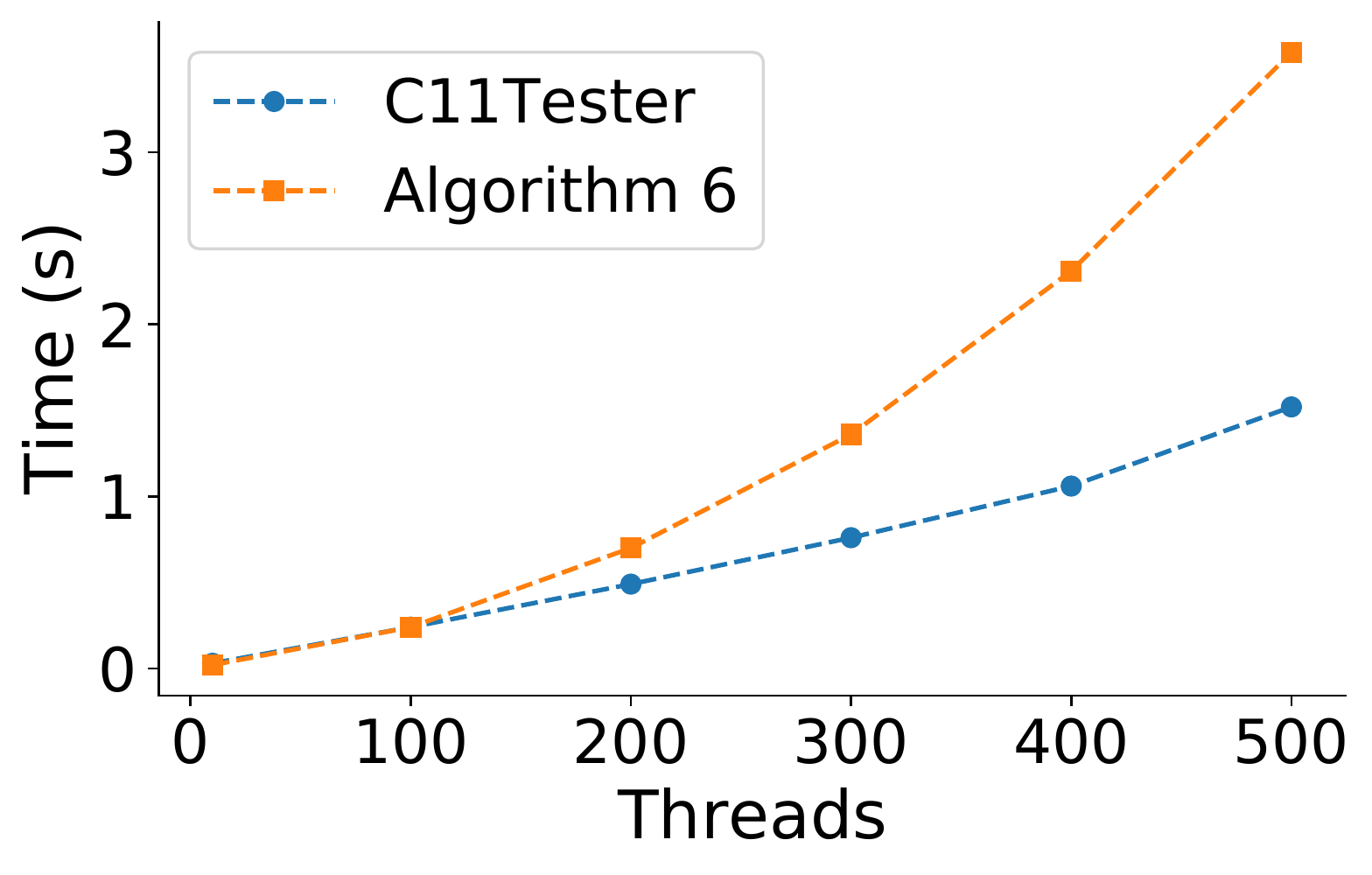}
		\label{subfig:hb_aware_result}
		\caption{HB-aware}
	\end{subfigure}
	\caption{
			Comparison of \celeventester and  \cref{algo:ra-rf-on-the-fly} on five handcrafted benchmarks with increasing number of threads/events.
	}
	\label{fig:scalability_results}
\end{figure}